\makeatletter
\@ifundefined{iflatexml}{\newif\iflatexml
	\latexmlfalse
}{}\makeatother
\iflatexml
	\documentclass[11pt]{article}
\else
	\documentclass[11pt]{scrartcl}
\fi
\input{dominik-arxiv-html.sty}

\usepackage[a4paper, total={16cm, 24cm}]{geometry}
\usepackage{amsmath,amssymb,amsthm}
\usepackage{microtype}
\usepackage{mathtools}
\usepackage{thmtools}
\usepackage{booktabs}
\usepackage{tabularx}
\newcolumntype{L}[1]{>{\raggedright\arraybackslash}p{#1}}
\newcolumntype{C}[1]{>{\centering\arraybackslash}p{#1}}
\usepackage{caption}
\iflatexml
\else
	\usepackage{wrapstuff}
\fi
\usepackage{tikz}
\usepackage{pgfplots}
\usepackage{todonotes}
\usepackage[ruled,vlined, noend]{algorithm2e}
\usepackage{algpseudocode}
\DontPrintSemicolon
\SetCommentSty{textnormal}

\usepackage{fancyhdr}
\setcapindent{0pt}
\setkomafont{captionlabel}{\sffamily\bfseries}

\KOMAoptions{bibliography=totoc}

\iflatexml
	\usepackage[round]{natbib}
\else
	\usepackage[backend=biber,style=authoryear,natbib=true,maxbibnames=99,maxcitenames=2,backref=true]{biblatex}
\fi

\usepackage[hypertexnames=false]{hyperref}
\iflatexml
	\hypersetup{
		colorlinks=true,
		citecolor=green!50!black,
		linkcolor=red!60!black,
		urlcolor=blue!90!black
	}
\else
	\hypersetup{
		pdfencoding=auto,
		psdextra,
		colorlinks=true,
		citecolor=green!50!black,
		linkcolor=red!60!black,
		urlcolor=blue!90!black
	}
\fi
\usepackage[nameinlink]{cleveref}
\Crefname{appsec}{Appendix}{Appendices}
\usepackage{subcaption}
\newtheorem{theorem}{Theorem}[section]
\newtheorem{lemma}[theorem]{Lemma}

\newtheorem{proposition}[theorem]{Proposition}
\newtheorem{fact}[theorem]{Fact}
\theoremstyle{definition}
\newtheorem{definition}[theorem]{Definition}
\newtheorem{example}[theorem]{Example}
\theoremstyle{remark}
\newtheorem{remark}[theorem]{Remark}
\DeclareMathOperator{\cone}{cone}
\newcommand{\R}{\mathbb{R}}
\newcommand{\Z}{\mathbb{Z}}

\definecolor{normalizDirectiveColor}{HTML}{4A5568}
\definecolor{normalizSupplyColor}{HTML}{2B6CB0}
\definecolor{normalizSizeColor}{HTML}{6AB832}
\definecolor{normalizUtilityColor}{HTML}{B7791F}

\newcommand{\normalizEntry}[2]{\ifnum#2=0 \textcolor{#1!28}{#2}\else\textcolor{#1}{\textbf{#2}}\fi}

\title{Core Existence in Approval-Based Committee Elections with up to Seven Voter Types}

\usepackage{authblk}

\author[1]{Patrick Becker}
\author[2]{Matthias Greger}
\author[2]{Dominik Peters}
\affil[1]{Technical University of Munich} 
\affil[2]{CNRS, LAMSADE, Universit\'e Paris Dauphine - PSL}
\date{\vspace{-1cm}}

\begin{document}
\maketitle

\begin{abstract}
	\noindent	
	In an approval-based committee election, the task is to select a committee of up to $k$ candidates from a set of $m$ candidates based on the preferences of $n$ voters, each of whom approves a subset of the candidates. 
    A central open question is whether there always exists a committee in the \emph{core}, a stability notion capturing proportional representation. 
    We prove core non-emptiness for all approval-based committee elections with at most seven voters. 
    The proof is based on affine monoid methods and shows that, for $n\le5$, every fractional committee admits a deterministic rounding to an integral committee that preserves each voter's utility up to floors. This no longer applies for larger $n$. However, for $n \in \{6,7\}$, we show that a Lindahl equilibrium can be adapted and rounded to obtain a core committee.
    For $n \le 5$, we further provide a polynomial-time algorithm for computing a committee in the core. 
    Our arguments work for the weighted voter setting, which implies core existence for instances with up to seven distinct approval sets. 
    We conclude by providing examples where our methods fail for more general models with additive valuations, non-unit candidate costs, or the related Droop core.
    \vspace{-5pt}
\end{abstract}

\setcounter{tocdepth}{1}
\iflatexml
	\tableofcontents
\else
	\makeatletter
	\begingroup
		\renewcommand*{\scr@tso@section@beforeskip}{.9em}
		\tableofcontents
	\endgroup
	\makeatother
\fi

\section{Introduction}

The core \citep{shapley1955markets} originates in cooperative game theory and has since become a central solution concept, as it captures a strong and natural notion of stability against coordinated deviations by groups of agents \citep[see, e.g.,][]{OsRu94a}. 
In the context of committee elections, \citet{ABC+16a} introduced the core as a \emph{proportional representation} criterion, guaranteeing that groups of voters are represented to an extent corresponding to the group size. 
Since then, a central open problem has been whether every approval-based committee election admits a committee in the core \citep{LaSk23a}.

In this setting, the input consists of a voter set $N$, a candidate set $C$, a target committee size $k \ge 0$, and an approval set $A_i \subseteq C$ for each voter $i \in N$.
A committee $W \subseteq C$ with $\lvert W\rvert \le k$ is in \emph{the core} if no non-empty group of voters $S \subseteq N$ can object.
A group $S$ has a \emph{valid objection} if it can form its own subcommittee $T \subseteq C$ of size
\[
\lvert T \rvert  \;\leq\; \frac{\lvert S\rvert}{n} \cdot k,
\] 
which every member of $S$ strictly prefers, i.e., $\lvert A_i \cap T \rvert > \lvert A_i \cap W \rvert$ for all $i \in S$. 
The core is commonly viewed as a strong proportionality requirement. 
It excludes committees that underrepresent large and sufficiently cohesive groups of voters, in the sense that such a group could use its \emph{proportional share} of the budget to deviate to a committee that every member of the group strictly prefers. 
For example, if 40\% of the voters exclusively approve the same sufficiently large set of candidates, then a core committee must allocate about 40\% of the seats to candidates from that set; otherwise, this group could block the committee by choosing its own representative subcommittee. 

While core existence remains open, several partial results have been obtained. 
Exact core committees are known to exist when the committee size is bounded, specifically for $k \le 8$, and also when the candidate set is small, namely for $m\le15$ \citep{peters2025fewseats}. 
Furthermore, every instance is known to admit committees satisfying constant-factor approximations of the core \citep{munagala2022approximate,jiang2020approximately,PeSk20a}. For specific subdomains, the core has been shown to be non-empty \citep{pierczynski2022core,BGP+24a}.
At the same time, an extensive literature has identified voting rules that guarantee relaxations of the core, most notably \emph{extended justified representation} (EJR) \citep{ABC+16a}, which restricts attention to objections $T$ that are unanimously approved by $S$.

In this paper, we prove that the core is non-empty whenever $n \le 7$. Our approach supports weighted voters, thus our results can equivalently be stated in terms of the number of \emph{voter types}, i.e., the number of different approval sets in the profile; hence we prove non-emptiness when there are at most seven voter types.
We also show that, when $n \le 5$, a core committee can be found in polynomial time, and we can moreover ensure that this committee is also Pareto-optimal. 
We believe that a small-$n$ existence result is particularly attractive because instances with only a few voters occur in practice: many collective decisions are made by small groups, such as juries, hiring committees, or boards of directors.
Moreover, our model also admits an alternative interpretation in which voters represent weighted decision criteria, where a decision maker needs to select a set of $k$ objects that should if possible possess certain binary features, with the decision maker assigning importance weights to the features.
Our results show that with at most seven criteria or features, there always exists a stable selection that balances the criteria proportionally to their weights.

\begin{wrapstuff}[r,width=0.33\textwidth,type=figure,lines=15,top=3]
	\centering
	\begin{tikzpicture}[yscale=0.45,xscale=1,voter/.style={anchor=south}]
		\foreach \i in {1,...,3}
			\node[voter] at (\i-0.5, -1.3) {$\i$};
		
		\draw[fill=orange!40] (0, 0) rectangle (2, 1);
		\draw[fill=orange!40] (0, 1) rectangle (2, 2);
		\draw[fill=orange!40] (0, 2) rectangle (2, 3);
		\draw[fill=orange!40] (0, 3) rectangle (2, 4);
		\draw[fill=orange!40] (0, 4) rectangle (2, 5);
		\draw[fill=orange!40] (0, 5) rectangle (2, 6);
		\draw[fill=orange!40] (0, 6) rectangle (2, 7);
		\draw[fill=orange!40] (0, 7) rectangle (2, 8);
		\draw[fill=orange!40] (0, 8) rectangle (2, 9);
		\draw[fill=orange!40] (0, 9) rectangle (2, 10);

		\draw[fill=violet!40] (0,10) rectangle (1, 11);  
		\draw[fill=violet!40] (1,10) rectangle (2, 11); 

		\draw[fill=teal!40] (2, 0) rectangle (3, 1);
		\draw[fill=teal!40] (2, 1) rectangle (3, 2);
		\draw[fill=teal!40] (2, 2) rectangle (3, 3);
		\draw[fill=teal!40] (2, 3) rectangle (3, 4);
		\draw[fill=teal!40] (2, 4) rectangle (3, 5);
		\draw[fill=teal!40] (2, 5) rectangle (3, 6);
		\draw[fill=teal!40] (2, 6) rectangle (3, 7);
		\draw[fill=teal!40] (2, 7) rectangle (3, 8);

		\node at ( 1, 0.5) {$c_{1}$};
		\node at ( 1, 1.5) {$c_{2}$};
		\node at ( 1, 2.5) {$c_{3}$};
		\node at ( 1, 3.5) {$c_{4}$};
		\node at ( 1, 4.5) {$c_{5}$};
		\node at ( 1, 5.5) {$c_{6}$};
		\node at ( 1, 6.5) {$c_{7}$};
		\node at ( 1, 7.5) {$c_{8}$};
		\node at ( 1, 8.5) {$c_{9}$};
		\node at ( 1, 9.5) {$c_{10}$};
		
		\node at ( 0.5, 10.5) {$c_{11}$};
		\node at ( 1.5, 10.5) {$c_{12}$};
								
		\node at ( 2.5, 0.5) {$c_{13}$};
		\node at ( 2.5, 1.5) {$c_{14}$};
		\node at ( 2.5, 2.5) {$c_{15}$};
		\node at ( 2.5, 3.5) {$c_{16}$};
		\node at ( 2.5, 4.5) {$c_{17}$};
		\node at ( 2.5, 5.5) {$c_{18}$};
		\node at ( 2.5, 6.5) {$c_{19}$};
		\node at ( 2.5, 7.5) {$c_{20}$};
	\end{tikzpicture}
	\caption{Example where PAV violates the core for three voters.}
\end{wrapstuff}
\label{fig:pav-core-counterexample}
Perhaps the most natural strategy for proving existence is to show that some voting rule selects committees in the core. This is what \citet{peters2025fewseats} did to show existence for $k \le 8$: he proved using linear programming that the classic Proportional Approval Voting (PAV) rule \citep{Thie95b} always selects a core committee in this case. However, this method does not work for the few-voter case since PAV fails the core even for $n = 3$ voters. 
An example is shown in Figure~\ref{fig:pav-core-counterexample}, where each voter approves the candidates appearing above the voter label. For example, voter 1 approves $\{c_1, \dots, c_{11}\}$. Let $k=18$. 
The PAV-winning committee is $W = \{c_1, \dots, c_{10}\} \cup \{c_{13}, \dots, c_{20}\}$, yielding a utility vector of $u=(10, 10, 8)$.
This committee is not in the core, since the coalition $S = \{1,2\}$ can choose the subcommittee $T =  \{c_1, \dots, c_{12}\}$ of size $\frac{\lvert S\rvert}{n}\cdot k = \frac{2}{3}\cdot 18 = 12$ that both voters strictly prefer (they approve $11$ candidates instead of $10$).
This shows that PAV violates the core even for only three voters.

Instead, our approach builds on a striking connection between \emph{fractional} and \emph{integral} committees for $n \le 5$ that might be of independent interest. A fractional committee can contain candidates fractionally, with the constraint that these fractions sum to at most $k$.
We show in \Cref{thm:monoidnormal} that, when there are up to five voter types, every \emph{integer} utility vector that can be achieved by a fractional committee can also be achieved by an integral committee. 
Combined with the existence of fractional core committees (\Cref{thm:fraccoreexists}), this allows us to round down the utility vector of a fractional core outcome and implement it by an integral committee, thereby obtaining a committee in the core (see \Cref{thm:roundfracutils}).

These results are obtained via the theory of \emph{affine monoids}.
To explain the approach, it is convenient to group candidates by their respective ``approval patterns'':
Each candidate is identified with the subset $R\subseteq N$ of voters approving it. Thus, there are at most $2^n$ \emph{candidate types}, denoted by $\mathcal R$. 
An instance can be specified by a \emph{supply vector} $c=(c_R)_{R \in \mathcal R}$, where $c_R$ denotes the number of available candidates of type $R$, together with a committee size $k$. 
An integral committee can be written as a vector $x=(x_R)_{R \in \mathcal R} \in \Z_{\geq0}^{\lvert \mathcal R \rvert}$ satisfying $0\le x_R\le c_R$ for all $R \in \mathcal R$ and $\sum_R x_R\le k$. 
Now, a utility vector $u\in \Z^n$ is \emph{integral-committee-feasible} if there exists an integral committee $x$ with $\sum_{R\ni i}x_R\ge u_i$ for all $i\in N$. In other words, each voter $i \in N$ approves at least $u_i$ candidates in $x$.
To obtain our theorem, we consider the set of instance--utility vector pairs:
\[
	M:=\{(c,k,u): u \text{ is integral-committee-feasible in }(c,k)\}.
\]
This is a subset of $\Z^{\lvert \mathcal R \rvert + 1 + n}$ that is closed under addition. Hence, it forms an \emph{affine monoid}. 
If we now consider $\operatorname{cone}(M)$, the real cone generated by $M$, we obtain the set of pairs of a fractional instance and a utility vector that is \emph{fractional-committee-feasible}.
Our main technical theorem (\Cref{thm:monoidnormal}) shows that for $n\le 5$,
\[
	M=\operatorname{cone}(M)\cap \Z^{\lvert \mathcal R \rvert + 1 + n}.
\]
This means that every integer utility vector that is fractional-committee-feasible is also integral-committee-feasible. 
In the language of affine monoids, $M$ is \emph{integrally closed} (and also \emph{normal} in our case) for $n \le 5$. Our normality results for $M$ are proven in \Cref{sec:coreexistence} and are based on the idea of decomposing instance--utility pairs $(c,k,u)$ into a ``base part'' $(c^0,k^0,u^0)$ and a ``residual part'' $(c^r,k^r,u^r)$:
\begin{align*}
	(c,k,u)=(c^0,k^0,u^0)+(c^r,k^r,u^r).
\end{align*}
This decomposition will be chosen such that $u^0$ is integral-committee-feasible in $(c^0,k^0)$, and subject to that condition, $(c^r,k^r,u^r)$ cannot be further reduced with respect to certain criteria. Using non-reducibility, we prove that the set of possible residuals $(c^r,k^r,u^r)$ is finite and well-characterized (see \Cref{thm:decomposition}). This enables us to check integral-committee feasibility for all of them and thus also for $(c,k,u)$.
Instead of using our decomposition argument, normality of $M$ can also be automatically verified using the \texttt{Normaliz} program \citep{Normaliz}.%
\footnote{\texttt{Normaliz} is also used in social choice theory to compute the frequency of voting paradoxes under the impartial anonymous culture distribution via Ehrhart polynomials \citep[see, e.g.,][]{BGS15a}.}
From normality and the existence of fractional committees in the fractional core, core non-emptiness for $n \le 5$ follows.

\begin{wrapstuff}[r,width=0.47\textwidth,type=figure,lines=9,top=2]
	\centering
		\begin{tikzpicture}[
			scale=0.75,
			voter/.style={circle,draw,fill=white}
			]
			\foreach \i/\x/\y in {1/0/1.4,2/-1.2/-0.7,3/1.2/-0.7,4/4.5/1.4,5/3.3/-0.7,6/5.7/-0.7}{
				\node[voter] (v\i) at (\x,\y) {$\i$};
			}
			
			\draw[very thick,black] (v1) -- (v2) -- (v3) -- (v1);
			\draw[very thick,black] (v4) -- (v5) -- (v6) -- (v4);
		\end{tikzpicture}
		\caption{Violation of normality of $M$ for $n = 6$ and $k = 3$, where the utility vector $u = (1,1,1,1,1,1)$ cannot be achieved.}
		\label{fig:two-triangles-intro}
\end{wrapstuff}

Unfortunately, normality of $M$ fails for six voters. This is illustrated by the two-triangle instance in \Cref{fig:two-triangles-intro}, where the voters are vertices approving their incident edges that represent the candidate types. The committee size is $k = 3$.
The integer utility vector $u=(1,1,1,1,1,1)$ is fractionally-committee-feasible as it is realized by the fractional committee $x$ with $x_R=\frac{1}{2}$ for all edges $R$. However, $u$ is not integral-committee-feasible: every committee of size at most $3$ has total utility at most $6$, while realizing $u$ would require total utility at least $6$. Thus, such a committee would have size $3$, and all six voters would need to receive utility exactly $1$. But any committee of three candidates contains two candidates from one of the two triangles, so some voter receives utility $2$, a contradiction.

\begin{wrapstuff}[r,width=0.47\textwidth,type=figure,lines=9,top=6]
	\centering
	\begin{tikzpicture}[
		scale=0.75,
		voter/.style={circle,draw,fill=white}
		]
		\foreach \i/\x/\y in {1/0/1.4,2/-1.2/-0.7,3/1.2/-0.7,4/4.5/1.4,5/3.3/-0.7,6/5.7/-0.7}{
			\node[voter] (v\i) at (\x,\y) {$\i$};
		}
		
		\draw[thick,black] (v1) -- (v2);
		\draw[thick,black] (v3) -- (v1);
		\draw[thick,black] (v4) -- (v6);
		\draw[ultra thick,green!80!black] (v2) -- (v3);
		\draw[ultra thick,green!80!black] (v4) -- (v5) -- (v6);
	\end{tikzpicture}
	\caption{The utility vector $u' = (0,1,1,1,1,1)$ is achieved by the thick green committee.}
	\label{fig:two-triangles-shorted}
\end{wrapstuff}

However, in the two-triangle example, an integral core committee does exist for a specific reason that applies to \emph{all} counterexamples for $n = 6$, which will allow us to deduce core existence for $n = 6$. In particular, the fractional committee $x$ with $x_R=\frac{1}{2}$ in fact forms a \emph{Lindahl equilibrium} \citep{Fole70a,KrPe25b}, which means it is in the fractional core and even in (a version of) the strict core that requires that only one of the deviators improves strictly, with other deviators improving weakly. Now, take an arbitrary voter (let's say the first one) and short its utility by 1, giving us a target utility vector $u' = (0,1,1,1,1,1)$. Note that $u'$ is integral-committee-feasible (\Cref{fig:two-triangles-shorted}) by a committee that we argue must be in the core. First, note that the voter who we shorted cannot deviate on its own, because a single voter cannot propose a candidate because $\frac{k}{n} = \frac{3}{6} < 1$. In addition, the voter also cannot deviate together with other voters; if such a deviation existed, it would induce utilities $(1,2,\dots,2)$ for the deviators, which would witness a strict core violation for the original Lindahl equilibrium, a contradiction.

Using a computer enumeration, we find that for $n = 6$, any counterexample to normality induces only one of 50 different residual parts $(c^r,k^r,u^r)$ (up to relabeling voters), with the two-triangle instance being one of those 50. Each of those 50 ``minimal holes'' has the same structure that allows us to short one voter and thereby obtain an integral core committee (namely: in the fractional Lindahl equilibrium every voter receives a utility that is an integer; a single voter cannot deviate on its own; after shorting the utility vector it becomes integral-committee-feasible). In \Cref{sec:coreexistencen=6}, we then establish that if the residual part satisfies these properties, the original instance has an integral core committee, thereby proving core existence for $n = 6$.

For $n = 7$, our computer enumeration yields 54,985 holes (up to relabeling voters), of which at most 21,818 can actually occur after filtering out holes that could not occur when starting at a Lindahl equilibrium. Of these, 21,520 holes can be dealt with using the same argument we just discussed for $n = 6$. This leaves 298 other holes that require a new argument. The part that breaks is that on those holes, the voters obtain a non-integral utility in the Lindahl equilibrium. However, we are able to design a sufficient condition that still allows us to deduce core existence on those 298 holes. Intuitively, the condition says that while utilities are not integral ($u_i(x) > \lfloor u_i(x) \rfloor$), the cumulative overshoot $\sum_{i \in N} \bigl(u_i(x) - \lfloor u_i(x) \rfloor\bigr)$ is not too big and, more importantly, is not too valuable at the Lindahl equilibrium prices; see \Cref{thm:reduceLindahlutilsfor7}. Using this sufficient condition, we establish that the core exists for $n = 7$ (\Cref{sec:core-existence-7}). The same technique could potentially extend to larger $n$, but there is a counterexample for $n = 13$ where new types of holes appear (\Cref{fig:pg23}). Enumerating all minimal holes quickly becomes intractable for larger $n$, though.

In \Cref{sec:corecomputation}, we derive a polynomial-time algorithm for computing a committee in the core for $n \le 5$. Our algorithm additionally ensures that the output committee is Pareto optimal. In the first step, it computes an approximate fractional committee in the fractional core by approximately optimizing a convex program recently proposed by \citet{KrPe25b}. We show that the obtained utility vector can be rounded down and later greedily increased so that the resulting integer utility vector is fractional-committee-feasible and Pareto optimal among such vectors. In a second step, the algorithm finds an integral committee that achieves these utilities, either via an integer linear program with a bounded number of variables or via a series of linear programs. Both steps heavily rely on normality of the respective monoids for $n \le 5$.

In \Cref{sec:counterex}, we show the limits of our techniques for related models. In detail, the analogous utility-rounding step fails (even for $n = 3$) when defining the core based on the more demanding Droop quota, as well as when moving to more general additive utilities or to a participatory budgeting setting without the unit cost assumption.

\section{Preliminaries}
\label{sec:prelims}

Let $N$ be a finite set of $n$ voters. The model of approval-based committee elections is typically defined using a finite set $C$ of candidates over which the voters express approvals. However, for our purposes it is convenient to group candidates by \emph{type}. 
The type of a candidate is the set of voters approving it. 
Thus, the set of all possible non-empty candidate types is $\mathcal R \coloneqq 2^N \setminus \{\emptyset\}$.
An \emph{instance} of the approval-based committee problem is given by a triple $(c,k,b)$, where 
\begin{itemize}
	\item $c \in \Z_{\ge 0}^{\lvert\mathcal R \rvert}$ is the \emph{supply vector}, specifying how many candidates of each type are available,
	\item $k \in \Z_{\ge 0}$ is the \emph{committee size}, and
	\item $b \in \mathbb{Q}_{> 0}^n$ with $\sum_{i \in N} b_i = k$ is the vector of \emph{voter budgets}.
\end{itemize}
Voter budgets function as weights. The standard committee election model typically assumes equal budgets $b_i = \frac{k}{n}$ for all $i \in N$. We allow arbitrary positive rational budgets, since this will enable us to merge voters with identical approval sets into a single voter by summing their budgets; in the standard unweighted model, this gives a type budget equal to $\frac{k}{n}$ times the number of voters of that type. We write $b(S) = \sum_{i \in S}b_i$ for the overall budget of a coalition $S \subseteq N$. For every non-empty voter set $N$, the assumptions $b_i>0$ for all $i\in N$ and $b(N)=k$ imply $k>0$.
Also, given a supply vector $c$, we write $\mathcal R_c \coloneqq \{R \in \mathcal R : c_R > 0\}$ for the set of candidate types that are present in the instance.
For a candidate type set $\mathcal R_c$ and a voter $i\in N$, we write
\[
d_i \coloneq \bigl|\{R\in\mathcal R_c : i\in R\}\bigr|
\]
for the number of available candidate types containing voter $i$.

We define the set of \emph{integral committees affordable for coalition $S \subseteq N$} as
\begin{equation*}
	\mathcal{W}_S \coloneqq \left\{x \in \Z_{\ge 0}^{\lvert \mathcal R \rvert} \colon 
	\quad x_R \le c_R \text{ for all } R \in \mathcal R ,
	\quad \sum_{R \in \mathcal R} x_R \le  b(S)  \right\}.
\end{equation*}
Here, $x_R$ is the number of selected candidates of type $R$. 
We write $\mathcal W = \mathcal W_N$ for the set of all \emph{integral committees} respecting the committee size constraint $\sum_{R \in \mathcal R} x_R \le b(N) = k$.
The \emph{utility} that voter $i$ obtains from a committee $x$ is given by $u_i(x) \coloneqq \sum_{R \ni i} x_R$.
Observe that this is equivalent to approval utilities.
A vector $u = (u_i)_{i \in N} \in \R^n$ is called \emph{integral-committee-feasible} if there exists an integral committee $x \in \mathcal W$ such that
\begin{equation*}\label{equation:integral_feasibility_condition}
	\sum_{R \ni i} x_R \ge u_i \quad \text{for all } i \in N.
\end{equation*}
In other words, $u$ is integral-committee-feasible if there is an integral committee that gives every voter \emph{at least} the prescribed utility level $u_i$. 
Since these vectors are used only as lower bounds, we allow negative utility levels for notational convenience.

We also need a fractional relaxation of committees. 
We define the set of \emph{fractional committees affordable for coalition $S \subseteq N$} as
\begin{equation*}
	\mathcal{P}_S \coloneqq \left\{x \in \R_{\ge 0}^{\lvert \mathcal R \rvert} \colon 
	\quad x_R \le c_R \text{ for all } R \in \mathcal R ,
	\quad \sum_{R \in \mathcal R} x_R \le b(S) \right\}.
\end{equation*}
Again, $\mathcal{P} \coloneqq \mathcal{P}_N$. Thus, $x_R$ may now be interpreted as selecting a fractional amount of candidates of type $R$. 
Voters' utilities extend linearly to fractional committees, i.e., still $u_i(x) \coloneqq \sum_{R \ni i} x_R$.
Analogous to the integral case, a vector $u = (u_i)_{i \in N} \in \R^n$ is \emph{fractional-committee-feasible} if there exists a fractional committee $x \in \mathcal P$ that guarantees every voter at least utility $u_i$.

\subsection{The core}

We now introduce the central solution concept studied in this paper.

\begin{definition}[Core]
	An integral committee $x \in \mathcal W$ is in the \emph{core} of an instance $(c,k,b)$ if there does not exist a non-empty \emph{blocking coalition} $S \subseteq N$ and a committee $y \in \mathcal W_S$ such that
	\[
		u_i(y) > u_i(x) \qquad \text{for all } i \in S.
	\]
\end{definition}

Thus, a committee is in the core if no coalition can deviate to a committee that it can afford and that makes all of its members strictly better off.%
\footnote{This definition corresponds to what is usually called the ``weak core'' in cooperative game theory, but ``core'' is the standard terminology for the committee election setting. A stronger fairness notion is the \emph{strict core}, where all coalition members need to weakly improve and at least one coalition member needs to strictly improve their utility \citep[Footnote 6]{aziz2018proportionalrepresentationapprovalbasedcommittee}. However, the strict core can be empty even in tiny examples, for example by taking 4 voters, $k = 2$, and $\mathcal R_c = \{\{1\},\{2\},\{3\},\{4\}\}$.}

Alongside the core, we will also consider its fractional analog.

\begin{definition}[Fractional core]
	A fractional committee $x \in \mathcal P$ is in the \emph{fractional core} of an instance $(c,k,b)$ if there does not exist a non-empty \emph{blocking coalition} $S \subseteq N$ and fractional committee $y \in \mathcal P_S$ such that
	\[
		u_i(y) > u_i(x) \qquad \text{for all } i \in S.
	\]
\end{definition}

In contrast to the integral core, the fractional core is known to always exist in approval-based committee elections, thanks to the concept of \emph{Lindahl equilibrium} \citep{Fole70a}, a virtual market equilibrium for public goods. \citet{KrPe25b} considered a model with linear utilities (which includes our 0/1 utility case) and proposed a convex program whose optimum solution is a special Lindahl equilibrium and which therefore lies in the fractional core. Notably, their computed Lindahl equilibrium does not always exhaust the entire budget.
An alternative way to derive the same existence result is from the work of \citet{Fole70a} who proved existence of a Lindahl equilibrium via fixed-point theorems in a very general model. However, his result applies to public-goods economies with continuous, concave, and \emph{strictly} monotone utilities.
Our approval-based setting allows zero valuations and hard supply caps on candidate types, neither of which is directly covered by Foley's theorem.
A way around this is to approximate the original utilities by strictly monotone concave utilities \citep{FGM16a,munagala2022approximate} for the case $\sum_{R\in\mathcal R} c_R > k$; otherwise one can just set $x_R=c_R$ for all $R \in \mathcal{R}$. By passing to a limit, one can deduce that Lindahl equilibria (and thus also the fractional core) exist in our model. We explain the details of this argument in \Cref{appendix:omitted_proofs}.

In particular, that proof technique yields a Lindahl equilibrium that exhausts the entire budget, which is important for our reasoning in \Cref{sec:coreexistencen=6}.

\begin{restatable}{theorem}{fracCoreExists}\label{thm:fraccoreexists}
    For every instance $(c,k,b)$, the fractional core is non-empty. If $\sum_{R\in\mathcal R} c_R > k$, a Lindahl equilibrium $(x,p)$ with $x \in \mathcal{P}$ and $\sum_{R \in \mathcal{R}}x_R=k$ exists.
\end{restatable}

\subsection{From the fractional core to the core}\label{sec:fromfraccoretocore}

Suppose we have a fractional committee $x$ that lies in the fractional core. In this section, we prove 
that if an integral committee induces a utility vector of at least $\lfloor u(x)\rfloor$, then it is in the core.

\begin{theorem}\label{thm:roundfracutils}
	Let $(c,k,b)$ be an arbitrary instance and $x$ a committee in its fractional core. If the vector of rounded-down utilities $\lfloor u(x)\rfloor=(\lfloor u_i(x)\rfloor)_{i \in N}$ is integral-committee-feasible, then its (integral) witness $x' \in \mathcal W$ lies in the core.
\end{theorem}

\begin{proof}
	Let $u = (u_i(x))_{i \in N}$ be the utility vector yielded by a fractional core committee $x$, and let $u' = (\lfloor u_i(x) \rfloor)_{i \in N}$ be the integer utility vector obtained by rounding down those utilities. 
	Note that $u$ is fractional-committee-feasible by $x$, and $u' \le u$ is integral-committee-feasible via some integral committee $x' \in \mathcal W$ with $\sum_{R \ni i} x'_R \ge u'_i$ for all $i \in N$.
	
	We claim that $x'$ is in the core. 
	If not, there exists a non-empty $S \subseteq N$ and $y \in \mathcal{W}_S$ such that 
	\begin{equation*}
		u_i(y) = \sum_{R \ni i} y_R > \sum_{R \ni i} x'_R \ge u'_i \quad \text{for all } i\in S.
	\end{equation*}
	Since $\sum_{R \ni i} y_R$ and $u'_i$ are integers, also $\sum_{R \ni i} y_R \ge u'_i + 1$ for all $i \in S$. 
	Moreover, as $u'_i = \lfloor u_i(x) \rfloor$, we have $\sum_{R \ni i} y_R \ge \lfloor u_i(x) \rfloor + 1 > u_i(x)$ for all $i \in S$. 
	Since $y \in \mathcal W_S \subseteq \mathcal P_S$, this means that $y$ blocks $x$ with respect to the fractional core, a contradiction.
\end{proof}

Later, we show that for $n \le 5$, $\lfloor u(x)\rfloor$ is integral-committee-feasible for \emph{any} fractional committee $x$ (\Cref{thm:monoidnormal}) and thus the core is non-empty.

\subsection{From voters to voter types} \label{sec:votertypes}

The way we defined our model involving candidate types and their supplies differs from the setup as it is usually given. To be able to phrase our existence result in the language of approval-based committee elections, we reformulate the model: 
Let $\hat N$ be a finite set of voters (we use a different symbol from $N$ because these voters will be \emph{unweighted}), and let $\hat C$ be a finite set of candidates. We are given a committee size $k$ and an \emph{approval profile} $(A_i)_{i \in \hat N}$ where $A_i \subseteq \hat C$ is the set of approved candidates of voter $i \in \hat N$. A \emph{committee} is a subset $W \subseteq \hat C$ with $|W| \le k$. It is in the \emph{core for the approval profile} $(A_i)_{i \in \hat N}$ if there is no non-empty coalition $S \subseteq \hat N$ and objection $T \subseteq \hat C$ such that $|T| \le \frac{|S|}{|\hat N|}\cdot k$ and $|A_i \cap T| > |A_i \cap W|$ for all $i \in S$. 

Let us say that two voters $i,j \in \hat N$ have the same \emph{voter type} in an approval profile if $A_i = A_j$. 
The following theorem shows that core-existence results for the compressed instances $(c,k,b)$ immediately extend to approval-based committee elections with the same number of voter types.

\begin{restatable}{theorem}{voterTypes}\label{thm:votertypes}
	If every instance $(c,k,b)$ with at most $t$ voters has a non-empty core, then every approval-based committee election with at most $t$ voter types has a non-empty core.
\end{restatable}

The proof is routine and can be found in \Cref{appendix:omitted_proofs}.

\section{From fractional- to integral-committee feasibility via affine monoids} \label{sec:fraccoretocore}

Our aim is to leverage the existence of fractional core committees to produce integral committees in the core. To achieve this, we need to understand which utility vectors can be (approximately) achieved by integral committees, compared to fractional committees. For this, we use affine monoids.

\subsection{Affine monoids}

Intuitively, an affine monoid is a collection of all integer vectors that can be obtained by repeatedly adding integer vectors taken from a fixed set of \emph{generators}. In other words, it is the set of all non-negative integer combinations of the generators. Equipping this set with the usual (coordinatewise) addition, we obtain a monoid.\footnote{Our exposition in this section broadly follows the terminology of the book by \citet{bruns2009polytopes}.}

\begin{definition}
	An \emph{affine monoid} is a finitely generated submonoid of $\Z^d$: a set
	$M \subseteq \Z^d$ of the form
	\[
	M = \Bigl\{\sum_{j=1}^{s} \alpha_j g_j : \alpha_j \in \Z_{\ge 0}\Bigr\}
	\]
	for some finite set of generators $g_1, \dots, g_s \in \Z^d$.
\end{definition}

\begin{figure}[t]
	\centering
	\begin{subfigure}[t]{0.31\textwidth}
		\centering
		\begin{tikzpicture}[scale=0.75,every path/.style={>=stealth}]
			\draw[->,gray!60] (0,0) -- (4.6,0);
			\draw[->,gray!60] (0,0) -- (0,4.6);
			\foreach \x in {0,...,4}
			\foreach \y in {0,...,4}
			\fill[gray!45] (\x,\y) circle (1.25pt);
			\foreach \x/\y in {0/0,1/0,2/0,3/0,1/2,2/2,3/2,2/4,3/4}
			\draw[->,blue!75!black,line width=0.75pt,shorten >=1.7pt] (\x,\y) -- ++(1,0);
			\foreach \x/\y in {0/0,1/0,2/0,3/0,1/2,2/2,3/2}
			\draw[->,green!55!black,line width=0.75pt,shorten >=1.7pt] (\x,\y) -- ++(1,2);
			\foreach \x/\y in {0/0,1/0,1/2,2/0,2/2,2/4,3/0,3/2,3/4,4/0,4/2,4/4}
			\fill[black] (\x,\y) circle (2.0pt);
			\node[blue!75!black] at (0.6,0.26) {$g_1$};
			\node[green!55!black] at (0.33,1.44) {$g_2$};
		\end{tikzpicture}
		\caption{The monoid $M$ consists of all points that can be formed by non-negative integer combinations of the generators.}
		\label{fig:monoid-example}
	\end{subfigure}
	\hfill
	\begin{subfigure}[t]{0.31\textwidth}
		\centering
		\begin{tikzpicture}[scale=0.75,every path/.style={>=stealth}]
			\fill[blue!8] (0,0) -- (4.35,0) -- (4.35,4.35) -- (2.18,4.35) -- cycle;
			\draw[->,gray!60] (0,0) -- (4.6,0);
			\draw[->,gray!60] (0,0) -- (0,4.6);
			\foreach \x in {0,...,4}
			\foreach \y in {0,...,4}
			\fill[gray!45] (\x,\y) circle (1.25pt);
			\draw[thick,blue!75!black] (0,0) -- (4.35,0);
			\draw[->,blue!75!black,line width=0.75pt,shorten >=1.7pt] (0,0) -- (1,0);
			\draw[thick,green!55!black] (0,0) -- (2.18,4.35);
			\draw[->,green!55!black,line width=0.75pt,shorten >=1.7pt] (0,0) -- (1,2);
			\draw[dashed,blue!35] (4.35,0) -- (4.35,4.35);
			\draw[dashed,blue!35] (2.18,4.35) -- (4.35,4.35);
			\node[blue!70!black] at (2.9,1.45) {$\cone(M)$};
			\foreach \x/\y in {0/0,1/0,1/2,2/0,2/2,2/4,3/0,3/2,3/4,4/0,4/2,4/4}
			\fill[black] (\x,\y) circle (1.8pt);
			\foreach \x/\y in {1/1,2/1,2/3,3/1,3/3,4/1,4/3}
			\draw[red!80!black] (\x,\y) circle (2.3pt);
		\end{tikzpicture}
		\caption{The cone of $M$ is formed by the real combinations of the generators; some integer points are in the cone but not in $M$.}
		\label{fig:cone-example}
	\end{subfigure}
	\hfill
	\begin{subfigure}[t]{0.31\textwidth}
		\centering
		\begin{tikzpicture}[scale=0.75,every path/.style={>=stealth}]
			\fill[blue!8] (0,0) -- (4.35,0) -- (4.35,4.35) -- (2.18,4.35) -- cycle;
			\draw[->,gray!60] (0,0) -- (4.6,0);
			\draw[->,gray!60] (0,0) -- (0,4.6);
			\foreach \x in {0,...,4}
			\foreach \y in {0,...,4}
			\fill[gray!45] (\x,\y) circle (1.25pt);
			\draw[thick,blue!75!black] (0,0) -- (4.35,0);
			\draw[->,blue!75!black,line width=0.75pt,shorten >=1.7pt] (0,0) -- (1,0);
			\draw[thick,green!55!black] (0,0) -- (2.18,4.35);
			\draw[->,green!55!black,line width=0.75pt,shorten >=1.7pt] (0,0) -- (1,2);
			\draw[dashed,blue!35] (4.35,0) -- (4.35,4.35);
			\draw[dashed,blue!35] (2.18,4.35) -- (4.35,4.35);
			\foreach \x/\y in {0/0,1/0,1/2,2/0,2/2,3/0,3/2}
			\draw[->,violet!85!black,line width=0.75pt,shorten >=1.7pt] (\x,\y) -- ++(1,1);
			\node[violet!85!black] at (3.72,2.22) {$g_3$};
			\foreach \x/\y in {0/0,1/0,1/1,1/2,2/0,2/1,2/2,2/3,2/4,3/0,3/1,3/2,3/3,3/4,4/0,4/1,4/2,4/3,4/4}
			\fill[black] (\x,\y) circle (1.8pt);
		\end{tikzpicture}
		\caption{After adding a new generator $g_3=(1,1)$, all integer points in the cone are in the new monoid $M'$; thus $M'$ is integrally closed.}
		\label{fig:normal-monoid-example}
	\end{subfigure}
	\caption{The affine monoid $M$ generated by $g_1=(1,0)$ and $g_2=(1,2)$ is not integrally closed in $\Z^2$, but the monoid $M'$ with additional generator $g_3=(1,1)$ is integrally closed in $\Z^2$.}
	\label{fig:monoid-cone-normality}
\end{figure}

For example, we can consider dimension $d = 2$ and the generators $g_1 = (1,0)$ and $g_2 = (1,2)$. The monoid generated by them consists of all vectors $(a + b, 2b)$ where $a,b \in \Z_{\ge 0}$. This set is shown in \Cref{fig:monoid-example}.

Given an affine monoid $M$ generated by $g_1, \dots, g_s$, its (real) \emph{cone} is 
\[
\cone(M) = \Bigl\{\sum_{j=1}^{s} \alpha_j g_j : \alpha_j \in \R_{\ge 0}\Bigr\},
\]
the set of all non-negative \emph{real} combinations of the generators. The cone of our example monoid is shown in \Cref{fig:cone-example}. Obviously $M \subseteq \cone(M)$, and $\cone(M)$ contains vectors that are not in $M$. We are interested in whether $\cone(M)$ happens to contain \emph{integer} vectors that are not in $M$. For our example monoid, this is the case: the vector $(1,1)$  is not in the monoid (its second coordinate being odd), but it is in the cone because $(1, 1) = \frac12 (1,0) + \frac12 (1,2)$. Thus, $(1,1)$ is a ``hole'' in the monoid. A monoid is called integrally closed if no such holes exist.

\begin{definition}
	An affine monoid $M$ is \emph{integrally closed} in $\Z^d$ if $M = \cone(M) \cap \Z^d$, i.e., every integer point in the real cone of $M$ is already in $M$ itself.
\end{definition}

If we change our example monoid by adding a third generator $g_3 = (1,1)$, then the resulting monoid is integrally closed, see \Cref{fig:normal-monoid-example}.
Finally, an affine monoid generated by $g_1, \dots, g_s \in \Z^d$ has an associated ``group of differences'' $\operatorname{gp}(M) = \bigl\{\sum_{j=1}^{s} \alpha_j g_j : \alpha_j \in \Z \bigr\}$, which is the set of (not necessarily non-negative) integer combinations of the generators.
\begin{definition}
	An affine monoid $M$ is \emph{normal} if $M = \cone(M) \cap \operatorname{gp}(M)$.
\end{definition}
For the monoids we consider in the following sections, we always have $\operatorname{gp}(M) = \Z^d$ (as the set of generators will contain all unit vectors or their negations), and so we will use ``normal'' as simply a synonym for ``integrally closed''.

\subsection{The committee election monoid}\label{subsec:committeeelectionmonoid}

We now define the affine monoid that we use for our core existence result. This monoid consists of all combinations of instances and their utility vectors. Specifically, it will consist of triples $(c,k,u)$ where $(c,k)$ is an instance of the approval-based committee problem (with voter budgets ignored\footnote{Note that the structure of the monoid depends only on the supply vector $c$ and the total budget $k$, not on the distribution of the budget among the voters.}) and $u$ is integral-committee-feasible for the instance $(c, k)$, and we will show that this set of ``instance-utility pairs'' forms an affine monoid within $\Z^{|\mathcal{R}|+1+n}$. We next give the definition of the monoid in terms of its generators, and explain what these generators mean afterwards.
\begin{definition}
	\label{def:committee-monoid}
	For a given $n$ and the associated collection of candidate types $\mathcal R = 2^N \setminus \{\emptyset\}$, the \emph{committee election monoid} $M_{\mathcal{R}}$ is the affine monoid in $\Z^{|\mathcal{R}|+1+n}$ generated by
	\begin{equation}\label{eq:generators}
		\begin{aligned}
			X_R &= (e_R,\; 1,\; a_R)  && \text{for each $R \in \mathcal{R}$} \: &&\text{(selecting one candidate of type $R$)}, \\
			Z_R &= (e_R,\; 0,\; 0)   && \text{for each $R \in \mathcal{R}$}  &&\text{(unused supply of type $R$)}, \\
			T   &= (0,\; 1,\; 0)      && &&\text{(unused seats)}, \\
			S_i &= (0,\; 0,\; {-e_i}) && \text{for each $i \in N$} &&\text{(slack in voter $i$'s utility coverage)}.
		\end{aligned}
	\end{equation}
	Here, $e_R$ is the $R$-th unit vector $(0,\dots,1,\dots,0)$ in $\Z^{|\mathcal{R}|}$, $a_R$ is the
	indicator of $R$ in~$\Z^{n}$, and $e_i$ is the $i$-th unit vector in $\Z^{n}$.
\end{definition}

Intuitively, the generators can be used to ``build'' an instance $(c,k)$ and the integer utility vectors that can be achieved by integral committees of this instance. The generator $X_R$ corresponds to adding $1$ unit of supply of candidate type $R$ and adding it to the committee; hence this candidate takes up $1$ seat in the committee and increases the utility of each voter in $R$ by $1$. The generator $Z_R$ corresponds to adding to the supply of $R$ without putting the respective candidate type into the committee. The generator $T$ allows for committees that do not use all $k$ seats. And finally the generator $S_i$ encodes that a feasible utility vector is just a lower bound on a voter's actual utility, so if we have a feasible utility vector and decrease some of its coordinates, it remains feasible. Based on these interpretations, we can check that the committee election monoid indeed consists of instance-utility pairs, as explained at the beginning of the section.

\begin{proposition}
	\label{prop:interpret-generators}
	We have $(c,k,u) \in M_{\mathcal{R}}$ if and only if $(c,k)$ is an instance and $u$ is an integral-committee-feasible utility vector in $(c,k)$. Moreover, $(c,k,u) \in \cone(M_{\mathcal{R}}) \cap \Z^{|\mathcal{R}|+1+n}$ if and only if $(c,k)$ is an instance and $u$ is a fractional-committee-feasible integer utility vector in $(c,k)$. 
\end{proposition}
\begin{proof}
	Suppose that $(c,k,u) \in M_{\mathcal{R}}$. By definition of $M_{\mathcal{R}}$, there exist non-negative integers $x_R, z_R, t, s_i$ such that
	\begin{equation*}
		(c, k, u) = \sum_{R \in \mathcal{R}} x_R X_R + \sum_{R \in \mathcal{R}} z_R Z_R + t\, T + \sum_{i \in N} s_i S_i.
	\end{equation*}
	Then $(c,k)$ is a non-negative integer vector and hence obviously an instance. Now consider the vector $x = (x_R)_{R \in \mathcal{R}}$; in fact, this is an integral committee in the instance $(c,k)$ because $x_R \le x_R + z_R = c_R$ (it satisfies the supply constraints) and $\sum_{R \in \mathcal{R}} x_R \le \sum_{R \in \mathcal{R}} x_R + t = k$ (it satisfies the budget constraint). Further, $x$ witnesses that $u$ is integral-committee-feasible since for each $i \in N$, we have $u_i = \sum_{R \ni i} x_R - s_i$ and hence $u_i \le \sum_{R \ni i} x_R$, as required.
	
	Conversely, suppose $(c,k)$ is an instance and $u$ is an integral-committee-feasible utility vector, witnessed by some integral committee $x$. For each $R \in \mathcal{R}$, define $z_R = c_R - x_R$; define $t = k - \sum_{R \in \mathcal{R}} x_R$; and for each $i \in N$, define $s_i = \sum_{R \ni i} x_R - u_i$. Then, because $x$ is an integral committee, all $z_R$ as well as $t$ are non-negative integers, and because $x$ witnesses feasibility of $u$, the $s_i$ are also all non-negative integers. In addition, we have
	\begin{equation*}
		(c, k, u) = \sum_{R \in \mathcal{R}} x_R X_R + \sum_{R \in \mathcal{R}} z_R Z_R + t\, T + \sum_{i \in N} s_i S_i,
	\end{equation*}
	and thus $(c,k,u) \in M_{\mathcal{R}}$, as required.
	
	The second part of the statement can be shown analogously. 
\end{proof}

Hence, normality of $M_{\mathcal{R}}$ for a fixed $n$ would imply that for a given instance $(c,k)$, every fractional-committee-feasible integer utility vector $u$ is also integral-committee-feasible. \Cref{thm:roundfracutils} would then prove core existence for all instances with $n$ voters.

\subsection{Decomposing instances}\label{sec:decomposition}

Our main goal will be to show that if an integral utility vector $u$ is fractional-committee-feasible in an instance $(c,k)$, then \emph{(i)} for $n \le 5$, $u$ is also integral-committee-feasible and \emph{(ii)} for $n \in \{6,7\}$, a Lindahl equilibrium allocation $x$ either has a floored utility vector $u=\lfloor u(x)\rfloor$ which is integral-committee-feasible or satisfies certain other conditions that allow us to construct a core committee from it. In both cases, if $(c,k,u) \in M_{\mathcal{R}}$ then $u$ is already integral-committee-feasible by \Cref{prop:interpret-generators}. Thus, we will analyze instance-utility pairs $(c,k,u)$ that are not in $M_{\mathcal{R}}$.

We call an instance-utility pair $(c,k,u)$ a \emph{hole} if $u$ is fractional-committee-feasible but not integral-committee-feasible in $(c,k)$.
To analyze holes, we will find it useful to decompose $(c,k,u)$ into a \emph{base part} $(c^0,k^0,u^0)$ consisting of generators (see \ref{eq:generators}) and a \emph{residual part} $(c^r,k^r,u^r)$, i.e.,
\[
	(c,k,u)=(c^0,k^0,u^0)+(c^r,k^r,u^r).
\]
Further, for fixed $n$, we will show that the set of possible residual parts $(c^r,k^r,u^r)$ is \emph{finite}. Later on, we can then enumerate all possible residual parts and check the conditions we need for core existence on each relevant one.

\begin{wrapstuff}[r,width=0.39\textwidth,type=figure,lines=10]
	\centering
	\begin{tikzpicture}[scale=0.8,every path/.style={>=stealth}]
		\fill[blue!8] (0,0) -- (4.35,0) -- (4.35,4.35) -- (2.18,4.35) -- cycle;
		\draw[->,gray!60] (0,0) -- (4.6,0);
		\draw[->,gray!60] (0,0) -- (0,4.6);
		\foreach \x in {0,...,4}
		\foreach \y in {0,...,4}
		\fill[gray!45] (\x,\y) circle (1.25pt);
		\draw[thick,blue!75!black] (0,0) -- (4.35,0);
		\draw[->,blue!75!black,line width=0.75pt,shorten >=1.7pt] (0,0) -- (1,0);
		\draw[thick,green!55!black] (0,0) -- (2.18,4.35);
		\draw[->,green!55!black,line width=0.75pt,shorten >=1.7pt] (0,0) -- (1,2);
		\foreach \x/\y in {0/0,1/0,1/2,2/0,2/2,2/4,3/0,3/2,3/4,4/0,4/2,4/4}
		\fill[black] (\x,\y) circle (1.8pt);
		\foreach \x/\y in {1/1,2/1,2/3,3/1,3/3,4/1,4/3}
		\draw[red!80!black] (\x,\y) circle (2.3pt);
		\draw[red!80!black, semithick] (2,3) circle (2.2pt);
		\draw[red!80!black, semithick] (1,1) circle (2.2pt);
		\draw[->,orange!90!black,line width=1.0pt,shorten >=3pt,shorten <=3pt] (2,3) -- (1,1)
		node[right=0.78pt, pos=0.7] {$-g_2$};
		\node[anchor=south west, red!80!black, node font=\footnotesize] at (1.81,3.03) {$(2,3)$};
		\node[anchor=north west, red!80!black, node font=\footnotesize] at (0.55,0.97) {$(1,1)$};
	\end{tikzpicture}
	\caption{Decomposing an element.}
	\label{fig:counterexample-reduction}
\end{wrapstuff}

For intuition, consider again the monoid from \Cref{fig:monoid-example}, generated by $g_1=(1,0)$ and $g_2=(1,2)$.
The integer point $(2,3)$ lies in the cone generated by $g_1$ and $g_2$ since
$(2,3)=\tfrac12 g_1+\tfrac32 g_2$, but it is not part of the monoid: there are no $\alpha_1,\alpha_2\in\mathbb Z_{\ge0}$ with $\alpha_1 g_1+\alpha_2 g_2=(2,3)$.
Thus $(2,3)$ witnesses a violation of being integrally closed. It can be written as $(2,3)=(1,2)+(1,1)$ where $(1,2)$ constitutes one of the generators as the base part and $(1,1)$ corresponds to the residual part.

Our reduction lemmas play the analogous role for the committee election monoid: they show that in any decomposition $(c,k,u)=(c^0,k^0,u^0)+(c^r,k^r,u^r)$ where the residual part is still ``too large'', it can be further reduced by subtracting suitable generators and adding them to the base part.
Consequently, it suffices to analyze instances that correspond to residual parts. These turn out to be much easier to understand since they live ``close to the origin'', similar to $(1,1)$ in the example.

The following theorem summarizes the decomposition properties established in the subsequent lemmas and used to prove core existence.

\begin{theorem}[Decomposition theorem]\label{thm:decomposition}
	Let $u\in\Z_{\ge0}^n$ be fractional-committee-feasible but not integral-committee-feasible in $(c,k)$, and let $x$ be a fractional witness with $\sum_{R\in\mathcal R}x_R=k$. Then there exists a decomposition
	\[
		(c,k,u)=(c^0,k^0,u^0)+(c^r,k^r,u^r)
	\]
	such that $(c^r,k^r,u^r)$ is fractional-committee-feasible but not integral-committee-feasible and the following conditions hold:
	\begin{enumerate}
		\renewcommand{\theenumi}{(D\arabic{enumi})}
		\renewcommand{\labelenumi}{\theenumi}
		\item\label{cond:decomp-base} $(c^0,k^0,u^0)$ is generated using only the generators $X_R$, $Z_R$, and $T$;
		\item\label{cond:decomp-feasibility} $u(x)-u^0$ is fractional-committee-feasible in $(c^r,k^r)$.
	\end{enumerate}
	Unless $u_i^r=0$ for some voter $i\in N$, the decomposition can additionally be chosen to satisfy:
	\begin{enumerate}
		\renewcommand{\theenumi}{(D\arabic{enumi})}
		\renewcommand{\labelenumi}{\theenumi}
		\setcounter{enumi}{2}
		\item\label{cond:decomp-supply} $\mathcal R_c^r$ is an antichain (no type in $\mathcal R_c^r$ is a strict subset of another) and $c_R^r=1$ for every $R\in\mathcal R_c^r$;
		\item\label{cond:decomp-support} for every $R \in \mathcal R_c^r$ we have $x_R\notin\mathbb Z$;
		\item\label{cond:decomp-witness} there exists a fractional witness $x^r$ for $u(x)-u^0$ in $(c^r,k^r)$ such that
		\[
		\sum_{R\in\mathcal R_c^r}x_R^r=k^r
		\quad\text{and}\quad
		0<x_R^r<1 \quad\text{for every }R\in\mathcal R_c^r;
		\]
		\item\label{cond:decomp-bounds} $k^r\le|\mathcal R_c^r|-2\le n-2$ and $1\le u_i^r\le d_i^r-1$ for every $i \in N$.
	\end{enumerate}
\end{theorem}

The rest of this subsection proves the decomposition theorem.
We start with $(c^r,k^r,u^r)=(c,k,u)$ and $(c^0,k^0,u^0)=(0,0,0)$. By ``moving'' a generator from the residual part to the base part, we mean subtracting it from $(c^r,k^r,u^r)$ and adding it to $(c^0,k^0,u^0)$. We only make such a reduction when $(c^r,k^r)$ remains an instance, $u^r$ remains non-negative, and $u(x)-u^0$ remains fractional-committee-feasible in $(c^r,k^r)$. Note that during all of these moves, the residual remains outside $M_{\mathcal R}$: otherwise, since the moved generator belongs to $M_{\mathcal R}$, the residual before the move would already belong to $M_{\mathcal R}$.
With this introduction, we now state the first reduction lemma.

\begin{lemma}\label{lem:minimal-reduction01}
	The decomposition can be chosen so that \ref{cond:decomp-base} and \ref{cond:decomp-feasibility} hold and, if no coordinate of $u^r$ is zero, \ref{cond:decomp-supply}--\ref{cond:decomp-witness} also hold. Moreover, starting from a decomposition satisfying \ref{cond:decomp-base}, \ref{cond:decomp-feasibility}, \ref{cond:decomp-support}, and \ref{cond:decomp-witness}, it can be further reduced until either a coordinate of $u^r$ becomes zero or \ref{cond:decomp-supply} also holds, while preserving the four initial conditions.
\end{lemma}

\begin{proof}
	Start with $x^r=x$. If some coordinate of $u^r$ is zero, the initial decomposition already satisfies \ref{cond:decomp-base} and \ref{cond:decomp-feasibility}, so we stop. Otherwise, while $x_R^r\ge1$, move $X_R$ from the residual part to the base part and replace $x_R^r$ by $x_R^r-1$. While $c_R^r-x_R^r\ge1$, move $Z_R$ to the base part. These operations preserve \ref{cond:decomp-base} and \ref{cond:decomp-feasibility}. If an $X_R$-reduction produces a zero coordinate of $u^r$, we stop. After completing the initial reductions, $c_R^r=1$ exactly for the originally non-integral coordinates $x_R$, and the remaining fractional committee $x^r$ satisfies $\sum_{R\in\mathcal R_c^r}x_R^r=k^r$ and $0<x_R^r<1$ for every $R\in\mathcal R_c^r$. This proves \ref{cond:decomp-base}, \ref{cond:decomp-feasibility}, \ref{cond:decomp-support}, and \ref{cond:decomp-witness}.
	
	It remains to obtain the antichain property; the same argument also proves the second statement. First, if $c_R^r>1$, move copies of $Z_R$ to the base part until $c_R^r=1$. Now suppose that $R',R\in\mathcal R_c^r$ with $R'\subsetneq R$. Move mass in $x^r$ from $R'$ to $R$ until $x^r_{R'}=0$ or $x^r_R=1$. The total mass is unchanged, and every voter weakly gains utility because $R'\subsetneq R$. If $x^r_{R'}=0$, move $Z_{R'}$ to the base part. If $x^r_R=1$, move $X_R$ to the base part, subtract its incidence vector from $u^r$, and replace $x_R^r$ by $0$. Thus, if both equalities hold, both generators are moved. The updated $x^r$ still witnesses fractional-committee feasibility of $u(x)-u^0$. If some coordinate of $u^r$ becomes zero, we stop. Otherwise, all the stated properties are preserved. Each such reduction removes an available type, so repeating the argument yields an antichain and proves \ref{cond:decomp-supply}.
\end{proof}

The next lemma bounds the number of available candidate types.

\begin{lemma}\label{lem:Rcle5}
	Starting from the decomposition in \Cref{lem:minimal-reduction01}, we can additionally ensure that $\lvert\mathcal R_c^r\rvert\le n$, unless a coordinate of $u^r$ becomes zero. Any of conditions \ref{cond:decomp-base}--\ref{cond:decomp-support} that hold initially are preserved throughout.
\end{lemma}

\begin{proof}
	Assume that $\lvert\mathcal{R}^r_c\rvert>n$.
	Let $x^r$ be the witness from \ref{cond:decomp-witness}. We have
	\[
		c^r_R=1 \quad\text{and}\quad 0<x_R^r<1 \qquad\text{for all } R\in\mathcal{R}^r_c.
	\]
	For each $R\in\mathcal{R}^r_c$, let $a_R\in\{0,1\}^n$ denote the incidence vector of $R$ (see also (\ref{eq:generators})). 
	Since $\lvert\mathcal{R}^r_c\rvert>n$, the vectors $(a_R)_{R\in\mathcal{R}^r_c}$ are linearly dependent in $\R^n$.
	Hence, there exists a non-zero vector
	\[
		\gamma=(\gamma_R)_{R\in\mathcal{R}^r_c}\in\R^{\lvert\mathcal{R}^r_c \rvert}
	\]
	such that $\sum_{R\in\mathcal{R}^r_c}\gamma_R a_R = 0$.
	Replacing $\gamma$ by $-\gamma$ if necessary, we may assume that $\sum_{R\in\mathcal{R}^r_c}\gamma_R \le 0$.
	For $\mu\in\R$, define a perturbed vector $(x^r)^\mu$ by
	\[
		(x^r)^\mu_R :=
		\begin{cases}
			x_R^r+\mu\gamma_R, & \text{if } R\in\mathcal{R}^r_c,\\
			0, & \text{if } R\notin\mathcal{R}^r_c.
		\end{cases}
	\]
	We claim that for a suitable $\mu>0$, the vector $(x^r)^\mu$ is again a fractional witness for $u(x)-u^0$, but satisfies $(x^r)^\mu_R\in\{0,1\}$ for some $R\in\mathcal{R}^r_c$.
	
	First, the utility vector induced by $(x^r)^\mu$ is unchanged. Indeed,
	\[
		\sum_{R\in\mathcal{R}^r_c}(x^r)^\mu_R a_R
		= \sum_{R\in\mathcal{R}^r_c}x_R^r a_R + \mu\sum_{R\in\mathcal{R}^r_c}\gamma_R a_R
		= \sum_{R\in\mathcal{R}^r_c}x_R^r a_R.
	\]
	Hence, every voter receives exactly the same utility under $(x^r)^\mu$ as under $x^r$.
	
	Second, the committee size does not increase. Since $\sum_{R\in\mathcal{R}^r_c}\gamma_R\le 0$, we have for every $\mu\ge 0$,
	\[
		\sum_{R\in\mathcal{R}}(x^r)^\mu_R
		= \sum_{R\in\mathcal{R}^r_c}x_R^r + \mu\sum_{R\in\mathcal{R}^r_c}\gamma_R
		\le \sum_{R\in\mathcal{R}^r_c}x_R^r
		\le k^r.
	\]
	
	It remains to preserve the box constraints. Because $0<x_R^r<1$ for all $R\in\mathcal{R}^r_c$ and $\gamma\neq 0$, the quantity
	\[
	\mu^* :=
	\min\!\left(
	\left\{\frac{1-x_R^r}{\gamma_R}: R \in \mathcal R^r_c,\ \gamma_R>0\right\}
	\cup
	\left\{\frac{x_R^r}{-\gamma_R}: R \in \mathcal R^r_c,\ \gamma_R<0\right\}
	\right)
	\]
	is well-defined and strictly positive. By construction, for every $\mu\in[0,\mu^*]$, we have $0\le (x^r)^\mu_R\le 1$ for all $R\in\mathcal{R}^r_c$ and for $\mu=\mu^*$ at least one coordinate satisfies $(x^r)^{\mu^*}_R\in\{0,1\}$.
	Thus, $(x^r)^{\mu^*}$ is a fractional witness for $u(x)-u^0$ with at least one coordinate in $\{0,1\}$. Replace $x^r$ by this witness. For every coordinate equal to $0$, move the corresponding $Z_R$ from the residual part to the base part. For every coordinate equal to $1$, move the corresponding $X_R$ and replace that coordinate by $0$. Conditions \ref{cond:decomp-base}--\ref{cond:decomp-support} remain true. If an $X_R$-reduction produces a zero coordinate in $u^r$, we stop.

	The perturbation may decrease the size of the witness. In that case, increase its coordinates until either its size reaches $k^r$ or some coordinate reaches $1$. The latter outcome permits another $X_R$-reduction. If no coordinate can be increased, the witness selects the entire residual supply and is integral, contradicting that the residual is a hole. Thus, unless a zero target coordinate occurs, we can restore \ref{cond:decomp-witness}. Iterating these reductions results in $\lvert\mathcal R_c^r\rvert\le n$.
\end{proof}

The following lemma gives the remaining bounds.

\begin{lemma}\label{lem:k-1<Rc}
	Unless $u_i^r=0$ for some voter $i\in N$, the decomposition satisfies \ref{cond:decomp-bounds}.
\end{lemma}

\begin{proof}
	Let $x^r$ be the witness from \ref{cond:decomp-witness}. Since every coordinate of $x^r$ is strictly less than $1$,
	\[
		u_i^r\le u_i(x)-u_i^0\le u_i(x^r)<d_i^r.
	\]
	As $u_i^r$ is integral and positive, this gives $1\le u_i^r\le d_i^r-1$ for every voter $i$.

	If $k^r\ge|\mathcal R_c^r|$, selecting the entire residual supply gives an integral witness for $u^r$, a contradiction. Now suppose that $k^r=|\mathcal R_c^r|-1$. Select all but an arbitrary candidate type $R'\in\mathcal R_c^r$. Every voter then receives utility at least $d_i^r-1\ge u_i^r$, again contradicting that the residual is a hole. Hence,
	\[
		k^r\le|\mathcal R_c^r|-2.
	\]
	Together with $|\mathcal R_c^r|\le n$ from \Cref{lem:Rcle5}, this proves \ref{cond:decomp-bounds}.
\end{proof}

\Cref{lem:minimal-reduction01,lem:Rcle5,lem:k-1<Rc} prove \Cref{thm:decomposition}. In particular, the residual parts of decompositions satisfying \ref{cond:decomp-base}--\ref{cond:decomp-bounds} belong to a finite set for fixed $n$.

\section{Core existence for $n \le 5$}\label{sec:coreexistence}

We now state our main result: for $n \le 5$, the committee election monoid turns out to be normal, which will allow us to deduce the existence of a core committee via \Cref{thm:roundfracutils}.

\begin{restatable}{theorem}{monoidnormal}\label{thm:monoidnormal}
	For $n \le 5$, the affine monoid $M_{\mathcal{R}}$ is normal.
\end{restatable}

Let us unpack this statement. First, it is easy to see that the group of differences $\operatorname{gp}(M_{\mathcal{R}})$ equals $\Z^{|\mathcal{R}|+1+n}$, since our generator set contains the unit vector (or its negation) for each coordinate. Thus, \Cref{thm:monoidnormal} is equivalent to saying that $M_{\mathcal{R}}$ is integrally closed in $\Z^{|\mathcal{R}|+1+n}$, and hence that $M_{\mathcal{R}} = \cone(M_{\mathcal{R}}) \cap \Z^{|\mathcal{R}|+1+n}$.

Next, let us consider $\cone(M_{\mathcal{R}})=\R_{\ge 0}\{X_R, Z_R, T, S_i\} \supseteq M_{\mathcal{R}}$ and in particular $\cone(M_{\mathcal{R}}) \cap \Z^{\lvert \mathcal R \rvert + 1 + n}$.
By \Cref{prop:interpret-generators}, the set $\cone(M_{\mathcal{R}}) \cap \Z^{|\mathcal{R}|+1+n}$ consists of instance-utility pairs $(c,k,u)$ such that $(c,k)$ is an instance and $u$ is an integer utility vector that is \emph{fractional-committee-feasible}. By normality, any such instance-utility pair $(c,k,u)$ is also a member of $M_{\mathcal{R}}$, and hence $u$ is also integral-committee-feasible by \Cref{prop:interpret-generators}.

\subsection{Proof using Normaliz}
A straightforward way to establish \Cref{thm:monoidnormal} is to use the \texttt{Normaliz} program, which is ``an open source tool for computations in affine monoids, vector configurations, lattice polytopes, and rational cones'' \citep{Normaliz,BrunsKoch2001}. To use it, one needs to write down the list of generators of the monoid as given in \Cref{def:committee-monoid}. For $n = 5$, there are $31+31+1+5=68$ generators.
 Then we call the program on this input, and in less than 0.05s, \texttt{Normaliz} reports that the ``\texttt{original monoid is integrally closed in chosen lattice}'', the chosen lattice being $\Z^{37}$.%
\footnote{This can also be tried inside a web browser, see \url{https://dominikpeters.github.io/normaliz.wasm/}.} 
\texttt{Normaliz} also confirms that the committee election monoid fails to be normal for $n = 6$, as witnessed by \Cref{fig:two-triangles-intro} and \Cref{ex:four-candidates}. %

\subsection{Explicit proof}
This way of confirming normality requires us to trust the correctness of \texttt{Normaliz} and does not provide much insight; hence, we give an explicit proof. 

Let $(c,k,u)$ be an instance-utility pair such that $u$ is fractional-committee-feasible in $(c,k)$. We need to show that $u$ is also integral-committee-feasible. 
First consider the case where there is insufficient supply for the committee size, i.e., $\sum_{R\in\mathcal R}c_R<k$. Then setting the supply vector $c$ itself forms an integral committee witnessing integral-committee-feasibility of $u$ (since $u$ is fractional-committee-feasible), and we are done. So we can consider the case with sufficient supply. Let $x$ be a witness of fractional-committee-feasibility of $u$. We may assume that $\sum_{R\in\mathcal R}x_R=k$ since a smaller witness can be topped up while still witnessing $u$.

Now, assume for a contradiction that $u$ is not integral-committee-feasible. Then the decomposition theorem (\Cref{thm:decomposition}) gives us a decomposition $(c,k,u)=(c^0,k^0,u^0)+(c^r,k^r,u^r)$.
We will show that for $n \le 5$, $u^r$ is integral-committee-feasible in $(c^r,k^r)$.
Because $u^0$ is integral-committee-feasible in $(c^0,k^0)$ by construction, this will imply that $u$ is integral-committee-feasible in $(c,k)$, which will be a contradiction. In addition, our proof will proceed by induction on $n$, so in some of the following lemmas we will assume that we have already shown normality for $n - 1$ voters.

We begin the proof of integral-committee-feasibility of $u^r$ with two degenerate cases; the arguments for those cases work for arbitrary $(c,k,u)$ but we will only apply them to $(c^r,k^r,u^r)$. The first degenerate case is for very small committees, $k^r \in \{0,1\}$.

\begin{proposition}\label{prop:k=1intfeas}
	For every instance-utility pair $(c,k,u)$ with $k \in \{0,1\}$, $u$ is integral-committee-feasible for $(c,k)$ if it is fractional-committee-feasible.
\end{proposition}

\begin{proof}
	The statement is immediate for $k=0$ as every utility vector $u$ that is fractional-committee-feasible satisfies $u \le 0$, and the empty committee is a witness for integral-committee feasibility of any non-positive utility vector.
	
	Now suppose that $u = (u_i)_{i \in N} \in \Z^n$ is an integer utility vector that is fractional-committee-feasible for an instance $(c,1)$. Hence, there exists a fractional committee $x$ with $u_i(x)\ge u_i$ for all $i \in N$. 
	If $u_i \le 0$ for all $i \in N$, then the empty committee witnesses integral-committee feasibility.
	Otherwise, since $k=1$, every voter's utility is at most $1$, and hence, $u_i \le 1$ for all $i \in N$. Furthermore, $u_i=1$ implies $i \in R$ for all $R \in \mathcal{R}_c$ with $x_R>0$. Thus, there exists $R^* \in \mathcal{R}_c$ with $i \in R^*$ for all voters $i$ with $u_i=1$. Setting $x'_{R^*}=1$ and $x'_R=0$ for all $R \in \mathcal{R}_c \setminus \{R^*\}$ defines an integral committee of size $1$ with $u_i(x')\ge u_i$ for all $i \in N$, so $u$ is also integral-committee-feasible.
\end{proof}

The second degenerate case concerns voters with $u_i^r = 0$, in which case we can argue inductively.

\begin{lemma}\label{lem:zero-utility}
	Assume that the affine monoid $M_{\mathcal{R}}$ is normal for $n-1$ voters. If a non-negative integer utility vector $u$ for an instance with $n$ voters is fractional-committee-feasible and satisfies $u_i=0$ for some voter $i$, then $u$ is integral-committee-feasible.
\end{lemma}

\begin{proof}
	Delete voter $i$ and replace every candidate type $R$ by $R\setminus\{i\}$, adding the supplies of types that become the same non-empty type and discarding the empty type. The utility vector obtained from $u$ by deleting coordinate $i$ remains fractional-committee-feasible. By normality for $n-1$ voters, it has an integral witness. Replace every selected candidate by an available original candidate from which it was obtained by deleting voter $i$. The resulting committee gives all remaining voters the required utilities, and it also satisfies voter $i$ because $u_i=0$.
\end{proof}

We will use \Cref{lem:zero-utility} to deal with cases where our decomposition has a residual part with $u_i^r = 0$ for some voter $i \in N$. When this is not the case, the decomposition theorem (\Cref{thm:decomposition}) ensures that the decomposition satisfies \ref{cond:decomp-base}--\ref{cond:decomp-bounds}. We now give two lemmas that obtain additional structure of the residual part under the inductive assumption that we already know normality for $n - 1$ voters.

\begin{lemma}\label{lem:voter-deletion-bounds}
	Assume that the affine monoid $M_{\mathcal{R}}$ is normal for $n-1$ voters. For the residual part from a decomposition satisfying \ref{cond:decomp-base}--\ref{cond:decomp-bounds}, it must hold that $\bigcap_{R \in \mathcal{R}^r_c}R=\emptyset$ and $1\le u_i^r\le k^r-1$ for all $i\in N$.
\end{lemma}

\begin{proof}
	Let $x$ be the witness from \ref{cond:decomp-witness}.
	Suppose for contradiction that there is a voter $i \in \bigcap_{R \in \mathcal{R}^r_c}R \neq\emptyset$. 
	Thus, $u_i(x)=k^r$.
	Remove voter $i$ from the instance and let $N'=N\setminus\{i\}$, $u'$ be obtained from $u^r$ by deleting the $i$-th coordinate, and replacing each type $R\in\mathcal R^r_c$ by $R\setminus\{i\}$. 
	Denote this reduced instance by $(c',k^r,u')$.
	
	The same fractional committee $x$ witnesses fractional-committee feasibility of $u'$ in $(c',k^r)$ since removing voter $i$ does not change the utilities of the remaining voters.
	As \Cref{thm:monoidnormal} holds for $n-1$ voters, there exists an integral committee $x'$ for $(c',k^r)$ realizing $u'$.
	
	If $x'$ has size less than $k^r$, add arbitrary remaining original candidate types until the committee has size $k^r$; this is possible by \ref{cond:decomp-bounds}.
	
	Interpret $x'$ as a committee in the original instance by re-adding voter $i$ to every selected candidate type from which it has been excluded. 
	Since $i$ belongs to every candidate type in $\mathcal R^r_c$, voter $i$ receives utility $u_i(x')=\sum_{R\in\mathcal R^r_c} x'_R = k^r$.
	Hence, $x'$ realizes $u^r$ in the original instance. 
	This contradicts the fact that $u^r$ is not integral-committee-feasible in $(c^r,k^r)$. 
	Hence, $\bigcap_{R \in \mathcal{R}^r_c}R=\emptyset$. Since $x$ has size exactly $k^r$, every coordinate $x_R$ is positive, and some residual type excludes each voter $i$, we have $u_i(x)<k^r$. Together with \ref{cond:decomp-feasibility}, this gives $u_i^r\le k^r-1$. The lower bound follows from \ref{cond:decomp-bounds}.
\end{proof}

\begin{lemma}\label{lem:Rcbounds}
	Assume that the affine monoid $M_{\mathcal{R}}$ is normal for $n-1$ voters. For the residual part from a decomposition satisfying \ref{cond:decomp-base}--\ref{cond:decomp-bounds}, it must hold that $2\le d_i^r\le\lvert\mathcal R_c^r\rvert-1$. Moreover, $c_i^r=0$ for all $i\in N$ and $c_N^r=0$.
\end{lemma}

\begin{proof}
	By \Cref{lem:voter-deletion-bounds}, $\bigcap_{R \in \mathcal{R}^r_c}R=\emptyset$ and thus, $d^r_i \le \lvert \mathcal{R}^r_c \rvert -1$ for all $i \in N$. 
	Also by \Cref{lem:voter-deletion-bounds}, $u^r_i\ge1$, and thus $2\le d_i^r$ by \ref{cond:decomp-witness}. As $\mathcal R_c^r$ forms an antichain by \ref{cond:decomp-supply}, we have $c_N^r=0$ and $c_i^r=0$ for all $i\in N$.
\end{proof}

We can use these lemmas to establish normality for each fixed $n \le 4$.

\begin{proposition}\label{prop:normality-n-at-most-3}
	The affine monoid $M_{\mathcal{R}}$ is normal for $n \le 3$.
\end{proposition}

\begin{proof}
	We proceed by induction on $n\in\{1,2,3\}$. Suppose that a hole $(c,k,u)$ exists in $M_{\mathcal{R}}$, and take its decomposition from \Cref{thm:decomposition}. If $u_i^r = 0$ for some $i \in N$, then for $n=1$ the empty committee is an integral witness, while for $n\in\{2,3\}$, \Cref{lem:zero-utility} and the induction hypothesis give a contradiction. Otherwise, the decomposition satisfies \ref{cond:decomp-bounds} which implies $k^r\le n-2\le1$. Then \Cref{prop:k=1intfeas} gives integral-committee-feasibility of the residual part, contradicting that $(c,k,u)$ is a hole.
\end{proof}

\begin{proposition}\label{prop:normality-n4}
	The affine monoid $M_{\mathcal{R}}$ is normal for $n=4$.
\end{proposition}

\begin{proof}
	Assume for a contradiction that a hole $(c,k,u)$ exists in $M_{\mathcal{R}}$, and take its decomposition from \Cref{thm:decomposition}. If $u_i^r = 0$ for some $i \in N$, then the residual part is integral-committee-feasible by \Cref{lem:zero-utility} (applicable since we have normality for $n-1=3$ voters by \Cref{prop:normality-n-at-most-3}), a contradiction.
	Otherwise, the decomposition satisfies \ref{cond:decomp-base}--\ref{cond:decomp-bounds}.
	By \ref{cond:decomp-bounds} and \Cref{prop:k=1intfeas}, $\lvert \mathcal{R}^r_c \rvert=4$ and $k^r=2$ have to hold.
	Moreover, $u^r_i=1$ must hold for all $i \in N$ by \Cref{lem:voter-deletion-bounds}.
	By \Cref{lem:Rcbounds}, $c^r_i=0$ for all $i \in N$, $c^r_{1234}=0$, and each voter must belong to at least two available candidate types.
	
	If $\mathcal{R}^r_c$ contains a triple of voters, say $R=\{1,2,3\}$, then Voter $4$ belongs to some other available candidate type $R'$. 
	These two candidate types cover all four voters, so selecting them as an integral committee realizes $u^r=(1,1,1,1)$.
	
	So we are left with instances where $\mathcal{R}^r_c$ contains no triple of voters. 
	Hence, all four available types are pairs. 
	Construct a graph with four vertices, corresponding to the voters, with an edge between two vertices $i, j \in N$ if $\{i, j\} \in \mathcal R^r_c$.
	By \Cref{lem:Rcbounds}, $d^r_i \geq 2$, i.e., every vertex in the graph has degree at least $2$. 
	Therefore, it must be a $4$-cycle. 
	Two opposite edges of this cycle cover all four vertices, meaning that choosing the two respective candidates gives $u^r_i = 1$ for all $i \in N$.
	This contradicts our assumption that $u^r$ is not integral-committee-feasible in $(c^r,k^r)$.
\end{proof}

Finally, we turn to the case of five voters. We begin by ruling out the possibility of having a residual $(c^r,2,u^r)$ where $u^r$ is not integral-committee-feasible in $(c^r,2)$.

\begin{lemma}\label{lem:n5-k2-impossible}
	For every residual hole from a decomposition satisfying \ref{cond:decomp-base}--\ref{cond:decomp-bounds} with $n=5$ and $k^r=2$, $u^r$ is integral-committee-feasible in $(c^r,k^r)$.
\end{lemma}

\begin{proof}
	Assume for contradiction that there exists such a residual $(c^r,2,u^r)$ with five voters, and let $x$ be the witness from \ref{cond:decomp-witness}.
	By \ref{cond:decomp-bounds}, $\lvert \mathcal{R}^r_c \rvert \in \{4,5\}$.
	Furthermore, $u^r=(1,1,1,1,1)$ by \Cref{lem:voter-deletion-bounds}.
	
	We again show that two available candidate types cover all five voters, and thus, they form a witness for integral-committee feasibility of $(1,1,1,1,1)$. 
	By \Cref{lem:Rcbounds}, $c^r_{i}=0$ for all $i \in N$ and $c^r_{12345}=0$.
	
	If some $R\in\mathcal{R}^r_c$ contains four voters, then the unique voter $i$ outside $R$ belongs to some other available candidate type $R'$ as $u^r_i=1$. 
	Thus, $R$ and $R'$ cover all five voters.
	
	If all available candidate types have size $2$, then
	\[
	5 = \sum_{i \in N} u^r_i
	\le \sum_{i \in N} \sum_{R\ni i}x_R
	= \sum_{R \in \mathcal{R}^r_c}|R|x_R
	\le 2\sum_{R \in \mathcal{R}^r_c} x_R
	\le 4,
	\]
	a contradiction.
	
	Thus, some available candidate type $R$ contains exactly three voters, w.l.o.g., $R=\{1,2,3\}$. If there is an available candidate type $R'$ containing both Voters $4$ and $5$, then $R$ and $R'$ cover all voters. 
	Otherwise, no available candidate type contains both Voters $4$ and $5$, and therefore
	\[
		u_4(x)+u_5(x)
		\le \sum_{R' \in \mathcal{R}^r_c \setminus \{R\}}x_{R'}
		\le 2-x_R
	<2,
	\]
	contradicting $u_4(x)\ge 1$ and $u_5(x)\ge 1$.
	We have exhausted all possible constructions and hence, no such residual exists.
\end{proof}

\begin{proposition}\label{prop:normality-n5}
	The affine monoid $M_{\mathcal{R}}$ is normal for $n=5$.
\end{proposition}

\begin{proof}
	Assume for a contradiction that a hole $(c,k,u)$ exists in $M_{\mathcal{R}}$, and take its decomposition from \Cref{thm:decomposition}. As before, if $u_i^r = 0$ for some $i \in N$, then the residual part is integral-committee-feasible by \Cref{lem:zero-utility} (applicable since we have normality for $n-1=4$ voters by \Cref{prop:normality-n4}), a contradiction.
	
	By \Cref{lem:n5-k2-impossible}, $k^r > 2$. Together with \ref{cond:decomp-bounds}, this implies $\lvert \mathcal{R}^r_c \rvert=5$ and $k^r=3$. Then, $u^r_i \in \{1,2\}$ for all $i \in N$ by \Cref{lem:voter-deletion-bounds}, and $u_i^r\le d_i^r-1$ for all $i\in N$ by \ref{cond:decomp-bounds}.
	Let $x$ be the witness from \ref{cond:decomp-witness}.
	
	Denote by $N^+=\{i\in N:u^r_i=d^r_i-1\}$ the set of voters that get maximal possible utility in $u^r$. Note that for all other voters, $u^r_i \le d^r_i-2$ and thus, choosing any three sets in $\mathcal{R}^r_c$ yields utility at least $u^r_i$ to each of them. In particular, if $N^+ = \emptyset$, this already implies integral-committee-feasibility of $u^r$ and we are done.
	
	Otherwise, we claim that $N^+=N$. If instead $N^+ \neq N$, we could reduce the instance by excluding all voters not in $N^+$. Let $N'=N^+$, let $u'$ be obtained from $u^r$ by deleting all entries of voters not in $N^+$, and replace each candidate type $R\in\mathcal R^r_c$ by one candidate type where all voters not in $N^+$ are excluded. Denote that instance with $\lvert N^+ \rvert$ voters by $(c',3,u')$.
	
	The same fractional committee $x$ witnesses fractional-committee feasibility of $u'$ in $(c',3)$ since removing some voters does not change the utilities of the remaining voters. As \Cref{thm:monoidnormal} holds for $\lvert N^+ \rvert$ voters by \Cref{prop:normality-n-at-most-3,prop:normality-n4}, there exists an integral committee $x'$ for $(c',3)$ realizing $u'$.
	If $x'$ has size less than $3$, add arbitrary remaining original candidate types until the committee has size $3$; this is possible because $\lvert\mathcal R^r_c\rvert=5$.
	
	Interpret $x'$ as a committee in the original instance by re-adding voters from $N \setminus N^+$ to every selected candidate type from which they have been excluded. Since $x'$ selects three of the five available candidate types, each voter $i \in N\setminus N^+$ receives utility at least $d^r_i-2 \ge u^r_i$.
	Hence, $x'$ realizes $u^r$ in the original instance. This would contradict the fact that $u^r$ is not integral-committee-feasible in $(c^r,3)$.
	
	Therefore, $u^r_i=d^r_i-1$ must hold for all $i \in N$. As $u^r_i \in \{1,2\}$, this also implies that no voter can appear in more than three sets in $\mathcal{R}^r_c$.
	
	Second, we claim that for every pair of sets $R,R' \in \mathcal{R}^r_c$, there exists a voter $i \in R \cap R'$. Otherwise, $\mathcal{R}^r_c \setminus \{R,R'\}$ would form a committee $x'$ of size $3$ with $u_i(x') \ge d^r_i-1 =u^r_i$ for all $i \in N$ which would contradict the fact that $u^r$ is not integral-committee-feasible. 
	
	Third, we claim that for every pair of sets $R,R' \in \mathcal{R}^r_c$, $x_R+x_{R'}\ge1$. Otherwise, there would be a pair $R,R' \in \mathcal{R}^r_c$ with $x_R+x_{R'}<1$. By the previous claim, there would also exist a voter $i$ with $i\in R\cap R'$. However, this would imply $u_i^r\le u_i(x)<d_i^r-1$ as $x_R+x_{R'}<1$, which would contradict $u^r_i=d^r_i-1$.
	
	Fourth, we claim that $\mathcal{R}^r_c$ cannot contain a pair of voters. Otherwise, assuming w.l.o.g. $c^r_{12}=1$ and denoting by $R_{12}$ the respective candidate type, the second claim and the fact that no voter is contained in more than three sets imply that each of the two voters must be contained in two other sets while no other set can contain both of them. Thus, $d^r_1=d^r_2=3$ and $u^r_1=u^r_2=2$, so 
	\[
	4=u^r_1+u^r_2 \le u_1(x)+u_2(x)=\sum_{R \in \mathcal{R}^r_c}x_R +x_{R_{12}}=3+x_{R_{12}}<4,
	\]
	where the first equality after $u_1(x)+u_2(x)$ is implied by the fact that each set apart from $R_{12}$ contains either voter $1$ or voter $2$. This is a contradiction as $x_{R_{12}}<1$ by \ref{cond:decomp-witness}.
	
	Therefore, $\sum_{R \in \mathcal{R}^r_c}\lvert R \rvert \ge 5 \cdot 3=15$, and together with the fact that no voter can be contained in more than three sets, we can conclude that each voter is contained in exactly three sets in $\mathcal{R}^r_c$, i.e., $d^r_i=3$ for all $i \in N$. Moreover, $u^r_i=2$ for all $i \in N$ by the second claim and 
	\[
		10=\sum_{i \in N}u^r_i \le \sum_{i \in N}\sum_{R \ni i}x_R=\sum_{R \in \mathcal{R}^r_c}\lvert R \rvert x_R=
		3\sum_{R \in \mathcal{R}^r_c}x_R=9,
	\]
	a contradiction.
	
	Hence, $u^r$ must be integral-committee-feasible in $(c^r,k^r)$ and thus $u$ must be integral-committee-feasible in $(c,k)$. 
\end{proof}

Combining \Cref{prop:normality-n-at-most-3,prop:normality-n4,prop:normality-n5} proves \Cref{thm:monoidnormal}, showing that the committee election monoid is normal for $n \leq 5$.

Hence, for every instance $(c,k,b)$ with up to five voters and every fractional core committee $x$ (which exists by \Cref{thm:fraccoreexists}), the rounded-down utility vector $\lfloor u(x)\rfloor$ is integral-committee-feasible, and any witness $x'$ lies in the core by \Cref{thm:roundfracutils}. This proves the following theorem.

\begin{theorem}\label{thm:n5coreexists}
	For $n \le 5$, the core is non-empty for every instance $(c,k,b)$.
\end{theorem}

Normality of the committee election monoid for $n \leq 5$ also implies that every fractional committee can be rounded to an integral committee while preserving (the integer part of) voter utilities and while only rounding the fractions of the candidates in the fractional committee up or down.

\begin{proposition}\label{prop:roundingcommittees}
	Let $n \le 5$ and let $(c,k,b)$ be an instance. Suppose $x \in \mathcal P$ is a fractional committee. Then there exists an integral committee $x' \in \mathcal{W}$ such that $u_i(x') \ge \lfloor u_i(x) \rfloor$ for all $i \in N$ and $\lfloor x_R \rfloor \le x'_R \le \lceil x_R \rceil$ for all $R \in \mathcal{R}$.
\end{proposition}
\begin{proof}
	Let $z$ be the subcommittee with $z_R=x_R-\lfloor x_R \rfloor$ for all $R \in \mathcal{R}$. By additivity of utilities, $u_i(x)=u_i(\lfloor x \rfloor)+u_i(z)$ for all $i \in N$ and in addition, $\lfloor u(x) \rfloor =u(\lfloor x \rfloor)+\lfloor u(z) \rfloor$ as $\lfloor x \rfloor$ induces integer utilities.
	Thus, $z$ is a witness for fractional-committee feasibility of the integer utility vector $\hat u:=\lfloor u(x)\rfloor-u(\lfloor x \rfloor)$ in the instance $(c,k',b')$ with $k'=k-\sum_{R\in\mathcal R}\lfloor x_R\rfloor$ and $b'_i=\frac{k'}{n}$ for all $i \in N$.
	If $k'=0$, the claim is immediate by taking the empty committee. Hence, assume $k'>0$.
	We also note that $z_R<1$ for all $R \in \mathcal{R}$, which remains true for the instance $(c',k',b')$ with $c'_R=1$ for all $R \in \mathcal{R}$ with $z_R>0$ and $c'_R=0$ otherwise. By normality, there exists an integral committee $z'$ such that $u_i(z') \ge \hat{u}_i$ for all $i \in N$ for the instance $(c',k',b')$, implying $z' \in \{0,1\}^{\lvert \mathcal{R}\rvert}$ and $z'_R=0$ if $z_R=0$.
	
	Returning to the original instance, the committee $x'=\lfloor x\rfloor+z'$ is of size at most $k$ and $\lfloor x_R \rfloor \le x'_R \le \lceil x_R \rceil$ for all $R \in \mathcal{R}$ by construction. Finally, for every $i \in N$, we get
	\[
	u_i(x')=u_i(\lfloor x \rfloor)+u_i(z')\ge u_i(\lfloor x \rfloor)+\hat{u}_i =\lfloor u_i(x) \rfloor. \qedhere
	\]
\end{proof}

Details on how to compute committees in the core for $n \le 5$ can be found in \Cref{sec:corecomputation}.

\section{Core existence for $n\in \{6,7\}$} \label{sec:coreexistencen=6}

\Cref{thm:monoidnormal} (normality of the committee election monoid) no longer holds for $n=6$. We already discussed a counterexample with $k=3$, consisting of two triangles, in the introduction (see \Cref{fig:two-triangles-intro}).
There are also other interesting counterexamples with $n = 6$, such as the following one that uses a committee size of only $k = 2$ and has only four candidates.

\begin{example}[Normality-violating instance with $n = 6$ and four candidates] \label{ex:four-candidates}
	Let $n=6$, $k=2$, and $\mathcal R_c=\{126,456,234,135\}$ with $c_R = 1$ for all $R \in \mathcal R_c$, see \Cref{fig:four-candidates}. Each voter $i$ has budget $b_i=\frac{1}{3}$.
	The integer utility vector $u=(1,1,1,1,1,1)$ is fractional-committee-feasible as it is realized by the fractional committee with $x_R=\frac{1}{2}$ for all $R \in \mathcal{R}_c$. However, $u$ is not integral-committee-feasible. 
	To see this, note that each selected candidate is approved by exactly three voters, so every committee of size at most $2$ yields total utility at most $6$. Realizing $u$ would require total utility at least $6$, so such a committee would have size $2$, and all six voters would need to receive utility exactly $1$. 
	However, any two selected candidates always share a voter, so some voter receives utility $2$, a contradiction. \qed
\end{example}

The core is still non-empty for both examples. In the instance from \Cref{fig:two-triangles-intro}, we can take $x_{12} = x_{13} = x_{45} = 1$, and in \Cref{ex:four-candidates}, we can take $x_{126}=x_{456}=1$, setting $x_{R}=0$ for all other $R \in \mathcal{R}_c$ in both cases.

\begin{figure}[thb]
		\centering
		\begin{tikzpicture}[
			scale=0.8,
			voter/.style={circle,draw,fill=white}
			]
			\foreach \i/\ang in {1/90,2/30,3/-30,4/-90,5/-150,6/150}{
				\node[voter] (v\i) at (\ang:2) {$\i$};
			}
			
			\draw[very thick,blue!75]          (v1) -- (v2) -- (v6) -- (v1);
			\draw[very thick,red!75]           (v4) -- (v5) -- (v6) -- (v4);
			\draw[very thick,green!60!black]   (v2) -- (v3) -- (v4) -- (v2);
			\draw[very thick, double=orange!80!black, draw=white, double distance=1.2pt]  (v1) -- (v3) -- (v5) -- (v1);
		\end{tikzpicture}
		\caption{Instance of \Cref{ex:four-candidates} showing a violation of normality for $n = 6$ voters.}
		\label{fig:four-candidates}
\end{figure}

Observe that both instances, together with their respective utility vectors, can occur as residuals in decompositions satisfying \ref{cond:decomp-base}--\ref{cond:decomp-bounds}. These examples show that additional tools are needed to prove non-emptiness of the core for $n>5$. In what follows, we develop such tools and use them to extend our results from up to five voters to up to seven voters.

\begin{theorem}\label{thm:upton7coreexists}
	For $n \leq 7$, the core is non-empty for every instance $(c,k,b)$.
\end{theorem}

For a given instance $(c,k,b)$, we can assume w.l.o.g. that $\sum_{R\in\mathcal R} c_R > k$, as otherwise the set of all candidates is affordable and constitutes an integral core committee. We start with a specific fractional core committee $x$, namely one that forms a Lindahl equilibrium $(x,p)$ with $\sum_{R \in \mathcal{R}}x_R=k$, which always exists by \Cref{thm:fraccoreexists}. If the rounded-down utility vector $\lfloor u(x) \rfloor$ is integral-committee-feasible, its witness $x'$ lies in the core by \Cref{thm:roundfracutils}.
Therefore, we can focus on instances $(c,k)$ for which $u:=\lfloor u(x) \rfloor$ is fractional-committee-feasible but \emph{not} integral-committee-feasible, i.e., $(c,k,u)$ is a \emph{hole} in the corresponding monoid (see \Cref{sec:decomposition}).

We may therefore apply \Cref{thm:decomposition} to $u=\lfloor u(x)\rfloor$ and the witness $x$. The resulting residual 
\[
	u^r=u-u^0=\lfloor u(x)-u^0\rfloor,
\]
because $u^0$ is integral by \ref{cond:decomp-base}. The theorem guarantees that $(c^r,k^r,u^r)$ is a hole. When \ref{cond:decomp-base}--\ref{cond:decomp-bounds} hold, we call the residual part a \emph{minimal hole}.

For $n \in \{6,7\}$, we now scan over all possible minimal holes and verify that all of them still let us construct an element in the core. For the latter task, we will use properties of Lindahl equilibria. Thus, we begin by formally defining the Lindahl equilibrium, discussing its properties, and deriving sufficient conditions implying the existence of integral core committees.

\subsection{Lindahl equilibrium and sufficient conditions for core existence} \label{sec:lindahl}

We use the definition as given by \citet{KrPe25b} in their capped public goods setting.

\begin{definition}[Lindahl equilibrium]\label{def:Lindahleq}
	For a given instance $(c,k,b)$, let $x \in \mathcal{P}$ be a fractional committee and $p=(p_i)_{i \in N}\in \mathbb{R}^{\lvert \mathcal{R}\rvert \times n}_{\ge 0}$ be a collection of individual price vectors. Then $(x,p)$ is a \emph{Lindahl equilibrium} if
	
	\begin{enumerate}
		\item $x$ is \emph{affordable}: we have $\langle x,p_i \rangle \le b_i$ for all $i \in N$,
		\item $x$ is \emph{utility-maximizing}: for every $i \in N$ and $y$ with $0 \leq y_R \leq c_R$ for all $R \in \mathcal{R}$ and $\langle y,p_i \rangle \le b_i$, we have $u_i(x)\ge u_i(y)$, and 
		\item prices are \emph{profit-maximizing}: for every $R \in \mathcal{R}$, we have $\sum_{i \in N}p_{iR} \le 1$, and whenever $x_R>0$ then $\sum_{i \in N}p_{iR} = 1$.
	\end{enumerate}
\end{definition}

We refer to \citet{KrPe25b} for a detailed discussion of these requirements. 
Given a Lindahl equilibrium $(x,p)$, we say that a voter $i \in N$ is \emph{unsatiated} if there exists some $R\ni i$ with $0<x_R<c_R$. The following proposition collects useful properties of unsatiated voters. Its first two items encode bang-per-buck conclusions from utility maximization. The third item shows that an unsatiated voter spends their full budget. The fourth item gives a property that \citet{KrPe25b} call \emph{zero-respecting}. The fifth item gives an alternative formulation of utility maximization. The final two items connect a voter's spending to their utility.

\begin{proposition}\label{prop:Lindahlprops}
	Let $(c,k,b)$ be an arbitrary instance and let $(x,p)$ be a Lindahl equilibrium. For every unsatiated voter $i \in N$ with $0<x_R<c_R$ for some $R\ni i$, we have  
	\begin{enumerate}
		\item $p_{iR}\ge p_{iR'}$ for all $R' \ni i$ with $0 < x_{R'}=c_{R'}$,
		\item $p_{iR}=p_{iR'}$ for all $R' \ni i$ with $0<x_{R'}<c_{R'}$, 
		\item $\langle x,p_i\rangle=b_i$,  
		\item $p_{iR'}=0$ for all $R' \not \ni i$ with $x_{R'}>0$, 
		\item $\langle y,p_i\rangle\ge b_i$ for all $y \in \mathcal{P}$ with $u_i(y)\ge u_i(x)$,
		\item $\langle x,p_i \rangle \le p_{iR} u_i(x)$, and
		\item if $|R|\ge 2$, then $u_i(x)>b_i$.
	\end{enumerate}
\end{proposition}

\begin{proof}
	Let $i \in N$ be a voter with $0<x_R<c_R$ for some $R\ni i$.
	
	First, if for some $i \in N$, $p_{iR}< p_{iR'}$ for some $R' \in \mathcal{R}$ with $x_{R'}=c_{R'}$, and $R' \ni i$, then $x$ would not be utility-maximizing for voter $i$, as voter $i$ could purchase less of $R'$ and a higher amount of $R$ (since $p_{iR}<p_{iR'}$), which would result in a utility higher than $u_i(x)$.
	
	Second, if $p_{iR'}>p_{iR}$ for some $R' \in \mathcal{R}$ with $0<x_{R'}<c_{R'}$ and $R' \ni i$, then $x$ would again not be utility-maximizing for voter $i$, as voter $i$ could purchase less of $R'$ and a higher amount of $R$ (since $p_{iR'}>p_{iR}$), which would result in a utility higher than $u_i(x)$.
	
	Third, if $\langle x,p_i\rangle<b_i$, voter $i$ could purchase more of $R$ and $x$ would not be utility-maximizing. 
	
	Fourth, if $p_{iR'}>0$ for some $R' \not \ni i$ with $x_{R'}>0$, voter $i$ could purchase less of $R'$ and use the saved budget to purchase more of $R$. 
	
	Fifth, if for some $y \in \mathcal{P}$ with $u_i(y)\ge u_i(x)$, we had $\langle y,p_i\rangle< b_i = \langle x, p_i\rangle$, then we would need to have $y_{R^*}<x_{R^*}$ for some $R^* \ni i$ (since $p_{iR'} = 0$ for all $R' \not\ni i$). Then, voter $i$ could purchase more of $R^*$ with the remaining budget, again contradicting that $x$ is utility-maximizing.
	
	For the sixth property, we conclude $p_{iR'}\le p_{iR}$ for all $R' \ni i$ with $x_{R'}>0$ and $p_{iR'}=0$ for all $R' \not \ni i$ with $x_{R'}>0$, which results in
	\[
		\langle x,p_i \rangle =\sum_{R' \ni i}p_{iR'}x_{R'} \le p_{iR}\sum_{R' \ni i}x_{R'}=p_{iR} u_i(x). 
	\]

	Finally, suppose that $|R|\ge2$ and choose some $j\in R\setminus\{i\}$. Voter $j$ is also unsatiated by definition. Since $b_j>0$, budget exhaustion and the preceding inequality applied to $j$ give
	$
		0<b_j=\langle x,p_j\rangle\le p_{jR}u_j(x),
	$
	so $p_{jR}>0$. Profit maximization for type $R$ implies $p_{iR}\le1-p_{jR}<1$. Applying budget exhaustion and the sixth property to voter $i$, and noting that $u_i(x)>0$ because $x_R>0$, yields
	$
		b_i=\langle x,p_i\rangle\le p_{iR}u_i(x)<u_i(x).
	$
\end{proof}

We prove two sufficient conditions for converting a Lindahl equilibrium fractional committee $x$ into an integral core committee.  
While in \Cref{thm:roundfracutils}, we considered integral committees inducing utilities $\lfloor u(x)\rfloor$, we will now consider the utility vector $\lfloor u(x)\rfloor - e_{i^*}$, where $e_{i^*}$ denotes the $i^*$-th unit vector in $\mathbb{Z}^n$. Thus, we will select some voter $i^*$ and reduce $i^*$'s target utility by 1.
Both conditions allow the rounded-down Lindahl utility of one specified voter $i^*$ to be reduced by one.
Under certain conditions on the voter $i^*$, we can show that any integral committee inducing this reduced utility vector is in the core.

\begin{theorem}\label{thm:reduceLindahlutils}
	Given an instance $(c,k,b)$ and a Lindahl equilibrium $(x,p)$, assume that there exists an unsatiated voter $i^* \in N$ with $u_{i^*}(x)>b_{i^*}$ (voter $i^*$ receives more utility than is attainable using only $b_{i^*}$) and $u_{i^*}(x) = \lfloor u_{i^*}(x) \rfloor$ (voter $i^*$ receives integral utility in the Lindahl equilibrium).
	Then every witness for integral-committee feasibility of $\lfloor u(x) \rfloor-e_{i^*}$ lies in the core.
\end{theorem}

\begin{proof}
	Let $x' \in \mathcal{W}$ be a witness for integral-committee feasibility of $u':=\lfloor u(x) \rfloor-e_{i^*}$ and assume for contradiction that there exists $S \subseteq N$ and $y \in \mathcal{W}_S$ such that $u_i(y)>u'_i$ for all $i \in S$.
	
	First, if $i^* \not \in S$, $u_i(y) >u'_i =\lfloor u_i(x) \rfloor$ for all $i \in S$, contradicting that $x$ is a fractional core committee as $u_i(y)$ is integral and thus larger than $u_i(x)$.
	
	Second, if $S=\{i^*\}$, $u_{i^*}(y) >u'_{i^*}=\lfloor u_{i^*}(x) \rfloor-1=u_{i^*}(x)-1$, which implies $u_{i^*}(y) \ge u_{i^*}(x) >b_{i^*}$ and thus contradicts $y \in \mathcal{W}_S$. 
	
	Finally, if $\{i^*\} \subsetneq S$, we have $u_i(y)\ge \lfloor u_i(x)\rfloor+1 >u_i(x)$ for all $i \in S\setminus \{i^*\}$ as well as $u_{i^*}(y)\ge \lfloor u_{i^*}(x) \rfloor=u_{i^*}(x)$. As $x$ is utility-maximizing for all $i \in S$ (see \Cref{def:Lindahleq}), we get that $\langle y,p_i\rangle > b_i$ for all $i \in S\setminus \{i^*\}$ as $u_i(y)>u_i(x)$. Since voter $i^*$ is unsatiated, we obtain $\langle y,p_{i^*} \rangle \ge  b_{i^*}$ by \Cref{prop:Lindahlprops} as $u_{i^*}(y)\ge u_{i^*}(x)$.
	All in all, we get
	\begin{align*}
		b(S)<\sum_{i \in S}\langle y,p_i \rangle
		=\sum_{R \in \mathcal{R}}\left(y_R \cdot \sum_{i \in S}p_{iR}\right) \le \sum_{R \in \mathcal{R}}y_R
	\end{align*}
	where the last inequality follows from profit maximization; this contradicts $y \in \mathcal{W}_S$.
\end{proof}

For $n=6$, we will show that such a voter $i^*$ can always be found and the corresponding reduced utility vector is integral-committee-feasible.

For $n=7$, there are cases where no voter receives integral utility in the Lindahl equilibrium, and so we will need another condition. For an unsatiated voter $i$, write $\beta_{i} = p_{iR}$ where $R \ni i$ is some candidate type with $0 < x_R < c_R$ (this is well-defined by \Cref{prop:Lindahlprops}). 
Our result will show that inducing the utility vector $\lfloor u(x) \rfloor-e_{i^*}$ suffices to be in the core when additionally assuming $\sum_{i \in N}\beta_{i} \phi_i+\beta_{i^*}<1$, where $\phi_i:=u_i(x)-\lfloor u_i(x) \rfloor$.
Intuitively, the condition says that the budget saved by lowering voter $i^*$'s utility is not sufficient to compensate all other members of a deviating coalition $S\subseteq N$. More precisely, suppose a coalition chooses some $y\in\mathcal{P}_S$. Lowering voter $i^*$'s utility by $\phi_{i^*}$ frees only a limited amount of budget. This saving is not enough to raise every other voter $i\in S$ from $\lfloor u_i(x)\rfloor$ to $\lceil u_i(x)\rceil$, which would require an additional gain of $1-\phi_i$ for each such voter.

\begin{theorem}\label{thm:reduceLindahlutilsfor7}
	Consider an instance $(c,k,b)$ and a Lindahl equilibrium $(x,p)$ in which every voter is unsatiated. Use the values $\beta_i$ and $\phi_i$ defined above.
	Suppose that for some voter $i^*$ we have $\sum_{i \in N}\beta_{i} \phi_i+\beta_{i^*}<1$ and $\lfloor u_{i^*}(x) \rfloor \ge 1$. Then every witness for integral-committee-feasibility of $\lfloor u(x) \rfloor-e_{i^*}$ lies in the core.
\end{theorem}

\begin{proof}
	Define $u':=\lfloor u(x)\rfloor-e_{i^*}$ and let $x'\in\mathcal W$ be a witness for integral-committee feasibility
	of $u'$. Thus, $u_i(x')\ge u'_i$ for every $i\in N$.

	Suppose, toward a contradiction, that $x'$ does not lie in the core. Then there exist a non-empty coalition $S\subseteq N$ and a committee $y\in\mathcal W_S$ such that $u_i(y)>u_i(x')$ for every $i\in S$. Since $x'$ and $y$ are integral committees, it follows that $u_i(y)\ge u'_i+1$ for every $i\in S$.

	We first establish a \emph{marginal-expenditure inequality}. We claim that, for every voter $i$ and every $z$ satisfying $0\le z_R\le c_R$ for all $R\in\mathcal R$,
	\[
		\langle z,p_i\rangle-b_i \ge \beta_i\bigl(u_i(z)-u_i(x)\bigr).
	\]
	To see this, first consider a type $R \in \mathcal{R}_c$ approved by voter $i$. By \Cref{prop:Lindahlprops}, together with voter $i$'s utility maximization in the case $x_R=0$, we have
	\[
		\begin{array}{c|c|c}
		x_R=0 & 0<x_R<c_R & x_R=c_R\\ \hline
		p_{iR}\ge\beta_i & p_{iR}=\beta_i & p_{iR}\le\beta_i.
		\end{array}
	\]
	Indeed, the only additional case is $x_R=0$, for which utility maximization implies $p_{iR}\ge\beta_i$. Since $z_R-x_R\ge0$ when $x_R=0$, $z_R-x_R\le0$ when $x_R=c_R$, and the middle case holds with equality, it follows in all three cases that
	\[
		p_{iR}(z_R-x_R) \ge \beta_i(z_R-x_R).
	\]

	Now consider a type $R$ not approved by voter $i$. If $x_R>0$, then $p_{iR}=0$ by \Cref{prop:Lindahlprops}; if $x_R=0$, then non-negativity of prices and $z_R\ge0$ imply $p_{iR}(z_R-x_R)=p_{iR}z_R\ge0$.
	Thus, summing over all $R\in\mathcal R$, we obtain
	\[
		\langle z-x,p_i\rangle \ge \beta_i\bigl(u_i(z)-u_i(x)\bigr).
	\]
	Since every voter is unsatiated by assumption, budget exhaustion yields $\langle x,p_i\rangle=b_i$.
	Therefore,
	\[
		\langle z,p_i\rangle-b_i \ge \beta_i\bigl(u_i(z)-u_i(x)\bigr).
	\]
	With this inequality in hand, we now show that no coalition can deviate. As in the proof of \Cref{thm:reduceLindahlutils}, any blocking coalition $S$ must contain $i^*$.

	If $S=\{i^*\}$, we get $u'_{i^*}+1=\lfloor u_{i^*}(x) \rfloor \le u_{i^*}(y) \le b_{i^*}$ since $y \in \mathcal{W}_S$.
	Hence,
	\begin{align*}
		\lfloor u_{i^*}(x) \rfloor \le b_{i^*}\le \beta_{i^*}u_{i^*}(x) &= \beta_{i^*}\lfloor u_{i^*}(x) \rfloor + \beta_{i^*}\phi_{i^*}\\
		&\le \beta_{i^*}\lfloor u_{i^*}(x) \rfloor + \sum_{i \in N}\beta_{i} \phi_i.
	\end{align*}
	Thus, $1-\beta_{i^*}\le \sum_{i \in N}\beta_{i} \phi_i$ as $\lfloor u_{i^*}(x) \rfloor \ge 1$ and $1-\beta_{i^*} \geq 0$, contradicting our assumption $\sum_{i \in N}\beta_{i} \phi_i+\beta_{i^*}<1$.

	It remains to consider the case $\{i^*\}\subsetneq S$. Blocking and integrality imply
	\begin{align*}
		u_{i^*}(y) &\ge u'_{i^*}+1 =\lfloor u_{i^*} \rfloor =u_{i^*}(x)-\phi_{i^*} &&\quad 
		\text{for }i^* \text{, and} \\
		u_i(y) &\ge u'_i+1 =\lfloor u_i \rfloor+1 =u_i(x)+(1-\phi_i) &&\quad \text{for all }i\in S\setminus\{i^*\}.
	\end{align*}

	Applying the \emph{marginal-expenditure inequality} to $y$ and summing over
	$i\in S$ yields
	\[
		\sum_{i\in S}\langle y,p_i\rangle
		\ge b(S)-\beta_{i^*}\phi_{i^*} +\sum_{i\in S\setminus\{i^*\}} \beta_i(1-\phi_i).
	\]
	By profit maximization and the affordability of $y\in\mathcal W_S$, we have $\sum_{i\in S}\langle y,p_i\rangle \leq \langle y,1\rangle \leq b(S)$. Thus, the preceding inequality simplifies to
	\[
		\sum_{i\in S\setminus\{i^*\}} \beta_i(1-\phi_i) \le \beta_{i^*}\phi_{i^*}.
	\]
	Rearranging this inequality finally gives us
	\[
		\sum_{i \in S} \beta_i \leq \beta_{i^*} + \sum_{i \in S} \beta_i \phi_i \leq \beta_{i^*} + \sum_{i \in N} \beta_i \phi_i < 1.
	\] 
	Now, let $u^{\max}:=\max_{i \in S}u'_i+1$. As $u^{\max}\le u_i(y)$ for some $i \in S$, $u^{\max}\le b(S)$.
	Then,
	\begin{align*}
		u^{\max} \le b(S)=\sum_{i \in S}\langle x,p_i\rangle &\le \sum_{i \in S} \beta_i u_i(x)
		\\
		&= \sum_{i \in S} \beta_i \left(u'_i+1\right)+\beta_{i^*}\phi_{i^*}-\,\sum_{\mathclap{i \in S\setminus\{i^*\}}} \: \beta_i \left(1-\phi_i\right)
		\\
		&\le u^{\max} \cdot \sum_{i \in S} \beta_i +\beta_{i^*}\phi_{i^*}-\,\sum_{\mathclap{i \in S\setminus\{i^*\}}} \: \beta_i \left(1-\phi_i\right).
	\end{align*}	
	Thus, $1-\sum_{i \in S}\beta_i\le \beta_{i^*}\phi_{i^*}-\sum_{i \in S\setminus\{i^*\}} \beta_i \left(1-\phi_i\right)$ as $u^{\max} \ge 1$ and $1-\sum_{i \in S}\beta_i > 0$.
	Rearranging terms yields
	\begin{align*}
		1 \le \beta_{i^*}\phi_{i^*}+\beta_{i^*}+\sum_{i \in S\setminus\{i^*\}} \beta_i \phi_i
		=\sum_{i \in S} \beta_i \phi_i+\beta_{i^*} \le \sum_{i \in N} \beta_i \phi_i + \beta_{i^*}
	\end{align*}
	which again contradicts our assumption $\sum_{i \in N}\beta_{i} \phi_i+\beta_{i^*}<1$.	
\end{proof}

\subsection{Core existence for $n=6$}\label{sec:core_for_six}

Our proof of core existence for $n = 6$ combines a computer enumeration with an invocation of our sufficient condition from \Cref{thm:reduceLindahlutils}. We begin by explaining the computer component, which enumerates all structures that form what we call a candidate minimal hole.

\begin{definition}
	An instance-utility pair $(c^r,k^r,u^r)$ is a \emph{candidate minimal hole} if it is a hole and satisfies the following conditions for each $i \in N$:
	\begin{enumerate}
		\renewcommand{\theenumi}{(R\arabic{enumi})}%
		\renewcommand{\labelenumi}{\theenumi}%
		\item\label{cond:antichain} $\mathcal{R}^r_c$ forms an antichain,
		\item\label{cond:supply} $c^r_i=0$, $c^r_N=0$, and $c^r_R=1$ for all $R\in\mathcal{R}^r_c$,
		\item\label{cond:degree} $2 \le d_i^r \le \lvert \mathcal{R}^r_c\rvert-1$,
		\item\label{cond:card} $3<k^r+2 \le \lvert \mathcal{R}^r_c\rvert \le n$, and %
		\item\label{cond:util} $1 \le u_i^r \le \min\{d^r_i-1, k^r-1\}$.
	\end{enumerate}
\end{definition}

Our core existence proof decomposes an input instance and then shows that the residual part of this decomposition forms a candidate minimal hole.
Note that conditions \ref{cond:antichain}--\ref{cond:util} imply that there are only \emph{finitely many} candidate minimal holes for $n = 6$. We use a computer program to enumerate all of them: it performs a depth-first search over all antichains and selects all those that satisfy conditions \ref{cond:antichain}--\ref{cond:degree}. It then further filters out any duplicates with respect to relabeling. For each of the resulting candidate antichains, it iterates over all choices of $k^r$ (compatible with \ref{cond:card}) and all utility vectors $u^r$ (compatible with \ref{cond:util}). For each resulting combination $(c^r,k^r,u^r)$, it checks that it is fractional-committee-feasible (using an exact LP) but not integral-committee-feasible (using brute-force search over all possible integral committees), i.e., that it constitutes a \emph{hole}. Hence, this procedure yields a list of 50 instance-utility pairs that are candidate minimal holes (based on 23 different antichains), which is exhaustive up to relabeling. See \Cref{tbl:50-holes}.

\begin{table}[t]
	\centering
	\scriptsize
	\begin{tabular}{L{0.02\linewidth}L{0.24\linewidth}C{0.01\linewidth}L{0.11\linewidth}@{\hspace{0.015\linewidth}}C{0.02\linewidth}L{0.27\linewidth}C{0.01\linewidth}L{0.11\linewidth}@{}}
		\cmidrule[\heavyrulewidth](r{1.0em}){1-4}\cmidrule[\heavyrulewidth]{5-8}
		Hole & $\mathcal{R}_c$ & $k$ & $u$ & Hole & $\mathcal{R}_c$ & $k$ & $u$ \\
		\cmidrule(r{1.0em}){1-4}\cmidrule{5-8}
		1 & $\{12,\,13,\,23,\,45,\,46,\,56\}$ & 3 & $(1,1,1,1,1,1)$ & 26 & $\{123,\,124,\,135,\,145,\,346,\,256\}$ & 4 & $(3,2,2,2,2,1)$ \\
		2 & $\{12,\,13,\,234,\,235,\,236,\,456\}$ & 3 & $(1,2,2,1,1,1)$ & 27 & $\{123,\,124,\,135,\,245,\,236,\,146\}$ & 2 & $(1,1,1,1,1,1)$ \\
		3 & $\{12,\,34,\,135,\,245,\,236,\,146\}$ & 2 & $(1,1,1,1,1,1)$ & 28 & $\{123,\,124,\,135,\,245,\,236,\,146\}$ & 4 & $(3,3,2,2,1,1)$ \\
		4 & $\{12,\,34,\,135,\,245,\,236,\,146\}$ & 3 & $(1,1,2,2,1,1)$ & 29 & $\{123,\,124,\,135,\,245,\,346,\,156\}$ & 2 & $(1,1,1,1,1,1)$ \\
		5 & $\{12,\,134,\,135,\,245,\,236,\,146\}$ & 2 & $(1,1,1,1,1,1)$ & 30 & $\{123,\,124,\,135,\,245,\,346,\,156\}$ & 3 & $(2,1,2,1,2,1)$ \\
		6 & $\{12,\,134,\,135,\,245,\,236,\,146\}$ & 3 & $(2,1,2,2,1,1)$ & 31 & $\{123,\,124,\,135,\,245,\,346,\,156\}$ & 3 & $(2,2,1,2,1,1)$ \\
		7 & $\{12,\,134,\,135,\,245,\,236,\,146\}$ & 3 & $(2,2,1,1,1,1)$ & 32 & $\{123,\,124,\,135,\,245,\,346,\,156\}$ & 4 & $(3,2,2,2,2,1)$ \\
		8 & $\{12,\,134,\,135,\,245,\,236,\,146\}$ & 4 & $(3,2,2,2,1,1)$ & 33 & $\{123,\,124,\,135,\,2345,\,346,\,256\}$ & 2 & $(1,1,1,1,1,1)$ \\
		9 & $\{12,\,134,\,235,\,345,\,246,\,156\}$ & 2 & $(1,1,1,1,1,1)$ & 34 & $\{123,\,124,\,135,\,2345,\,346,\,256\}$ & 3 & $(1,2,2,2,2,1)$ \\
		10 & $\{12,\,134,\,235,\,345,\,246,\,156\}$ & 3 & $(1,1,2,2,2,1)$ & 35 & $\{123,\,124,\,135,\,2345,\,346,\,256\}$ & 3 & $(2,2,2,1,1,1)$ \\
		11 & $\{12,\,134,\,235,\,345,\,246,\,156\}$ & 3 & $(2,2,1,1,1,1)$ & 36 & $\{123,\,124,\,135,\,2345,\,346,\,256\}$ & 4 & $(2,3,3,2,2,1)$ \\
		12 & $\{12,\,134,\,235,\,246,\,156\}$ & 2 & $(1,1,1,1,1,1)$ & 37 & $\{123,\,124,\,135,\,346,\,256\}$ & 2 & $(1,1,1,1,1,1)$ \\
		13 & $\{12,\,134,\,235,\,246,\,156\}$ & 3 & $(2,2,1,1,1,1)$ & 38 & $\{123,\,124,\,135,\,346,\,256\}$ & 3 & $(2,2,2,1,1,1)$ \\
		14 & $\{12,\,134,\,235,\,246,\,156,\,3456\}$ & 3 & $(1,1,2,2,2,2)$ & 39 & $\{123,\,124,\,135,\,346,\,256,\,456\}$ & 3 & $(1,1,1,2,2,2)$ \\
		15 & $\{12,\,134,\,235,\,246,\,156,\,3456\}$ & 3 & $(2,2,1,1,1,1)$ & 40 & $\{123,\,124,\,135,\,346,\,256,\,1456\}$ & 3 & $(2,1,1,2,2,2)$ \\
		16 & $\{123,\,124,\,134,\,235,\,236,\,156\}$ & 3 & $(2,2,2,1,1,1)$ & 41 & $\{123,\,124,\,135,\,346,\,256,\,1456\}$ & 3 & $(2,2,2,1,1,1)$ \\
		17 & $\{123,\,124,\,134,\,235,\,246,\,156\}$ & 2 & $(1,1,1,1,1,1)$ & 42 & $\{123,\,124,\,135,\,346,\,256,\,1456\}$ & 4 & $(3,2,2,2,2,2)$ \\
		18 & $\{123,\,124,\,134,\,235,\,246,\,156\}$ & 3 & $(2,2,1,2,1,1)$ & 43 & $\{123,\,124,\,345,\,346,\,156,\,256\}$ & 2 & $(1,1,1,1,1,1)$ \\
		19 & $\{123,\,124,\,134,\,235,\,246,\,156\}$ & 4 & $(3,3,2,2,1,1)$ & 44 & $\{123,\,124,\,345,\,346,\,156,\,256\}$ & 4 & $(2,2,2,2,2,2)$ \\
		20 & $\{123,\,124,\,134,\,2345,\,2346,\,156\}$ & 3 & $(2,2,2,2,1,1)$ & 45 & $\{123,\,145,\,2345,\,246,\,1346,\,356\}$ & 3 & $(1,2,2,2,2,1)$ \\
		21 & $\{123,\,124,\,125,\,345,\,136,\,246\}$ & 2 & $(1,1,1,1,1,1)$ & 46 & $\{123,\,145,\,2345,\,246,\,1346,\,356\}$ & 4 & $(2,2,3,3,2,2)$ \\
		22 & $\{123,\,124,\,125,\,345,\,136,\,246\}$ & 3 & $(2,2,1,2,1,1)$ & 47 & $\{123,\,145,\,2345,\,246,\,356\}$ & 2 & $(1,1,1,1,1,1)$ \\
		23 & $\{123,\,124,\,125,\,345,\,136,\,246\}$ & 4 & $(3,3,2,2,1,1)$ & 48 & $\{123,\,145,\,2345,\,246,\,356\}$ & 3 & $(1,2,2,2,2,1)$ \\
		24 & $\{123,\,124,\,135,\,145,\,346,\,256\}$ & 2 & $(1,1,1,1,1,1)$ & 49 & $\{123,\,145,\,246,\,356\}$ & 2 & $(1,1,1,1,1,1)$ \\
		25 & $\{123,\,124,\,135,\,145,\,346,\,256\}$ & 3 & $(2,1,1,2,2,1)$ & 50 & $\{1234,\,1235,\,1236,\,1456,\,2456,\,3456\}$ & 3 & $(2,2,2,2,2,2)$ \\
		\cmidrule[\heavyrulewidth](r{1.0em}){1-4}\cmidrule[\heavyrulewidth]{5-8}
	\end{tabular}
	\caption{The 50 residual instance-utility pairs with $n = 6$ that constitute a candidate minimal hole (up to voter relabeling). The instance from \Cref{fig:two-triangles-intro} corresponds to Hole~1, and the one from \Cref{fig:four-candidates} corresponds to Hole~49.}
	\label{tbl:50-holes}
\end{table}

Next, by going through each of these 50 minimal holes $(c^r,k^r,u^r)$, we can easily check that each satisfies two important properties: there is a voter $i^*$ who obtains integral utility in a Lindahl equilibrium, and rounding down the utility vector and reducing the utility of $i^*$ by 1 results in a utility vector that is integral-committee-feasible. For the first part, we can avoid having to reason about Lindahl equilibria by showing that \emph{any} witness of fractional-committee feasibility gives $i^*$ integral utility.

Reproduction scripts for these computations are available at \url{https://github.com/DominikPeters/corefewvoters}.

\begin{fact}[Checked computationally]\label{thm:pinnable-patchable}
	Let $(c^r,k^r,u^r)$ be one of the 50 candidate minimal holes with six voters. Let $x$ be any witness for fractional-committee feasibility of $u^r$ in $(c^r,k^r)$. Then for every voter $i\in N$, we get $u_i(x)=u_i^r$ and $u^r-e_i$ is integral-committee-feasible in $(c^r,k^r)$.
\end{fact}

Based on this computation, we can now invoke \Cref{thm:reduceLindahlutils}, our first sufficient condition for rounding a Lindahl equilibrium to an integral core committee, to establish core existence for arbitrary instances $(c,k,b)$ with $n=6$.

\begin{theorem}\label{thm:n6coreexists}
	For $n=6$, the core is non-empty for every instance $(c,k,b)$.
\end{theorem}

\begin{proof}
	For $\sum_{R \in \mathcal{R}}c_R \le k$, the committee consisting of all candidates is feasible and lies in the core.
	For every other instance $(c,k,b)$, a Lindahl equilibrium $(x,p)$ with $\sum_{R \in \mathcal{R}}x_R=k$ is known to exist by \Cref{thm:fraccoreexists}. If the rounded-down utility vector $u:=\lfloor u(x)\rfloor$ is integral-committee-feasible, every witness lies in the core by \Cref{thm:roundfracutils}.
	
	Otherwise, apply \Cref{thm:decomposition} to $u$ and $x$ to get a decomposition $(c,k,u)=(c^0,k^0,u^0)+(c^r,k^r,u^r)$. The residual cannot have a zero coordinate by \Cref{lem:zero-utility} and \Cref{thm:monoidnormal}, so $(c^r,k^r,u^r)$ is a minimal hole, i.e., \ref{cond:decomp-base}--\ref{cond:decomp-bounds} hold. 
	Now by using \ref{cond:decomp-supply}, \ref{cond:decomp-bounds}, \Cref{lem:voter-deletion-bounds,lem:Rcbounds}, we get that $(c^r,k^r,u^r)$ also satisfies \ref{cond:antichain}--\ref{cond:util} and, thus, is a candidate minimal hole. It must therefore (up to relabeling) be one of the $50$ candidate minimal holes found by our computer enumeration. By \Cref{thm:pinnable-patchable}, $u^r-e_i$ is integral-committee-feasible for all $i \in N$. By \ref{cond:decomp-base}, the base part $(c^0,k^0,u^0)$ has an integral committee that gives utility exactly $u^0$. Combining it with a witness for $u^r-e_i$ shows that $u-e_i=u^0+u^r-e_i$ is integral-committee-feasible in $(c,k)$ for all $i \in N$.
	
	Next, we claim that $u_{i}(x)>b_{i}$ for all $i \in N$. Fix an arbitrary voter $i \in N$.
	By \ref{cond:decomp-witness} and \ref{cond:decomp-bounds}, voter $i$ approves some type $R^{i}\in\mathcal R_c^r$. By \ref{cond:decomp-support}, $x_{R^{i}}\notin\mathbb Z$, which implies $0<x_{R^{i}}<c_{R^{i}}$. Thus, $i$ is unsatiated. Moreover, $c^r_{i}=0$, so $|R^{i}|\ge2$, and \Cref{prop:Lindahlprops}(7) gives $u_{i}(x)>b_{i}$. Since $i$ was arbitrary, the claim follows.
		
	To apply \Cref{thm:reduceLindahlutils}, we show that $u_{i^*}(x)=\lfloor u_{i^*}(x)\rfloor$ for some voter $i^*$. Let $v=u(x)-u^0$ and $u^r=\lfloor v\rfloor$.
	By \Cref{thm:pinnable-patchable}, every fractional witness $z$ for $u^r$ in $(c^r,k^r)$ satisfies $u_{i^*}(z)=u_{i^*}^r$.
	By \ref{cond:decomp-feasibility}, there exists such a witness $z$ with $u(z)\geq v$ and, consequently, $u_{i^*}^r=u_{i^*}(z)\geq v_{i^*}\geq u_{i^*}^r$.
	Since $u_{i^*}^0$ and $u_{i^*}^r$ are integral, it follows that
	\[
		u_{i^*}(x) =u_{i^*}^0+v_{i^*} =u_{i^*}^0+u_{i^*}^r \in\mathbb Z.
	\]
	Thus, $u_{i^*}(x)=\lfloor u_{i^*}(x)\rfloor$, as required.
	
	Hence, we can apply \Cref{thm:reduceLindahlutils} and conclude that the existing witness for integral-committee feasibility of $u-e_{i^*}$ lies in the core.
\end{proof}

\subsection{Core existence for $n=7$}\label{sec:core-existence-7}

Unfortunately, there are instances with seven voters for which no suitable voter in the sense of \Cref{thm:pinnable-patchable} can be found.
To see this, consider the ``Fano-plane instance'' shown in \Cref{fig:fano-plane} with 7 equally weighted voters, 7 candidate types (each with supply at least $1$), and $k = 5$. By symmetry, the Lindahl equilibrium of this instance assigns $x_R = \frac{5}{7}$ to each type $R \in \mathcal{R}_c$. Because each voter is part of exactly three types, $u_i = \frac{15}{7} \approx 2.14$ for each $i \in N$. Thus, no voter $i^*$ satisfies $u_{i^*}(x)=\lfloor u_{i^*}(x)\rfloor$ in this instance. However, every voter is unsatiated in the Lindahl equilibrium. In particular, we can invoke \Cref{thm:reduceLindahlutilsfor7} to show that the Fano-plane instance has a non-empty core: in the Lindahl equilibrium, we have $\beta_i = \frac13$ for all $i \in N$, and hence the key quantity in the statement of \Cref{thm:reduceLindahlutilsfor7} satisfies $\sum_{i \in N} \beta_i (u_i(x) - \lfloor u_i(x) \rfloor) + \beta_{i^*} \approx 7 \cdot \frac13 \cdot 0.14 + \frac 13 = 0.66 < 1$, as required by that theorem. Inspired by this, we will adapt our arguments to show that for $n = 7$, any hole satisfies sufficient conditions to apply either \Cref{thm:reduceLindahlutils} or \Cref{thm:reduceLindahlutilsfor7}.

\begin{figure}[thb]
	\centering
		\begin{tikzpicture}[
		scale=0.9,
		voter/.style={circle,draw,fill=white,inner sep=1.3pt,minimum size=6mm}
		]
		\coordinate (p1) at (0,2);
		\coordinate (p2) at (0.866,0.5);
		\coordinate (p3) at (1.732,-1);
		\coordinate (p4) at (0,-1);
		\coordinate (p5) at (-1.732,-1);
		\coordinate (p6) at (-0.866,0.5);
		\coordinate (p7) at (0,0);
		\draw[very thick,violet!80] (p1) -- (p7) -- (p4);
		\draw[very thick,teal!80!black] (p3) -- (p7) -- (p6);
		\draw[very thick,brown!80!black] (p5) -- (p7) -- (p2);
		\draw[very thick, double=orange!80!black, draw=white, double distance=1.2pt] (p7) circle (1);
		\draw[very thick,blue!75] (p1) -- (p2) -- (p3);
		\draw[very thick,red!75] (p3) -- (p4) -- (p5);
		\draw[very thick,green!60!black] (p5) -- (p6) -- (p1);
		\foreach \i in {1,...,7}{
			\node[voter] at (p\i) {$\i$};
		}
	\end{tikzpicture}
	\caption{The Fano-plane (also known as $\mathrm{PG}(2,2)$) instance for $n=7$: the seven colored triples are the candidate types $\{123,345,156,147,367,257,246\}$.}
	\label{fig:fano-plane}
\end{figure}

Recall that we still work with the decomposition from \Cref{thm:decomposition}. This time, we need to begin by considering residuals having a zero coordinate or a voter who approves every available type. Previously, we ruled out these cases using \Cref{lem:voter-deletion-bounds,lem:Rcbounds}. However, both lemmas rely on normality of the committee election monoid for $n-1$ voters and therefore do not apply when $n=7$. We thus handle these boundary cases separately first.

\begin{lemma}\label{lem:n7boundary}
	Let $n=7$ and $(x,p)$ be a Lindahl equilibrium for an instance $(c,k,b)$ with $\sum_{R\in\mathcal R}x_R=k$.
	Set $u=\lfloor u(x)\rfloor$, suppose that $(c,k,u)$ is a hole and let
	$
		(c,k,u)=(c^0,k^0,u^0)+(c^r,k^r,u^r)
	$
	be a decomposition satisfying \ref{cond:decomp-base}, \ref{cond:decomp-feasibility}, and \ref{cond:decomp-support}. 
	Suppose we are in one of the following two cases:
	\begin{enumerate}
		\renewcommand{\theenumi}{(\arabic{enumi})}%
		\renewcommand{\labelenumi}{\theenumi}%
		\item there exists a voter $i^* \in N$ with $u_{i^*}^r = 0$, or 
		\item \ref{cond:decomp-supply}, \ref{cond:decomp-bounds}, and $i^*\in\bigcap_{R\in\mathcal R_c^r}R$ hold.
	\end{enumerate}
	Then the core of $(c,k,b)$ is non-empty.
\end{lemma}

\begin{proof}
	Let $v:=u(x)-u^0$, so that $u^r=\lfloor v\rfloor$, and let $x^r$ be a fractional witness for $v$ given by \ref{cond:decomp-feasibility}.

	Delete voter $i^*$ from the residual instance: replace every type $R$ by $R\setminus\{i^*\}$, merge types that become identical, and discard the empty type. Denote the resulting six-voter instance by $(\tilde c,k^r)$. Let $\tilde u$ be obtained from $u^r$ by deleting coordinate $i^*$, and let $\bar x$ be the projection of $x^r$. Since $x^r$ witnesses $v$, $\bar x$ witnesses $\tilde u$.

	\medskip
	\noindent\emph{Step 1: Reduction to a six-voter hole.}
	Any integral witness for $\tilde u$ can be lifted to the original residual instance. In case (1) of the lemma statement, we have $u^r_{i^*}=0$, and then the lifted committee already realizes $u^r$. In case (2), $i^*$ belongs to every residual type, so every selected residual candidate gives $i^*$ one unit of utility. If necessary, we may further add unused residual candidates until $i^*$ reaches utility $u^r_{i^*}$. This is possible since $u^r_{i^*}\le k^r$ and \ref{cond:decomp-bounds} guarantees sufficient residual supply.

	Consequently, $(\tilde c,k^r,\tilde u)$ must be a hole, since otherwise an integral witness for $\tilde u$ would lift to an integral witness for $u^r$, which would mean that $(c,k,u)$ is not a hole.

	If $|\bar x|<k^r$, increase its coordinates subject to the supply constraints until its size is $k^r$, and call the resulting witness $\tilde x$. This must be possible: otherwise the projected supply would be exhausted, and selecting
	all of it would give an integral witness for $\tilde u$. Moreover, every coordinate of $\tilde u$ is positive; otherwise \Cref{lem:zero-utility} and \Cref{thm:monoidnormal} would again imply that $\tilde u$ is integral-committee-feasible.

	\medskip
	\noindent\emph{Step 2: Apply the six-voter classification.}
	Apply \Cref{thm:decomposition} to $(\tilde c,k^r,\tilde u)$ with witness $\tilde x$ to get a decomposition $(\tilde c,k^r,\tilde u) = (\hat c^0,\hat k^0,\hat u^0) + (\hat c^r,\hat k^r,\hat u^r)$. Its residual cannot have a zero utility coordinate, again by \Cref{lem:zero-utility} and \Cref{thm:monoidnormal} (normality for $n - 1 = 5$). By the six-voter analysis above, this residual is therefore one of the $50$ candidate minimal holes from \Cref{thm:pinnable-patchable}.
	By \ref{cond:decomp-feasibility}, there exists a fractional witness $\hat x^r$ for $u(\tilde x)-\hat u^0$. Since $\hat u^r=\tilde u-\hat u^0$, \Cref{thm:pinnable-patchable} gives, for every $j\ne i^*$,
	\[
		\hat u_j^r =u_j(\hat x^r) \ge u_j(\tilde x)-\hat u_j^0 \ge \tilde u_j-\hat u_j^0 =\hat u_j^r.
	\]
	Hence, $u_j(\tilde x)=\tilde u_j$. Since $\tilde x$ dominates the projection $\bar x$ of $x^r$, while $x^r$ witnesses $v$, we obtain
	\[
		\tilde u_j =u_j(\tilde x) \ge v_j \ge \lfloor v_j\rfloor =\tilde u_j.
	\]
	Thus, $v_j=\tilde u_j$ is integral, and therefore, $u_j(x)=u_j^0+v_j\in\mathbb Z$ for every $j\ne i^*$.

	Moreover, \Cref{thm:pinnable-patchable} implies that $\hat u^r-e_j$ is integral-committee-feasible. Combining a corresponding witness with the base part of the six-voter decomposition gives an integral witness for $\tilde u-e_j$. By the same lifting argument as above, this yields an integral witness for $u^r-e_j$ in $(c^r,k^r)$, and adding the original base part gives an integral witness for $u-e_j=u^0+(u^r-e_j)$ in $(c,k)$.

	\medskip
	\noindent\emph{Step 3: Apply \Cref{thm:reduceLindahlutils}.}
	Fix any $j\ne i^*$. By \Cref{lem:Rcbounds}, the six-voter minimal hole contains a residual type that contains $j$ and another voter. This type originates from some $R\in\mathcal R_c^r$ containing the same two voters.
	By \ref{cond:decomp-support}, $x_R\notin\mathbb Z$, and hence $0<x_R<c_R$. Thus, $j$ is unsatiated. Since $|R|\ge2$, \Cref{prop:Lindahlprops}(7) gives $u_j(x)>b_j$.

	We have already shown that $u_j(x)$ is integral and that $\lfloor u(x)\rfloor-e_j=u-e_j$ is integral-committee-feasible. Hence, \Cref{thm:reduceLindahlutils} yields an integral committee in the core.
\end{proof}

Hence, we can focus on decompositions satisfying \ref{cond:decomp-base}--\ref{cond:decomp-bounds} for which $\bigcap_{R \in \mathcal{R}_c^r}R=\emptyset$. By \ref{cond:decomp-bounds}, $u_i^r\ge1$ for every voter. Since the witness in \ref{cond:decomp-witness} has size exactly $k^r$ and some residual type excludes each voter, we also have $u_i^r\le k^r-1$. As in the proof of \Cref{lem:Rcbounds}, it follows that $c^r_N=0$, $c^r_i=0$, and $2 \le d^r_i \le \lvert \mathcal{R}^r_c \rvert -1$ for all $i \in N$. Thus, every remaining minimal hole is a \emph{candidate minimal hole} and additionally satisfies

\begin{enumerate}
	\renewcommand{\theenumi}{(R\arabic{enumi})}%
	\renewcommand{\labelenumi}{\theenumi}%
	\setcounter{enumi}{5}
	\item\label{cond:intersectionempty} $\bigcap_{R \in \mathcal{R}_c^r}R=\emptyset$.
\end{enumerate}

Candidate minimal holes satisfying \ref{cond:antichain}-\ref{cond:intersectionempty} constitute new structures compared to the six-voter instances.
Like in \Cref{sec:core_for_six}, we can enumerate all instance-utility pairs that satisfy these conditions. There are 54,985 such holes (up to relabeling) on 10,292 antichain classes, compared to only $50$ candidate minimal holes for $n=6$.

We cut down on this large number of holes by observing that minimal holes coming from the decomposition of a Lindahl equilibrium must satisfy a further property that involves the $\beta_i$'s of the Lindahl equilibrium (see the Lindahl definition in \Cref{sec:lindahl}).
Let $(x,p)$ be a Lindahl equilibrium for the instance-utility pair $(c,k,u)$ that we decomposed.
As we have shown (and assuming we did not already finish by invoking the conditions of \Cref{lem:n7boundary}), the residual part satisfies
\ref{cond:antichain}-\ref{cond:intersectionempty}. 
Now consider a voter $i \in N$. By \ref{cond:decomp-witness} and \ref{cond:decomp-bounds}, there exists some $R \in \mathcal{R}_c^r$ with $i \in R$, and by \ref{cond:decomp-support} we have $0<x_R<c_R$. Thus, voter $i$ is unsatiated, and we can write $\beta_i = p_{iR}$ as in \Cref{sec:lindahl}. Since $b_i>0$, budget exhaustion and \Cref{prop:Lindahlprops}(6) give $0<b_i\le\beta_i u_i(x)$, and hence $\beta_i>0$. Since $i$ was arbitrary, we have $\beta_i > 0$ for all $i \in N$. Moreover, \ref{cond:decomp-support} and the Lindahl equilibrium conditions imply that $\sum_{i \in R}\beta_i=1$ for all $R \in \mathcal{R}_c^r$.

During the computer enumeration of candidate minimal holes, we do not know the original instance or its Lindahl equilibrium, so we do not have access to the $\beta_i$'s. However, the above argument implies that the candidate minimal hole satisfies the following condition:
\begin{enumerate}
	\renewcommand{\theenumi}{(R\arabic{enumi})}%
	\renewcommand{\labelenumi}{\theenumi}%
	\setcounter{enumi}{6}
	\item\label{cond:lindahlbetas}
	There exists $\beta \in (0,1]^n$ such that $\sum_{i \in R}\beta_i=1$ for all $R \in \mathcal{R}_c^r$.
\end{enumerate}
This property further restricts the number of minimal holes that need to be considered to 21,818.%

Now we are ready to describe the result of our computer enumeration. For this, we use the notation $\Delta^{\mathcal{R}_c^r}:=\{\beta \in [0,1]^n \colon \sum_{i \in R}\beta_i=1 \text{ for all } R \in \mathcal{R}_c^r\}$ to describe potential $\beta$-vectors associated with the original instance; note that we allow $\beta_i = 0$ in the definition of $\Delta^{\mathcal{R}_c^r}$ but not in condition \ref{cond:lindahlbetas}.

Reproduction scripts for these computations are available at \url{https://github.com/DominikPeters/corefewvoters}.

\begin{fact}[Checked computationally]\label{thm:n7minimalholes}
	For $n=7$, there exist (up to relabeling) 21,818 candidate minimal holes $(c^r,k^r,u^r)$ satisfying \ref{cond:antichain}--\ref{cond:lindahlbetas}, of which
	\begin{enumerate}
		\renewcommand{\theenumi}{(\arabic{enumi})}%
		\renewcommand{\labelenumi}{\theenumi}%
		\item 21,520 have a voter $i^*$ such that $u^r-e_{i^*}$ is integral-committee-feasible in $(c^r,k^r)$ and $u_{i^*}(x)=u_{i^*}^r$ for all witnesses $x$ of fractional-committee feasibility of $u^r$ in $(c^r,k^r)$;
		\item 298 are such that for every voter $i^* \in N$,
		\begin{itemize}
			\item $u^r-e_{i^*}$ is integral-committee-feasible in $(c^r,k^r)$, and
			\item $\max_{\beta \in \Delta^{\mathcal{R}_c^r}} \beta_{i^*}+k^r-\sum_{i\in N}\beta_i u_i^r \le \frac{8}{9}<1$.
		\end{itemize}
	\end{enumerate} 
\end{fact}

To obtain this result, we proceeded as follows. Integral-committee feasibility for the reduced utility vectors $u^r-e_{i^*}$ was checked by exhaustively enumerating all integral-committee-feasible utility vectors. For (1), the statements $u_{i^*}(x)=u^r_{i^*}$ were checked using Z3, which uses exact rational arithmetic. Z3 reported the unsatisfiability of a formula looking for a witness $x$ with $u_{i^*}(x)>u^r_{i^*}$. For (2), the second condition involving a maximization over $\Delta^{\mathcal{R}_c^r}$ was also verified using Z3. Because this condition is phrased as an optimization over $\beta$, it applies in particular to the Lindahl-induced $\beta$, which is unknown at computation time, and will allow us to apply \Cref{thm:reduceLindahlutilsfor7}.

Combining the computational results and our insights for $n \le 6$, we can now prove core existence for arbitrary instances $(c,k,b)$ with $n=7$.

\begin{theorem}\label{thm:n7coreexists}
	For $n=7$, the core is non-empty for every instance $(c,k,b)$.
\end{theorem}

\begin{proof}
	Let $(c,k,b)$ be an arbitrary instance with seven voters. If $\sum_{R \in \mathcal{R}}c_R \le k$, selecting all candidates yields an integral core committee. So we can assume $\sum_{R \in \mathcal{R}}c_R > k$ which implies that there exists a Lindahl equilibrium $(x,p)$ of total size $k$ by \Cref{thm:fraccoreexists}.
	If $\lfloor u(x) \rfloor$ is integral-committee-feasible in $(c,k)$, the core is non-empty by \Cref{thm:roundfracutils}.
	
	Otherwise, let $u=\lfloor u(x)\rfloor$ and start the decomposition by moving, for every $R\in\mathcal R$, $\lfloor x_R\rfloor$ copies of $X_R$ and $c_R-\lceil x_R\rceil$ copies of $Z_R$ to the base part. Set $x_R^r=x_R-\lfloor x_R\rfloor$. Then $c_R^r=1$ exactly when $x_R\notin\mathbb Z$, $u^r=\lfloor u(x^r)\rfloor$, and $\sum_{R\in\mathcal R_c^r}x_R^r=k^r$. Thus, \ref{cond:decomp-base}, \ref{cond:decomp-feasibility}, \ref{cond:decomp-support}, and \ref{cond:decomp-witness} hold for this decomposition. (If we had instead invoked the standard decomposition theorem (\Cref{thm:decomposition}), it would not have given us \ref{cond:decomp-support} and \ref{cond:decomp-witness} when $u^r$ has a zero coordinate.)

	Starting from this decomposition, apply \Cref{lem:minimal-reduction01,lem:Rcle5}. These lemmas preserve \ref{cond:decomp-base}, \ref{cond:decomp-feasibility}, \ref{cond:decomp-support}, and \ref{cond:decomp-witness} and, unless a coordinate of $u^r$ becomes zero, additionally establish \ref{cond:decomp-supply} and $\lvert\mathcal R_c^r\rvert\le n = 7$. If a coordinate $u^r_{i^*}$ becomes zero, the core is non-empty by \Cref{lem:n7boundary} and the argument following it. Otherwise, \Cref{lem:k-1<Rc} gives \ref{cond:decomp-bounds}; putting these together we can deduce that $(c^r,k^r,u^r)$ satisfies \ref{cond:antichain}--\ref{cond:util} and, thus, is a candidate minimal hole. Next, if some voter belongs to every residual type, the core is again non-empty by \Cref{lem:n7boundary} and the argument following it. We may therefore assume that $(c^r,k^r,u^r)$ is a candidate minimal hole satisfying \ref{cond:antichain}--\ref{cond:intersectionempty}.

	For every voter $i\in N$, \ref{cond:decomp-witness} and \ref{cond:decomp-bounds} give a type $R\in\mathcal R_c^r$ with $i\in R$, and \ref{cond:decomp-support} gives $0<x_R<c_R$. Thus, every voter is unsatiated, and the Lindahl coefficients $\beta_i^*$ are well-defined. Since $b_i>0$, budget exhaustion and \Cref{prop:Lindahlprops}(6) imply that $\beta_i^*>0$. Moreover, \ref{cond:decomp-support} and the Lindahl equilibrium conditions give $\sum_{i\in R}\beta_i^*=1$ for every $R\in\mathcal R_c^r$. Hence, $\beta^*\in\Delta^{\mathcal R_c^r}$ and the residual also satisfies \ref{cond:lindahlbetas}. Thus, the candidate minimal hole is covered by the enumeration of \Cref{thm:n7minimalholes}, and must belong to one of the two classes described there.

	For the first class in \Cref{thm:n7minimalholes}, choose the voter $i^*$ given there. The witness $x^r$ from \ref{cond:decomp-witness} witnesses both $u(x)-u^0$ and $u^r$. Hence, the equality $u_{i^*}(x^r)=u_{i^*}^r$ from \Cref{thm:n7minimalholes} implies
	\[
		u_{i^*}^r=u_{i^*}(x^r)\ge u_{i^*}(x)-u_{i^*}^0\ge u_{i^*}^r.
	\]
	Thus, $u_{i^*}(x)$ is integral. Moreover, $u^r-e_{i^*}$ is integral-committee-feasible. Combining its witness with the integral committee for the base part from \ref{cond:decomp-base} gives an integral witness for $u-e_{i^*}$ in the original instance. Finally, \ref{cond:decomp-support}, \ref{cond:decomp-witness}, and \ref{cond:decomp-bounds} give a type $R\ni i^*$ with $0<x_R<c_R$. Thus, $i^*$ is unsatiated. Since $c^r_{i^*}=0$, we have $|R|\ge2$, and \Cref{prop:Lindahlprops}(7) gives $u_{i^*}(x)>b_{i^*}$. By \Cref{thm:reduceLindahlutils}, the core is non-empty.

	For the second class in \Cref{thm:n7minimalholes}, fix any voter $i^*\in N$. By \Cref{thm:n7minimalholes}, $u^r-e_{i^*}$ is integral-committee-feasible in $(c^r,k^r)$, and combining its witness with the integral committee for the base part from \ref{cond:decomp-base} shows that $u-e_{i^*}$ is integral-committee-feasible in $(c,k)$.

	Let $v=u(x)-u^0$ and $\phi_i=u_i(x)-\lfloor u_i(x)\rfloor$. Since $u^0$ is integral and $u^r=\lfloor v\rfloor$, we have $\phi_i=v_i-u_i^r$ for every $i\in N$. Moreover, the witness $x^r$ from \ref{cond:decomp-witness} satisfies $v\le u(x^r)$. Hence,
	\[
		\beta^*_{i^*}+\sum_{i\in N}\beta_i^*\phi_i
		\le \beta^*_{i^*}+\langle\beta^*,u(x^r)-u^r\rangle
		=\beta^*_{i^*}+k^r-\langle\beta^*,u^r\rangle
		\le\frac89<1,
	\]
	where the equality uses $\beta^*\in\Delta^{\mathcal R_c^r}$ and $\sum_{R\in\mathcal R_c^r}x_R^r=k^r$, and the final inequality follows from \Cref{thm:n7minimalholes}. Finally, $\lfloor u_{i^*}(x)\rfloor=u_{i^*}^0+u_{i^*}^r\ge1$ by \ref{cond:decomp-base} and \ref{cond:decomp-bounds}. Thus, \Cref{thm:reduceLindahlutilsfor7} implies that the core is non-empty.
\end{proof}

\begin{figure}[thb]
	\centering
	\begin{tikzpicture}[
		scale=0.95,
		line cap=round,
		line join=round,
		pgline/.style={very thick,draw=white,double=#1,double distance=1.2pt},
		voter/.style={circle,draw,fill=white,inner sep=1.3pt,minimum size=6mm}
		]
		
		\coordinate (p1) at (-1,1);
		\coordinate (p2) at (0,1);
		\coordinate (p3) at (1,1);
		\coordinate (p4) at (-1,0);
		\coordinate (p5) at (0,0);
		\coordinate (p6) at (1,0);
		\coordinate (p7) at (-1,-1);
		\coordinate (p8) at (0,-1);
		\coordinate (p9) at (1,-1);
		
		\coordinate (p10) at (45:3);
		\coordinate (p11) at (0:3);
		\coordinate (p12) at (-45:3);
		\coordinate (p13) at (-90:3);
		
		\draw[pgline=blue!75] (p1) -- (p2) -- (p3) to[out=0,in=135] (p11);
		\draw[pgline=blue!75] (p4) -- (p5) -- (p6) -- (p11);
		\draw[pgline=blue!75] (p7) -- (p8) -- (p9) to[out=0,in=-135] (p11);
		
		\draw[pgline=green!60!black] (p7) to[out=135,in=135,distance=70] (p2) -- (p6) to[out=-45,in=90] (p12);
		\draw[pgline=green!60!black] (p1) -- (p5) -- (p9) -- (p12);
		\draw[pgline=green!60!black] (p3) to[out=135,in=135,distance=70] (p4) -- (p8) to[out=-45,in=180] (p12);
		
		\draw[pgline=red!75] (p1) -- (p4) -- (p7) to[out=-90,in=135] (p13);
		\draw[pgline=red!75] (p2) -- (p5) -- (p8) -- (p13);
		\draw[pgline=red!75] (p3) -- (p6) -- (p9) to[out=-90,in=45] (p13);
		
		\draw[pgline=violet!80] (p1) to[out=-135,in=-135,distance=70] (p8) -- (p6) to[out=45,in=-90] (p10);
		\draw[pgline=violet!80] (p7) -- (p5) -- (p3) -- (p10);
		\draw[pgline=violet!80] (p9) to[out=-135,in=-135,distance=70] (p4) -- (p2) to[out=45,in=180] (p10);
		
		\draw[pgline=orange!80!black]
		(p10) to[bend left=17] (p11)
		(p11) to[bend left=17] (p12)
		(p12) to[bend left=17] (p13);
		
		\foreach \i in {1,...,13}{
			\node[voter] at (p\i) {$\i$};
		}
	\end{tikzpicture}
	\caption{The projective plane $\mathrm{PG}(2,3)$, describing an instance with $n=13$ where the method we used for $n=7$ fails.}
	\label{fig:pg23}
\end{figure}

As far as we know, the properties of candidate minimal holes that we identified in \Cref{thm:n7minimalholes} could continue to work for larger $n$. However, there is an example with $n = 13$ where there is a candidate minimal hole that satisfies neither of our properties. The example is specified by the projective plane $\mathrm{PG}(2,3)$ (see \Cref{fig:pg23}), which encodes an instance with 13 voters with equal budgets and 13 candidate types (each line in the drawing is a candidate type), such that each type is approved by 4 voters and each voter is part of 4 types. Taking $k = 11$, the Lindahl equilibrium assigns $x_R = \frac{11}{13}$ to each type $R$ (by symmetry), giving each voter a utility of $u_i = \frac{44}{13} \approx 3.385$. In addition, we have $\beta_i = \frac{1}{4}$ in the Lindahl equilibrium. But then we have $\beta_{i^*}+k-\sum_{i\in N}\beta_i\lfloor u_i(x)\rfloor = \frac{1}{4} + 11 - 13 \cdot \frac{1}{4} \cdot 3 = 1.5$, which is not strictly less than 1. Thus, the analogue of \Cref{thm:n7minimalholes} fails for $n = 13$ (and perhaps earlier).

\section{Computing committees in the core for $n \le 5$}\label{sec:corecomputation}

In this section, we provide an efficient algorithm for computing committees in the core for $n \le 5$. Using the translation of \Cref{thm:votertypes}, the same algorithm will be applicable to arbitrary instances with up to five voter types.
To do so, we have to solve several algorithmic challenges. In particular, given a fractional-committee-feasible integer utility vector, we need to be able to compute an integral committee achieving that vector. We would also ideally want to find a fractional core committee that we can use as a starting point for the rounding process; however, there is no known efficient algorithm for finding such a fractional committee exactly, so we will instead show that a sufficiently close approximation suffices. 
Finally, we will aim for our algorithm to produce a committee that is not just in the core but that is also Pareto-optimal.%
\footnote{\label{fn:pareto}An integral committee $x$ is \emph{Pareto optimal} if there is no other integral committee $y$ such that $u_i(y) \ge u_i(x)$ for all $i \in N$ with strict inequality for at least one $i \in N$. Note that the core with $S = N$ alone only implies \emph{weak} Pareto optimality: the core only excludes Pareto improvements that strictly increase \emph{every} voter's utility.}
We will tackle each of these issues in turn and also explain the difficulties of extending our algorithms for $n>5$.

\subsection{Implementing fractional-committee-feasible utility vectors}

Given a fractional committee $x \in \mathcal{P}$ for an instance $(c,k,b)$ with $n \le 5$, we want to round $x$ to an integral committee $x' \in \mathcal{W}$ such that $u_i(x') \ge \lfloor u_i(x)\rfloor$ for all $i \in N$.
\Cref{alg:roundfraccommittees} achieves this by implementing the rounding procedure described in the proof of \Cref{prop:roundingcommittees}. It solves at most $\lvert \mathcal{R}_c \rvert$ linear programs while guaranteeing $\lfloor x_R \rfloor \le x'_R \le \lceil x_R \rceil$ for all $R \in \mathcal{R}$.

\bigskip
\begin{algorithm}[H]\LinesNumbered
	\caption{Round fractional committees}\label{alg:roundfraccommittees}
	\KwIn{Instance $(c,k,b)$ with $n \le 5$, and a fractional committee $x \in \mathcal{P}$.}
	\KwOut{Integral committee $x' \in \mathcal{W}$ with $u_i(x') \ge \lfloor u_i(x)\rfloor$ for all $i \in N$ and $\lfloor x_R \rfloor \le x'_R \le \lceil x_R \rceil$ for all $R \in \mathcal{R}$.}
	
	$x' \leftarrow \lfloor x \rfloor$ \tcp{initial integral (sub-)committee}

	$\widetilde{\mathcal{R}} \leftarrow \{R \in \mathcal{R}_c \colon x_R>\lfloor x_R \rfloor\}$ \tcp{remaining candidate types to consider}
	
	\While{$u_i(x')<\lfloor u_i(x)\rfloor$ for some $i \in N$}
	{$R^* \leftarrow \text{pick arbitrary } R \in \widetilde{\mathcal{R}}$\;
		\If{the following linear program has a feasible solution $z$
		\[
		\begin{aligned}
			\sum_{\mathclap{R \in \widetilde{\mathcal{R}} : R \ni i}} z_R &\ge %
			\lfloor u_i(x) \rfloor-u_i(x')
			\quad \textup{for all $i \in N$}, \\
			\sum_{R \in \widetilde{\mathcal{R}}} z_R &\le k-\sum_{R \in \mathcal{R}} x'_R, \quad
			z_{R^*}=1, \quad
			0 \le  z_R \le 1 \quad \textup{for all $R \in \widetilde{\mathcal{R}} \setminus \{R^*\}$}.
		\end{aligned}
		\]}
		{
			$x'_{R^*}\leftarrow x'_{R^*}+1$\;
		}
	$\widetilde{\mathcal{R}} \leftarrow \widetilde{\mathcal{R}} \setminus \{R^*\}$\;
	}	
	\Return $x'$\;	
\end{algorithm}
\bigskip

\Cref{prop:roundingcommittees} promises that there exists a rounded committee $x'$ with $x'_R \ge \lfloor x_R \rfloor$ for all $R \in \mathcal{R}$, and so \Cref{alg:roundfraccommittees} starts by locking in $\lfloor x \rfloor$ as part of the final integral committee $x'$. On top of that, \Cref{prop:roundingcommittees} also ensures $x'_R \le \lceil x_R \rceil$ for all $R \in \mathcal{R}$. Therefore, our task boils down to deciding which candidate types $R$ to round up ($x'_R = \lceil x_R \rceil$) and which ones to leave rounded down ($x'_R = \lfloor x_R \rfloor$).
The algorithm goes through the types $R \in \mathcal{R}$ one-by-one, in each step making exactly this decision. 

Observe that we can skip types with $x_R \in \mathbb Z$ since rounding them up is the same as rounding them down (line 2). 
When the while-loop reaches a type $R^*$, we solve a linear program to decide if it is safe to round up this specific type.
Note that our remaining task at this point is to pick at most $k -\sum_{R \in \mathcal{R}} x'_R$ types to round up while guaranteeing that the final output committee satisfies $u_i(x') \ge \lfloor u_i(x) \rfloor$ for all $i \in N$. 
Given the candidates we have placed into $x'$ already, the candidates we will place into the committee in the future will need to give each $i \in N$ a utility of at least $\lfloor u_i(x) \rfloor - u_i(x')$. 
The linear program essentially asks whether there is a way to achieve this when deciding to round up $R^*$ (constraint $z_{R^*} = 1$). Notice that the linear program can be seen as a check for whether the integer utility vector $\lfloor u(x) \rfloor - u(x')$ is fractional-committee-feasible in a reduced instance obtained after fixing the previous decisions and fixing $R^*$ to be rounded up. By normality, this is equivalent to integral-committee feasibility.

Thus, if the LP is feasible, then there exists a valid way to round the remaining types up or down that achieves our utility requirement while rounding up $R^*$. So the algorithm rounds up $R^*$ and proceeds to the next type. If not, then there does not exist such a valid way with $R^*$ rounded up. But from \Cref{prop:roundingcommittees}, we know that there does exist a valid rounding; we can deduce that in this rounding, $R^*$ is rounded down. Hence the algorithm does so and proceeds to the next type.

\begin{remark}
	The above rounding technique fails for $n=6$. Consider a six-voter instance with $k=4$, $\mathcal R_c = \{12,13,23,45,46,56,14\}$, and $c_R =1$ for each $R \in \mathcal{R}_c$. Observe that this instance corresponds to the example in \Cref{fig:two-triangles-intro} with one additional edge between voters $1$ and $4$. The utility vector $u = (2,1,1,2,1,1)$ is fractional-committee-feasible with witness $x$, where $x_R = 1/2$ for all $R \in \mathcal{R}_c \setminus \{14\}$ and $x_{14} = 1$.
	Rounding down $x$ results in $x_{14} = 1$ and $x_R = 0$ for all $R \in \mathcal{R}_c \setminus \{14\}$.
	Observe that we ended up in the minimal hole from \Cref{fig:two-triangles-intro}. Hence, no witness for integral-committee-feasibility can have $x_{14}=1$. However, the committee consisting of $\{12,13,45,46\}$ is actually a witness for integral-committee feasibility of $u$ in $(c,k)$.
\end{remark}

\subsection{Rounding an approximate fractional core committee}
\label{sec:round-fractional-core}

Our main algorithm for computing a core committee begins by constructing an $\varepsilon$-approximate fractional core committee. We then round down the obtained utilities. Using normality of the committee election monoid with at most five voters, the rounded-down utilities can be achieved by an integral committee and this committee will be in the core. We can find such a committee using \Cref{alg:roundfraccommittees} from the previous section. However, the algorithm crucially depends on the rounded-down utilities being fractional-committee-feasible so that we can apply \Cref{thm:monoidnormal}. But this is not obvious as we start with an approximate fractional committee (which may overspend the budget $k$). Using linear programming arguments, we show that for an appropriately chosen $\varepsilon$, fractional-committee feasibility can be guaranteed.

It is a long-standing open problem to find an efficient algorithm that computes a fractional core committee. Recently, \citet{KrPe25b} found a polynomial-time algorithm that computes an \emph{approximate} fractional core committee. Their method works by approximately optimizing a convex program whose optimum is a Lindahl equilibrium. Their notion of approximation allows the fractional committee both to overshoot the available budget $k$ by $\varepsilon$ and to admit core deviations as long as they give at most an additive $\varepsilon$ of extra utility to coalition members. For our purposes, it is more convenient not to relax the notion of core deviations and to allow only budget overshooting. The result of \citet{KrPe25b} can easily be adapted to achieve this. We say that a vector $x \in \R_{\ge 0}^{\lvert \mathcal{R}\rvert}$ with $x_R \le c_R$ for all $R \in \mathcal{R}$ satisfies the \emph{fractional core condition} if there does not exist a non-empty $S \subseteq N$ and fractional committee $y \in \mathcal{P}_S$ such that $u_i(y) > u_i(x)$ for all $i \in S$.

\begin{restatable}[based on \citealp{KrPe25b}]{theorem}{lindahlEfficientComputation}\label{thm:Lindahl_efficient_computation}
	For a given instance $(c,k,b)$ and $\varepsilon > 0$, there exists an algorithm running in time polynomial in the problem size and in $\log\frac{1}{\varepsilon}$ that computes a vector $x \in \R_{\ge 0}^{\lvert \mathcal{R}\rvert}$ with $\sum_{R \in \mathcal{R}} x_R \le k + \varepsilon$ that satisfies the fractional core condition.
\end{restatable}
We give the proof in \Cref{app:omitted-corecomputation}.

Next, we will show that there is a value of $\varepsilon$ small enough that the budget overshoot is not an issue for our rounding task (it turns out that $\varepsilon < \frac{1}{60}$ suffices).
For this, we need some notation.
Given an integer utility vector $u \in \Z^n$, let $\tau_c(u)$ denote the minimum fractional committee size needed to realize $u$, that is,
\begin{equation}\label{eq:lp_min_feasibility}
    \tau_c(u) := \min \left\{ \sum_{R \in \mathcal R} z_R : z \in \R_{\ge 0}^{\lvert \mathcal R \rvert},\; z_R \le c_R \ \forall R \in \mathcal R,\; \sum_{R \ni i} z_R \ge u_i \ \forall i \in N \right\}.
\end{equation}
Moreover, for $n \ge 1$, define
\[
    L_n := \operatorname{lcm}\left\{ \lvert \det(B)\rvert : B \text{ is an invertible binary square matrix of order at most } n \right\}
\]
as the least common multiple ($\operatorname{lcm}$) of the set of absolute values of the determinants of these matrices. Combining the known maximum determinants for binary matrices of order $n \le 5$ \citep{Ehli64a} with the fact that every smaller determinant value is attainable \citep{Craig90a}, we obtain the explicit values of $L_n$ shown in \Cref{tab:Ln}; in particular, $L_5 = 60$.

\begin{table}
	\centering
\begin{tabular}{c|ccccc}
	$n$ & 1 & 2 & 3 & 4 & 5 \\
	\hline
	$L_n$ & 1 & 1 & 2 & 6 & 60
\end{tabular}
\caption{Least common multiples of the set of possible determinants of binary matrices of order at most $n$.}
\label{tab:Ln}
\end{table}
Note that according to its definition, $\tau_c(u)$ is the solution of a linear program and is therefore rational. By reasoning about basic feasible solutions of this linear program and using Cramer's rule, we can give a bound on the denominator of this rational number. In particular, for $n = 5$, our bound implies that $\tau_c(u)$ is a multiple of $\frac{1}{60} \approx 0.0166$. Now suppose that we can establish, for example, that $\tau_c(u) \le k + 0.01$. Then, thanks to the denominator bound, we can deduce that in fact $\tau_c(u) \le k$. The formal statement is as follows; we give the proof in \Cref{app:omitted-corecomputation}.

\begin{restatable}{lemma}{tauDenominator}\label{lem:tau-denominator}
    For every supply vector $c \in \Z_{\ge 0}^{|\mathcal R|}$ and every integer utility vector $u \in \Z^n$ with $\tau_c(u)<\infty$, we have
    \[
        \tau_c(u)\in \frac{1}{L_n}\Z.
    \]
    In particular, if $0<\delta<\frac{1}{L_n}$ and $\tau_c(u) \leq k'+\delta$ for some $k' \in \mathbb Z$, then $\tau_c(u)\le k'$.
\end{restatable}

We can now present our algorithm.

\begin{algorithm}\LinesNumbered
	\caption{Compute core committee for $n \le 5$}\label{alg:coren=5}
	\KwIn{Instance $(c,k,b)$ with $n \le 5$.}
	\KwOut{Integral committee $z \in \mathcal{W}$.}
	$\varepsilon \gets \frac{1}{L_n + 1}$\;
	Compute $x \in \R_{\ge 0}^{\lvert \mathcal{R}\rvert}$ with $\sum_{R \in \mathcal{R}} x_R \le k + \varepsilon$ that satisfies the fractional core condition according to \Cref{thm:Lindahl_efficient_computation}\;
	$u \gets \lfloor u(x)\rfloor$\;
	Compute fractional committee $x' \in \mathcal{P}$ achieving utility vector $u$ by solving the LP \eqref{eq:lp_min_feasibility}\;
	Compute integral committee $z \in \mathcal{W}$ by rounding $x'$ via \Cref{alg:roundfraccommittees}\;
	\Return $z$\;
\end{algorithm}

The algorithm sets $\varepsilon = \frac{1}{L_n + 1} < \frac{1}{L_n}$ and computes an approximate fractional core solution $x$ according to \Cref{thm:Lindahl_efficient_computation} satisfying $\sum_{R \in \mathcal R} x_R \leq k + \varepsilon$. It then takes the integer utility vector
\[
u := \lfloor u(x)\rfloor \in \Z_{\ge 0}^n.
\]
Because $x$ achieves this utility vector, we have 
$
\tau_c(u) \le k + \varepsilon.
$
By \Cref{lem:tau-denominator} with $k'=k$ and $\delta=\varepsilon<\frac{1}{L_n}$, we deduce that $\tau_c(u) \le k$.
Hence, the integer utility vector $u$ is fractional-committee-feasible for budget $k$, i.e., for the original instance $(c,k,b)$, and the algorithm computes a fractional committee $x'$ witnessing this by solving the LP \eqref{eq:lp_min_feasibility}.
Since $n \le 5$, \Cref{thm:monoidnormal} implies that there exists an integral committee $z \in \mathcal{W}$ achieving utility vector $u$ that can be computed via \Cref{alg:roundfraccommittees} from $x'$.

It remains to show that $z$ lies in the core. Suppose for contradiction that there exists a non-empty coalition $S \subseteq N$ and a committee $y \in \mathcal{W}_S$ such that
\[
u_i(y) > u_i(z) \qquad\text{for all } i \in S,
\]
and hence, because $u_i(y)$ is an integer and $z$ achieves utility vector $u = \lfloor u(x)\rfloor$, we have
\[
u_i(y) \ge u_i(z)+1 \ge \lfloor u_i(x) \rfloor + 1 > u_i(x) \qquad\text{for all } i \in S.
\]
Since $y \in \mathcal{W}_S \subseteq \mathcal{P}_S$, this contradicts the fact that $x$ satisfies the fractional core condition.
This proves that for $n \le 5$, a core committee can be computed in polynomial time. 

For $n \in \{6,7\}$, our theoretical core existence results require us to start with a special fractional core committee, namely one corresponding to a Lindahl equilibrium.
This suggests that algorithms for $n \in \{6,7\}$ also require more than finding an arbitrary fractional core allocation and rounding it while preserving the component-wise rounded-down utilities. To make use of our theory for Lindahl equilibria, they must determine the rounded down utilities of a committee $x$ in a Lindahl equilibrium.
Hence, we require exact feasibility tests of the form
\[
	u_i(x)\ge q \qquad\text{and}\qquad u_i(x)<q \qquad \text{for } q \in \mathbb{Z}.
\]
Approximate numerical methods do not suffice in general, since some $u_i(x)$ may be equal to, or arbitrarily close to, the threshold $q$, so a finite-precision approximation cannot reliably distinguish these two cases.

\subsection{Computing a Pareto-optimal core committee}

Now we explain how we can guarantee that the output committee not only lies in the core but is also Pareto-optimal (see \Cref{fn:pareto}).

\begin{theorem}\label{thm:core_po_compute_n5}
	For $n \le 5$, a Pareto-optimal core committee can be computed in polynomial time.
\end{theorem}
\begin{proof}
	Run the first three lines of \Cref{alg:coren=5} to compute the utility vector $u$. Consider the following integer linear program that looks for a committee maximizing utilitarian social welfare subject to achieving the utility vector $u$.
	\begin{equation*}
	\begin{aligned}
		\max \quad & \sum_{R\in\mathcal R} \lvert R\rvert\,x_R\\
		\text{s.t.}\quad
		& \sum_{R\ni i} x_R \ge u_i \quad \text{for all } i\in N,\\
		& \sum_{R\in\mathcal R} x_R \le k, \quad 0\le x_R \le c_R, \quad x_R\in\Z_{\ge 0}\\
	\end{aligned}
	\end{equation*}
	As argued in \Cref{sec:round-fractional-core}, this ILP has a feasible solution due to \Cref{thm:monoidnormal}. Since $n \le 5$, the program has at most $2^5 - 1 = 31$ variables.
	Thanks to a theorem by \citet{len1983integer}, an ILP with a constant number of variables can be solved in polynomial time. So let $x^*$ be the optimal solution of the program. Due to the constraints in the bottom row of the program, $x^*$ is a committee. The committee is Pareto-optimal, since any Pareto improvement would strictly increase the sum of utilities while preserving feasibility. By the argument in \Cref{sec:round-fractional-core}, since we have $u_i(x^*) \ge u_i$ for all $i \in N$, $x^*$ is in the core. Hence $x^*$ is a Pareto-optimal core committee.
\end{proof}

\begin{remark}\label{rem:lpsinsteadofilp}
	An alternative way of obtaining Pareto-optimality that avoids needing to solve an ILP is to run some type of serial dictatorship starting at the utility vector $u$ that we obtained as part of \Cref{alg:coren=5}. In detail, start at $u' \gets u$, then iterate through the voters $i \in N$ and repeatedly increase $u'_i$ by 1 whenever this increase preserves fractional-committee-feasibility. This can be checked by the following linear feasibility program with variables $(x_R)_{R \in \mathcal{R}}$: 
	\begin{equation*}
	\begin{aligned}
		\sum_{R \ni i} x_R \ge u'_i & \quad \text{for all $i \in N$}, \\
		\sum_{R \in \mathcal R} x_R \le k \\
		0 \le x_R \le c_R & \quad \text{for all $R \in \mathcal R$}.
	\end{aligned}
	\end{equation*}
	If the next increase for $i$ would destroy feasibility, move to the next voter. This process terminates after at most $nk$ steps. Next, use \Cref{alg:roundfraccommittees} to obtain a committee $x \in \mathcal{W}$ achieving the integer utility vector $u'$. Since $u' \ge u$, by the argument in \Cref{sec:round-fractional-core}, the resulting committee is in the core. It is also Pareto optimal: if there was a Pareto improvement $y \in \mathcal{W}$, let $i$ be the first voter (in the order used by the algorithm) with $u_i(y) > u_i(x)$ (and hence $u_i(y) \ge u_i(x) + 1 > u_i'$). But then the algorithm would not have proceeded to the next voter when it did, since increasing the utility for $i$ by $1$ would have preserved fractional-committee-feasibility (witnessed by $y \in \mathcal{P}$). 
	
	Note that the sequence of Pareto improvements is not limited to serial dictatorships but can be chosen arbitrarily. In contrast to the Lenstra-ILP-based method from \Cref{thm:core_po_compute_n5}, this way of obtaining Pareto optimality is only pseudo-polynomial when the supplies $c_R$ and $k$ are encoded in binary.
\end{remark}

\section{Counterexamples beyond approval-based committee elections}\label{sec:counterex}
In principle, the applied techniques via normal affine monoids could be used in a variety of models beyond approval-based committee elections to prove existence of core-stable outcomes.
However, in this section we give examples that illustrate that the techniques do not work for many other natural models, since for them normality of the associated monoid already fails for $n=3$.

\subsection{Non-unit costs and participatory budgeting}

The approval-based committee setting is based on a cardinality constraint. More generally, one can replace this by a knapsack constraint, where each candidate is assigned a cost and an outcome is feasible if its total cost does not exceed a given global budget. This model is also called \emph{participatory budgeting} (PB), a process used by many cities to let residents influence how the local government spends its budget. 
Proportionality properties for the PB setting have been well-studied \citep{ALT18a,RSM25a,BFL+23a,PPS21a}. The definition of the core extends easily to this more general setting, using \emph{cardinality-based utilities} (where a voter's utility is the number of approved projects in the outcome). However, normality fails even for $n = 3$.

\begin{example}[Non-unit costs]\label{ex:nonunitcosts-n3}
    Let $N=\{1,2,3\}$ and $\mathcal{R}_c = \{12, 13, 23\}$ with $c_R = 1$ for each $R \in \mathcal R_c$, and let each candidate have a cost of $2$. 
    The total budget of $3$ is split equally among the three voters.
    Consider the utility vector $u=(1,1,1)$.
    This vector is fractional-committee-feasible: select half of each of the three candidates; then the total cost is $3 \cdot (\frac{1}{2} \cdot 2) = 3$, and each voter receives utility $2 \cdot \frac{1}{2}= 1$.

    However, $u$ is not integral-committee-feasible. 
    Indeed, with a budget $3$ one can afford at most one candidate, and each candidate is approved by only two voters. 
    Hence, no integral committee can give all three voters utility at least $1$.
    Thus, the integer utility vector $u$ is fractional-committee-feasible but not integral-committee-feasible in the PB setting, so the corresponding monoid is not normal for $n=3$. \qed
\end{example}

The same example can be adapted to the case of \emph{cost utilities}, where a voter's utility is the total cost of approved projects in the outcome by using the utility vector $u = (2,2,2)$. For cost utilities, there actually exists an example where the core is empty for $n = 3$ \citep{maly2025empty}. For cardinality-based utilities, non-emptiness remains an open problem in the approval voting context.

\subsection{Additive valuations}

Returning to the committee setting with unit costs, another generalization concerns the utility functions: instead of approvals, one could let voters specify a non-negative score for each candidate, with a voter's utility for an outcome being the sum of the scores of the candidates in it. This extends our model to \emph{additive utilities}. The core for additive utilities has been studied in the context of core approximations \citep{munagala2022approximate} as well as relaxations like EJR1 and FJR \citep{PPS21a}. The core itself can be empty in this model. The smallest known counterexample uses $n = 6$ voters (\citealp{PPS21a}, Example 2; see also \citealp{fain2018fair}, Appendix C) and combines two ``Condorcet cycles'' (reminiscent of \Cref{fig:two-triangles-intro}). Normality for this model fails even for $n = 3$ and a single Condorcet cycle.

\begin{example}[Additive valuations]\label{ex:additive-n3}
    Let $N=\{1,2,3\}$, $k=1$, and consider three unit-cost candidates $C = \{c_1,c_2,c_3\}$ with additive valuations given by
    \[
        \begin{pmatrix}
        2 & 0 & 1\\
        1 & 2 & 0\\
        0 & 1 & 2
        \end{pmatrix},
    \]
    where the entry in row $i$ and column $j$ denotes voter $i$'s valuation for candidate $c_j$.
    Again, consider the utility vector $u=(1,1,1)$. 
    This vector is fractional-committee-feasible: the fractional committee $x_{c_1}=x_{c_2}=x_{c_3}=\frac{1}{3}$
    has total size $1$ and provides every agent with utility $1$.

    However, $u$ is not integral-committee-feasible. 
    Since $k=1$, any integral committee consists of a single candidate. 
    But each of them gives one voter utility $0$, so no integral committee can realize $(1,1,1)$.
    Hence, the corresponding monoid is not normal already for $n=3$. \qed
\end{example}

\subsection{Droop quota}

The core is defined based on an intuition that a group of size at least $\ell \cdot \frac{n}{k}$ should deserve to decide over $\ell$ committee members. The threshold ``$\frac{n}{k}$'' is known as the \emph{Hare quota}. However, it is possible to define the core based on the smaller \emph{Droop quota} ``$\frac{n}{k+1}$'' which gives the same guarantees to groups of size \emph{strictly larger} than $\ell \cdot \frac{n}{k+1}$.
Using standard notation (see \Cref{sec:votertypes}), a committee $W\subseteq C$ with $|W|\le k$ is in the \emph{Droop core} if there is no non-empty coalition $S\subseteq N$ and no committee $T\subseteq C$ such that $|A_i\cap T| > |A_i\cap W|$ for all $i \in S$ and $|S| > |T| \cdot \frac{n}{k+1}$ (while the Hare quota only requires $|S| \ge |T| \cdot \frac{n}{k}$). 

The Droop quota is attractive since it gives guarantees to smaller groups. It turns out that for several representation axioms like EJR, if their definitions are strengthened to the Droop quota, they remain satisfiable by standard rules like PAV or MES \citep{janson2018thresholds,casey2025droop}. However, for the core, moving to Droop can lead to difficulties. In particular, \citet[Section 6]{peters2025fewseats} found that the PAV rule violates the Droop core even for $k=6$, while it satisfies the Hare core for up to $k = 8$.

Regarding the few-voters case, we find that our approach does not extend to the Droop core. 
The natural way to incorporate the Droop quota into our unit-cost approval-based model is to set each voter's budget $b_i$ to a value $\frac{k+1}{n} - \varepsilon$ for a sufficiently small $\varepsilon$.
Then the overall budget $b(N)$ that can be used by a fractional committee is just below $k+1$, but any integral committee can have size at most $k$.
In this model, the corresponding utility-rounding statement already fails for $n=3$ and $k=1$.

\begin{example}[Droop quota]\label{ex:droop-normality-fails}
    Fix $0<\varepsilon<\frac16$. Let $N=\{1,2,3\}$, $k=1$, and $\mathcal R_c=\{12,13,23\}$ with $c_R = 1$ for all $R \in \mathcal R_c$, and give every voter budget $b_i=\frac23-\varepsilon$.
    The integer utility vector $u=(1,1,1)$ is fractional-committee-feasible under the Droop-type fractional bound, since
    \[
        x_{12}=x_{13}=x_{23}=\tfrac12
    \]
    has total size $\frac{3}{2}<2-3\varepsilon=b(N)$ and yields utility $1$ to every voter.

    By contrast, $u$ is not integral-committee-feasible via any integral committee of size at most $k=1$, since every such committee consists of a single candidate and will therefore leave one voter with utility $0$.
    Thus, an integer utility vector may be fractional-committee-feasible under the Droop bound without being integral-committee-feasible at committee size $k$. 
    In particular, the normality property underlying our approach fails for the Droop quota. \qed
\end{example}

\section{Related work}

\paragraph{The core in approval-based committee elections.}
\citet{ABC+16a} were the first to connect proportionality notions in approval-based committee elections \citep[see, e.g.,][for a thematic overview]{LaMa23a} to stability notions from cooperative game theory and proposed the core as a strong fairness axiom, incorporating proportional representation. 
They further showed that core stability strictly strengthens another proportionality axiom introduced there, called extended justified representation (EJR), which is satisfied by several natural voting rules \citep{ABC+16a, PeSk20a, BrillPeters2023Robust}.

The game-theoretic perspective on committee elections has been further developed by \citet{HKMS24a}, who modeled approval-based committee elections as a budgeting game and related the core to the $\alpha$-core from game theory \citep{aumann1961bargaining}. 
For budgeting games with an underlying tree structure, they proved the existence of a strong equilibrium, a concept stronger than the core and thus implying core existence in the corresponding restricted domain. 
For general instances, the existence of strong equilibria also remains open.

\citet{ABC+16a} concluded that ``the core stability condition appears to be too demanding, as none of the voting rules considered in our work is guaranteed to produce a core stable outcome''. Almost ten years later, core non-emptiness in approval-based committee elections is still an open question, although some progress has been made since then. 

\paragraph{The core in sub- and superdomains.} 
\citet{pierczynski2022core} proved that the core is always non-empty for the subdomain of linear-consistent instances, a class that contains both the voter-interval and candidate-interval domains \citep{ElLa15a}. In these restricted domains, a core outcome can be computed in polynomial time. 
The core is also known to be non-empty when there are at least $k$ copies of each candidate type \citep{BGP+24a}.
More recently, \citet{peters2025fewseats} showed that proportional approval voting (PAV) returns a core committee whenever the committee size satisfies $k \leq 8$, regardless of the number of voters and candidates. 
The same paper also proves that the core is always non-empty whenever the number of candidates is at most $15$ or $k \in \{m-1,m\}$. 
\citet{BTC+26a} use an approach based on mixed-integer linear programming to quantify the stability of core committees and show relations to various notions of \emph{priceability} \citep{PeSk20a,munagala2022auditing}.

Despite these positive existence results, there is also some evidence that the core might be empty, as \citet{ABC+16a} implicitly conjectured. 
The Method of Equal Shares (MES) \citep{PeSk20a,PPS21a}, arguably one of the most promising candidates for satisfying the core apart from PAV, only satisfies a weaker variant based on constrained deviations.
We provide some insights on the connection between MES and the core in \Cref{appendix:mes}. 
\citet[Section 5]{PeSk20a} also showed that no welfarist rule can satisfy the core, and that PAV achieves only a factor $2$ approximation, which is optimal among rules satisfying a Pigou--Dalton-type fairness principle.
Moreover, in many superdomains of approval-based committee elections, e.g., additive valuations or participatory budgeting with cost utilities, there are examples for which the core is empty (see \Cref{sec:counterex}).

\paragraph{Approximations of the core.}
A related line of work studies relaxations of the core rather than exact existence. 
\citet{cheng2020group} established the existence of stable lotteries, i.e., lotteries over committees that satisfy core stability in expectation, for several preference models in committee elections.
Using an intricate rounding procedure for such stable lotteries, \citet{jiang2020approximately} showed the existence of a $16$-approximately stable integral committee for monotone preferences, where the relaxation is on the size of a deviating coalition: a coalition must be larger than its proportional size by a constant factor in order to form a blocking coalition. 
Using a different approach based on multilinear extensions and market-clearing ideas, \citet{munagala2022approximate} obtained a multiplicative $9.27$-approximation to the core in the more general setting of additive utilities, where the approximation is with respect to the utility guarantee in the blocking condition.
There, a deviation is ruled out unless every member of the coalition can improve by a sufficiently large multiplicative factor, even after allowing an additional ``additament'' candidate. 
Recently, \citet{song2026ordinal} combined a Lindahl-type equilibrium with an iterative dependent rounding procedure to obtain a $5.15$-approximate-core guarantee for the approval setting, again in the sense of a multiplicative relaxation on the size of blocking coalitions.
Lastly, \citet{gao2025computation} showed the existence of a $3.65$-approximately stable committee.
The corresponding algorithm is randomized and runs in expected polynomial time; it computes a Lindahl equilibrium and then samples from an associated strongly Rayleigh distribution.

\paragraph{Fractional relaxations of the core.} 
Another natural route is to allow for fractional committees, meaning that candidates can be chosen ``partially''. 
Such divisible settings exhibit close connections to public-goods markets where core non-emptiness can be guaranteed via Lindahl equilibria.

The setting with at least $k$ copies of each candidate type is closely related to budget aggregation \citep[see, e.g.,][for an overview]{SuTe26a}. There, \citet{FGM16a} related the notion of Lindahl equilibrium to maximizing Nash welfare, i.e., maximizing the product of voters' utilities, for a variety of utility functions. They further showed that this connection enables an efficient approximation of utilities in fractional core outcomes.  

For the setting where the number of available candidates of a certain type might be capped, \citet{KrPe25b} gave a polynomial-time algorithm for computing an approximate fractional core outcome for additive utilities, where the approximation is additive with respect to utilities and the committee size. 
\citet{SuVo24a} model committee elections via networks and derive a new proportionality axiom (weaker than the core; incomparable to the fractional core) from max flows on these networks. 
They give polynomial-time algorithms for finding a fractional committee that satisfies this \emph{ex-ante} notion of proportionality and for decomposing it into a probability distribution over committees such that each committee in the support yields strong \emph{ex-post} proportionality guarantees, an approach known as best-of-both-worlds-fairness \citep{FSV20a}.

\section{Conclusion}

The core exists when there are at most $7$ voter types. Together with previously known existence results for small committee size or small candidate sets due to \citet{peters2025fewseats}, we can conclude:
\begin{quote}
	If there exists an approval-based committee election instance with an empty core, then it must have at least $8$ different voter types, at least $16$ candidates, and a committee size of at least $9$.
\end{quote}
Thus, no small counterexample exists if one exists at all.

It would be interesting to see if the rounding technique based on affine monoids can be further extended to obtain additional results, such as core existence for novel restricted preference domains. 
In addition, the possibility for $n \le 5$ of rounding fractional committees to integral committees in a utility-preserving way might be useful to obtain existence results for solution concepts other than the core. One plausible direction might be strengthenings of the core such as stable priceability \citep{peters2021market} or Lindahl priceability \citep{munagala2022auditing}, though these concepts are perhaps not closely enough linked to utilities. One could try to create additional tools to tackle the case with eight or more voters. However, our approach of enumerating and checking all minimal holes will very quickly become computationally intractable. On an Apple M1 Pro MacBook Pro (8 cores, 16 GB RAM), using Z3 4.16.0, it takes 4.3 seconds to check the case $n = 6$ but 2.5 hours to check the case $n = 7$.

Another research direction concerns the fractional setting and the question whether a fractional core committee can be computed exactly in polynomial time.

\section*{Use of Large Language Models}
The writing of this paper was done by us, but several mathematical ideas were contributed by large language models. 
\begin{itemize}
	\item While being asked to optimize the runtime of a script that we used to search for counterexamples, \texttt{gpt-5-4}, acting as a coding agent, used \texttt{Normaliz} and implicitly affine monoids to speed up the search. Based on its report, we realized that affine monoids can be used to solve our problem. In addition, our first version of the proof of normality for $n = 5$ (\Cref{thm:monoidnormal}) used a lot of case analysis, which \texttt{gpt-5-5} managed to simplify, notably by proposing \Cref{lem:k-1<Rc}. 
	\item We then tried to extend the $n = 5$ argument to $n = 6$, hoping that we could somehow patch each of the 50 minimal counterexamples (\Cref{tbl:50-holes}) that show that normality fails for $n = 6$. It is easy to check that in each of those minimal counterexamples, a core committee exists. Interestingly, \texttt{gpt-5-5-pro} could not derive a way to use this fact to deduce that the core is also non-empty for non-minimal counterexamples (and neither could we). However, \texttt{claude-fable-5} (while it was briefly publicly available in June 2026) did find a way to deduce this, using a decomposition similarly to the one we develop in \Cref{sec:decomposition} and the idea of reducing the utility of one voter by 1. When asked to extend further to $n = 7$, it also suggested a complicated way to deal with the 298 minimal holes mentioned in \Cref{thm:n7minimalholes} that cannot be handled by the technique that works for $n = 6$. Later, \texttt{gpt-5-5-pro} suggested the simpler property based on $\beta$ that we actually use in \Cref{thm:n7minimalholes}.
	\item The counterexample on which MES fails the core for $n = 9$ (\Cref{ex:mes_violates_core_unique} in the appendix) was constructed by \texttt{gpt-5-4} (which managed to find an example with $n = 36$) and \texttt{gpt-5-5} (which managed to reduce further to $n = 9$). 
	\item We received help from \texttt{gpt-5-4} with the details of \Cref{lem:tau-denominator} about denominators in solutions of LPs. 
\end{itemize}

Despite these ideas being developed by frontier LLMs, we invested significant effort into verifying them and making them accessible to humans.

A \href{https://github.com/DominikPeters/ABCVotingLean}{formalization of our existence result} for $n \le 5$ in Lean 4 was produced by \texttt{gpt-5-5} acting in the codex \texttt{/goal} harness in approximately 10 hours of autonomous work.

\iflatexml
	\daharxivbib{extract}
\else
	\printbibliography
\fi
\newpage
\appendix
\crefalias{section}{appsec}
\crefalias{subsection}{appsec}

\section{Omitted proofs}\label{appendix:omitted_proofs}

\subsection{Omitted proofs from \Cref{sec:prelims}}

\fracCoreExists*
\begin{proof}
	For every $\varepsilon>0$, define the perturbed utilities by
	\[
		u_i^\varepsilon(x) = \sum_{\substack{R\in\mathcal R\\ i\in R}} \min\{x_R,c_R\} + \varepsilon\sum_{R\in\mathcal R}x_R.
	\]
	Each $u_i^\varepsilon$ is continuous and concave on $\mathbb R_{\geq 0}^{\lvert \mathcal R\rvert}$. Moreover, it is strictly increasing
	in every coordinate, since increasing any coordinate $x_R$ by $\delta>0$ increases $u_i^\varepsilon$ by at least $\varepsilon\delta$.

	We can assume that $\sum_{R\in\mathcal R}c_R>k>0$, as otherwise we can just set $x_R=c_R$ for all $R \in \mathcal{R}$, which even yields an integral core committee. By Foley's existence theorem, there exists a Lindahl equilibrium $(x^\varepsilon,p^\varepsilon)$ for this perturbed economy.

	We claim that $x^\varepsilon$ satisfies the supply constraints, i.e., $x_R^\varepsilon \le c_R$ for all $R\in\mathcal R$.
	Suppose not, and let $R'\in\mathcal R$ be such that $x_{R'}^\varepsilon>c_{R'}$.
	Since preferences are strictly increasing, the full budget is spent in equilibrium, so $\sum_{R\in\mathcal R} x_R^\varepsilon = k$.
	As $\sum_{R\in\mathcal R} c_R > k$, there exists some type $R''\in\mathcal R$ with $x_{R''}^\varepsilon<c_{R''}$.
	Choose $\delta>0$ so that
	\[
		x_{R'}^\varepsilon-\delta \ge c_{R'}
		\qquad\text{and}\qquad
		x_{R''}^\varepsilon+\delta \le c_{R''}.
	\]
	Let $x'$ be the fractional committee in which, compared to $x^{\varepsilon}$, a mass of $\delta$ has been moved from $R'$ to $R''$.
	Then for every voter $i$, the common perturbation term $\varepsilon\sum_R x_R$ is unchanged, while the capped part, $\sum_{R\ni i}\min\{x_R,c_R\}$, does not decrease: the contribution of $R'$ is unchanged because $x_{R'}^\varepsilon-\delta \ge c_{R'}$, and the contribution of $R''$ weakly increases.
	Moreover, for every voter $i\in R''$, the contribution of $R''$ increases strictly.
	Since $R''\neq\emptyset$, the allocation $x'$ is a Pareto improvement over $x^\varepsilon$, contradicting Pareto efficiency of the Lindahl equilibrium. 
	Hence, $x^\varepsilon_R\le c_R$ for all $R$, and therefore $x^\varepsilon\in\mathcal P$.

	Since $\mathcal P$ is compact, there exists a sequence $\varepsilon_t\to 0$ and some $x^*\in\mathcal P$ such that $x^{\varepsilon_t}\to x^*$. Moreover, for every $i\in N$ and $R\in\mathcal R$, non-negativity of personalized prices and producer optimality imply
	\[
		0\leq p_{iR}^{\varepsilon_t} \leq \sum_{j\in N}p_{jR}^{\varepsilon_t} \leq 1.
	\]
	Hence, $(p^{\varepsilon_t})_t$ is bounded. Passing to a further subsequence (relabelling if necessary), there exists $p^* \in [0,1]^{n \cdot \lvert \mathcal{R} \rvert}$ such that $p^{\varepsilon_t}\to p^*$. Consequently,
	\[
		(x^{\varepsilon_t},p^{\varepsilon_t}) \to (x^*,p^*).
	\]
	
	Moreover, $\sum_{R\in\mathcal R} x_R^*=\lim_{t \to \infty}\sum_{R\in\mathcal R} x_R^{\varepsilon_t} = \lim_{t \to \infty} k=k$.
	
	We now verify that $(x^*,p^*)$ satisfies the conditions of a Lindahl equilibrium for the original economy.
	
	\medskip
	\noindent\textbf{$x^*$ is affordable.}
	For every $t$ and $i \in N$, $\langle p_i^{\varepsilon_t}, x^{\varepsilon_t}\rangle\leq b_i$. Since the scalar product is continuous on $\mathcal{P}\times [0,1]^{n \cdot \lvert\mathcal{R}\rvert}$, taking limits gives $\langle p_i^*,x^*\rangle\leq b_i$. Thus, $x^*$ is affordable with respect to the personalized prices in the limit.

	\medskip
	\noindent\textbf{$x^*$ is utility-maximizing.}
	Suppose, toward a contradiction, that $x^*$ is not utility-maximizing for voter $i$. Then there exists $y\geq 0$ ($0 \le y_R \le c_R$ for all $R \in \mathcal{R}$) such that $\langle p_i^*,y\rangle\leq b_i$ and $u_i(y)>u_i(x^*)$.
	By continuity of $u_i$, there exists $\lambda\in(0,1)$, sufficiently close to $1$, such that $u_i(\lambda \cdot y)>u_i(x^*)$.
	Because $b_i>0$,
	\[
		\langle p_i^*,\lambda \cdot y\rangle
		= \lambda \langle p_i^*, y\rangle
		\leq \lambda b_i
		< b_i.
	\]
	Hence, $\lambda \cdot y$ is strictly affordable at $p_i^*$. Since $p_i^{\varepsilon_t}\to p_i^*$, it follows that $\langle p_i^{\varepsilon_t}, \lambda \cdot y \rangle \leq b_i$ for all sufficiently large $t$. Consumer optimality of $x^{\varepsilon_t}$ in the perturbed economy gives $u_i^{\varepsilon_t}(x^{\varepsilon_t}) \geq u_i^{\varepsilon_t}(\lambda \cdot y)$. Expanding the perturbed utilities yields
	\[
		u_i(x^{\varepsilon_t}) + \varepsilon_t \sum_{R\in\mathcal R}x_R^{\varepsilon_t}
		\geq u_i(\lambda \cdot y) + \varepsilon_t\lambda \sum_{R\in\mathcal R} y_R.
	\]
	Since $(x^{\varepsilon_t})_t$ is bounded and $y$ is fixed,
	\[
		\lim_{t \to \infty} \varepsilon_t \sum_{R\in\mathcal R}x_R^{\varepsilon_t} = 0 \quad \text{and} \quad
		\lim_{t \to \infty}\varepsilon_t\lambda \sum_{R\in\mathcal R} y_R = 0.
	\]
	Taking limits and using continuity of $u_i$ yields $u_i(x^*) \geq u_i(\lambda\cdot y)$, contradicting $u_i(\lambda\cdot y)>u_i(x^*)$. Therefore,
	\[
		x^* \in \arg\max_{\substack{z\geq 0 \text{ s.t. } \\ \forall R \in \mathcal{R}: 0 \le z_R \le c_R}} \left\{ u_i(z) : \langle p_i^*,z\rangle\leq b_i \right\}.
	\]
	Since $i$ was arbitrary, consumer optimality holds for every voter.

	\medskip
	\noindent\textbf{$p^*$ is profit-maximizing.}
	Since $q_R^{\varepsilon_t} \coloneqq \sum_i p_{iR}^{\varepsilon_t}\leq 1$ for every $t$ and $R \in \mathcal{R}$, taking limits gives $q_R^*\leq 1$ for every $R \in \mathcal R$.
	Now suppose that $x_R^*>0$. Since $x_R^{\varepsilon_t}\to x_R^*$, we have $x_R^{\varepsilon_t}>0$ for all sufficiently large $t$. Hence, $q_R^{\varepsilon_t}=1$ for all sufficiently large $t$, and therefore $q_R^*=1$. Thus,
	\[
		q_R^* \leq 1 \qquad \text{for every }R\in\mathcal R,
	\]
	with equality whenever $x_R^*>0$.

	\medskip
	Therefore, $(x^*,p^*)$ is a Lindahl equilibrium of the original economy.
	Since every Lindahl equilibrium belongs to the core of the corresponding public-goods economy, $x^*$ also lies in the fractional core of the original instance.
\end{proof}

\voterTypes*
\begin{proof}
	If $k=0$, the empty committee is trivially in the core; hence, assume $k>0$.
	Let $A^1,\dots,A^t$ be the distinct approval sets occurring in the profile, and let $N=\{1,\dots,t\}$ be the corresponding set of voter types. For each type $i\in N$, let
	\[
	n_i = \bigl|\{j\in\hat N:A_j=A^i\}\bigr|, \quad
	b_i = n_i\frac{k}{|\hat N|}, \quad
	c_R = \bigl|\{g\in\hat C: \{i\in N:g\in A^i\}=R\}\bigr| \text{ for all } R\in \mathcal{R}.
	\]
	
	By assumption, the associated instance $(c,k,b)$ admits a committee $ x=(x_R)_{R \in \mathcal{R}}$ in the core. Construct a committee $W\subseteq\hat C$ by selecting exactly $x_R$ candidates of type $R$ for every $R\in \mathcal{R}$. This is possible because $x_R\le c_R$. Moreover, $\lvert W\rvert=\sum_{R\in \mathcal{R}}x_R\le k$.
	
	We show that $W$ belongs to the core of the original approval instance. Suppose, for contradiction, that $W$ is blocked by a coalition $S\subseteq\hat N$ and a set of candidates $T\subseteq\hat C$. Thus, $|T| \le |S|\frac{k}{|\hat N|}$ and $|A_j\cap T| > |A_j\cap W|$ for every $j \in S$.
	We may assume that $S$ contains either all or none of the voters of each type. Indeed, if $j\in S$ and $j'\notin S$ satisfy $A_j=A_{j'}$, then
	\[
	|A_{j'}\cap T| = |A_j\cap T| > |A_j\cap W| = |A_{j'}\cap W|.
	\]
	Hence, $j'$ can be added to $S$ without destroying the blocking property, and the enlarged coalition can still afford $T$. Repeating this argument yields a blocking coalition that is a non-empty union of complete voter types. Let $I \subseteq N$ be exactly this set of voter types contained in coalition $S$.
	We may also assume that every candidate in $T$ is approved by at least one voter, since deleting an unapproved candidate preserves every voter's utility and weakly decreases $|T|$. We can then rewrite $T$ by letting $y_R$ be the number of candidates in $T$ of type $R$ for every $R \in \mathcal{R}$. Hence, we get $y \in \mathcal{W}_I$ and
	\[
	\sum_{R\in \mathcal{R}}y_R = |T| \le |S|\frac{k}{|\hat N|} = \sum_{i\in I}b_i,
	\]
	so coalition $I$ can afford $y$. Further, for every $i\in I$,
	\[
	\sum_{R\ni i}y_R = |A^i\cap T| > |A^i\cap W| = \sum_{R\ni i}x_R.
	\]
	Thus, $I$ blocks $x$, contradicting that $x$ is in the core of the instance $(c,k,b)$. Consequently, $W$ must also be in the core of the original approval instance.
\end{proof}

\subsection{Omitted proofs from \Cref{sec:corecomputation}}
\label{app:omitted-corecomputation}

\lindahlEfficientComputation*
\begin{proof}[Proof]
Let $\hat\varepsilon = \frac{\varepsilon}{n+1}$. \citet[Section 5.3]{KrPe25b} show that one can compute in time polynomial in the problem size and in $\log\frac{1}{\hat\varepsilon}$ a vector $\hat x \in \R_{\ge 0}^{\lvert \mathcal{R}\rvert}$ with 
\begin{enumerate}
	\item[(i)] $\hat x_R \le c_R$ for all $R \in \mathcal{R}$,
	\item[(ii)] $\sum_{R \in \mathcal R} \hat x_R \le k+\hat\varepsilon$, and
	\item[(iii)] there does not exist a non-empty $S \subseteq N$ and fractional committee $y \in \mathcal{P}_S$ such that $u_i(y) > u_i(\hat x) + \hat\varepsilon$ for all $i \in S$.
\end{enumerate}
We will now change $\hat x$ so that it overshoots the budget by more (increasing the violation of (ii)) but removes the $\hat\varepsilon$ gap in condition (iii). We do this by enlarging $\hat x$ so that every voter's utility goes up by $\hat\varepsilon$ (if possible).

To do so, begin by setting $x \gets \hat x$. Then iterate through the set of voters. For each $i \in N$, compute the residual approved capacity $r_i := \sum_{R \ni i} (c_R - x_R)$ and increase $x$ on approved types of voter $i$ by a total amount of $\alpha_i=\min\{ \hat \varepsilon, r_i\}$ while respecting the supply constraints (which is possible by definition of $\alpha_i$).

After processing all voters, the vector $x$ satisfies
\[
\sum_{R \in \mathcal R} x_R \le \sum_{R \in \mathcal R} \hat x_R + \sum_{i \in N} \alpha_i \le (k + \hat\varepsilon) + n\hat\varepsilon = k + (n+1) \hat\varepsilon = k +\varepsilon.
\]
By construction, $x$ satisfies the supply constraints. We finish the proof by showing that $x$ satisfies the fractional core condition.
Let $S \subseteq N$ be a non-empty coalition, and consider $y \in \mathcal{P}_S$. By (iii), there exists $i \in S$ with $u_i(\hat x) + \hat\varepsilon \ge u_i(y)$. Now, by construction of $x$, we have $u_i(x) \ge u_i(\hat x) + \alpha_i$. Therefore, either $u_i(x) \ge u_i(\hat x) + \hat \varepsilon$ and hence $u_i(x) \ge u_i(y)$, or we had $r_i < \hat \varepsilon$ during the construction of $x$. In the latter case, $x$ saturates the supply of each candidate approved by voter $i$ and therefore gives $i$ the highest possible utility among vectors satisfying the supply constraints, so again $u_i(x) \ge u_i(y)$. In either case, we have $u_i(x) \ge u_i(y)$ and have thus verified the core condition for this deviation.
\end{proof}

\tauDenominator*
\begin{proof}
    Consider the linear program \eqref{eq:lp_min_feasibility} defining $\tau_c(u)$. 
	Under the assumption that $\tau_c(u) < \infty$, the feasible region is non-empty. 
    Moreover, it is bounded because each variable satisfies $0 \le z_R \le c_R$. Hence, an optimum exists. 
    By standard linear programming theory, some optimum is attained at a basic feasible solution.

    Let $z$ be such a basic feasible optimum. We first show that at most $n$ variables of $z$ are not fixed at one of their bounds $0$ or $c_R$.

    Indeed, the linear program has $|\mathcal R|$ variables. 
    At a basic feasible solution, one needs $|\mathcal R|$ linearly independent active constraints. Suppose that exactly $t$ variables of $z$ are not fixed at a bound. 
    Then the remaining $|\mathcal R|-t$ variables are fixed either by an active lower bound $z_R=0$ or by an active upper bound $z_R=c_R$. Thus, at most $|\mathcal R|-t$ of the required active constraints can come from the box constraints. 
    The only other constraints are the $n$ voter constraints
    \[
        \sum_{R \ni i} z_R \ge u_i
        \qquad (i \in N).
    \]
    Therefore, the total number of linearly independent active constraints is at most $(|\mathcal R|-t) + n$.
    Since a basic feasible solution requires at least $|\mathcal R|$ such constraints, we obtain $(|\mathcal R|-t) + n \ge |\mathcal R|$, and hence, $t \le n$.
    In other words, at most $n$ variables of $z$ are not fixed at one of their bounds ($0$ or $c_R$).

    Let $J \subseteq \mathcal R$ be the set of indices of the $t= |J|$ variables that are not fixed at a bound.
    If $t=0$, every coordinate of $z$ is fixed at an integral bound, and thus $\tau_c(u) = \sum_R z_R$ is integral. Hence, we may assume that $t\ge1$.
    We now focus on these ``free'' variables.
    Since all other variables are fixed at either $0$ or $c_R$, their contributions can be moved to the right-hand side of the active voter constraints. 
    In this way, the free variables are determined by a system of linear equations
    \[
        B z_J = d
    \]
    that consists of $t$ active constraints, where:
    \begin{itemize}
    	\item $B$ is a $t\times t$ binary matrix that contains the approval sets of voters from active voter constraints as rows,
        \item $z_J$ is the vector of free variables,
        \item $d$ is an integer vector; it is obtained from the integer utility vector $u$ by subtracting contributions of variables $z_R$ fixed at $c_R$.
        
    \end{itemize}
    Since $z$ is a basic feasible solution, the active constraints determining the free variables are linearly independent. 
    Thus, $B$ is invertible. 
    In particular, $t\le n$, so $B$ is an invertible binary square matrix of order at most $n$.

    By Cramer's rule, each coordinate of $z_J$ is of the form
    \[
        \frac{\det(B_j)}{\det(B)},
    \]
    where $B_j$ is obtained from $B$ by replacing the $j$-th column by $d$. 
    Since both $B$ and $d$ are integral, all determinants appearing here are integers. 
    Therefore, each value of a free variable is rational and can be represented with a denominator dividing $\lvert \det(B) \rvert$, which divides $L_n$ by definition. 
    The variables fixed at $0$ or $c_R$ are integers. 
    Consequently, every coordinate of the basic optimum $z$ lies in $\frac{1}{L_n}\mathbb Z$.

    Summing all coordinates, we conclude that
    \[
        \tau_c(u)=\sum_{R\in\mathcal R} z_R \in \frac{1}{L_n}\Z.
    \]

    For the second statement, let $k'\in\mathbb Z$ and assume that
    \[
    \tau_c(u) \le k'+\delta \qquad\text{with}\qquad 0<\delta<\frac{1}{L_n}.
    \]
    Since $k' \in \mathbb Z \subseteq \frac{1}{L_n}\mathbb Z$, the smallest element of $\frac{1}{L_n}\mathbb Z$ that is strictly larger than $k'$ is $k' + \frac{1}{L_n}$.
    As $\delta<\frac{1}{L_n}$ and $\tau_c(u)\le k'+\delta$, it follows that $\tau_c(u)\le k'$.
\end{proof}

\section{Discussion of the Method of Equal Shares} \label{appendix:mes}

In this section, we present two examples illustrating how the Method of Equal Shares (MES) can fail to return a committee in the core and propose a specific tie-breaking that always returns core committees for three voters. In this section, we assume that $k$ is at most as large as the number of available candidates.
We begin by showing that MES may return a committee that is not in the core, and that this phenomenon already arises with three voters.

\begin{wrapstuff}[r,width=0.33\textwidth,type=figure]
	\centering
	\begin{tikzpicture}[yscale=0.56,xscale=0.9,voter/.style={anchor=south}]		
		\foreach \i in {1,...,3}
			\node[voter] at (\i-0.5,-1) {$\i$};

		\foreach[count=\y from 0] \lab in {$a_1$,$a_2$,$a_3$,$a_4$,$a_5$,$a_6$}{
			\draw[fill=orange!40] (0,\y) rectangle (2,\y+1);
			\node at (1,\y+0.5) {\lab};
		}

		\foreach[count=\y from 6] \lab in {$b_1$,$b_2$}{
			\draw[fill=violet!40] (0,\y) rectangle (1,\y+1);
			\draw[fill=violet!40] (2,\y) rectangle (3,\y+1);
			\node at (0.5,\y+0.5) {\lab};
			\node at (2.5,\y+0.5) {\lab};
		}

		\foreach[count=\y from 8] \lab in {$c_1$,$c_2$}{
			\draw[fill=teal!40] (1,\y) rectangle (3,\y+1);
			\node at (2,\y+0.5) {\lab};
		}

		\foreach[count=\y from 0] \lab in {$d_1$,$d_2$,$d_3$}{
			\draw[fill=blue!25] (2,\y) rectangle (3,\y+1);
			\node at (2.5,\y+0.5) {\lab};
		}
	\end{tikzpicture}
	\caption{Example for MES violating the core for $n=3$.}
	\label{fig:mes-core-counterexample}
\end{wrapstuff}
\begin{example}\label{ex:mes-not-core-n3}
	Consider the approval-based committee election instance with $n=3$, $k=9$, equal individual budgets, and the approval profile shown in \Cref{fig:mes-core-counterexample}. One possible MES outcome is $W=\{a_1,\dots,a_6\}\cup\{d_1,d_2,d_3\}$.
	Intuitively, under a suitable tie-breaking rule, Voters $1$ and $2$ first exhaust their budgets on the six candidates $a_1,\dots,a_6$, after which Voter $3$ can still afford the three candidates $d_1,d_2,d_3$.
	Now consider the committee $W'=\{a_1,\dots,a_5\}\cup\{b_1,b_2,c_1,c_2\}$.
	Every voter strictly prefers $W'$ to $W$, since each voter approves strictly more candidates in $W'$ than in $W$. 
	Hence, the grand coalition blocks $W$, and $W$ is therefore not in the core.
	Moreover, $W'$ is precisely the type of committee returned by \Cref{alg:coren=3} below. \qed
\end{example} 

One could argue that voter $3$ should have chosen $\{b_1,b_2,c_1\}$ instead of $\{d_1,d_2,d_3\}$ at the end, as these candidates are approved by supersets of voters. In fact, the simple structure of candidate types for $n=3$ allows us to give another explicit polynomial-time algorithm (\Cref{alg:coren=3}) for finding a committee in the core that can be interpreted as MES for cost utilities where ties are broken in favor of candidates with larger approval sets.
In the algorithm, voters fund only approved candidates, costs are split as equally as possible, and candidates are selected according to the minimum possible maximum amount that an approver must invest, as in MES.

\begin{algorithm}[t]\LinesNumbered
	\caption{Compute core committee for $n=3$}\label{alg:coren=3}
	\KwIn{Instance $(c,k,b)$ with three voters and $b_1 \ge b_2 \ge b_3$.}
	\KwOut{Integral committee $x$ of size $k$.}
	$C \leftarrow c$ \tcp{vector of candidate types and their respective numbers} 
	$B \leftarrow b$ \tcp{vector of individual budgets}
	$x \leftarrow (0,0,\dots,0)$ \tcp{vector of length $7$ to encode the (initially empty) committee}
	
	\While{$\mathrm{sum}(x) < k \textnormal{ and } C_{123}-x_{123}>0$}{
		$x_{123} \leftarrow x_{123}+1$\;
		$\begin{aligned}
			B \leftarrow B -\biggl(&1 -\min\left\{B_2,\frac{1-\min \{B_3,1/3\}}{2}\right\} - \min\{B_3,1/3\},\\ &\min \left\{B_2,\frac{1-\min\{B_3,1/3\}}{2}\right\}, \min\{B_3,1/3\}\biggr)
		\end{aligned}$\;
	}
	
	\While{$B_i \ge 1$ \textnormal{for all} $i \in N$ \textnormal{and} $C_{ij}-x_{ij}>0$ \textnormal{for all pairs of voters} $(i,j)$}{
		$x_{ij} \leftarrow x_{ij}+1$ for all pairs $(i,j)$\;
		$B_i \leftarrow B_i-1$ for all $i \in N$\;
		}
	
	\While{$B_i+B_{j}\ge 1 \textnormal{ and } C_{ij}-x_{ij}>0 \textnormal{ for some pair $(i,j)$ of voters}$}{
		Let $(i,j)$ be a pair for which $B_i+B_j$ is maximal among pairs for which $\min\{B_i,B_{j}\}$ is maximal among all pairs with $B_i+B_{j}\ge 1$, $B_{i} \ge B_{j}$, and $C_{ij}-x_{ij}>0$\;
		$x_{ij} \leftarrow x_{ij}+1$\;
		$B_{i} \leftarrow B_{i}-1+\min\{B_{j},1/2\}$\;
		$B_{j} \leftarrow B_{j}-\min\{B_{j},1/2\}$\;
	}
	
	\For{$i \in N$}{
		$x_i \leftarrow \min\{\lfloor B_i \rfloor,C_i\}$\;
		$B_i \leftarrow B_i-x_i$\;
	}
	
	\While{$\mathrm{sum}(x)<k$ \textnormal{and} $C_{ij}-x_{ij}>0$ \textnormal{for some pair} $(i,j)$ \textnormal{of voters}}{
		$x_{ij}\leftarrow x_{ij}+1$\;
	}
	\While{$\mathrm{sum}(x)<k$}{
		$x_{i}\leftarrow x_{i}+1$ for some $i \in N$ with $C_{i}-x_{i}>0$\;
		
	}
	\Return $x$\;
\end{algorithm}

First, we include all candidates that are unanimously approved and distribute costs as equally as possible among voters.
Second, while each voter still has budget at least $1$, triples of candidates are included, with each candidate approved by a different pair of voters. This step is essential for ensuring that no pair of voters has a beneficial deviation when considering the core.
Third, candidates that are approved by pairs of voters (and can be funded with their combined budgets) are chosen, with pairs having the largest minimal individual budgets prioritized.
Fourth, each voter funds candidates with her remaining budget that only she approves. 
Finally, the committee is filled up to size $k$, with candidates having larger approval sets chosen first.

\begin{theorem}\label{thm:algn=3}
	\Cref{alg:coren=3} always returns a Pareto-optimal core committee of size $k$ in polynomial time.
\end{theorem}

We first establish two auxiliary lemmas for \Cref{alg:coren=3}.

\begin{lemma}\label{lem:non-negbudgets}
	Until line 18, $\sum_{R \in \mathcal{R}_c}x_R+\sum_{i \in N}B_i=k$ and $B_i \ge 0$ at the end of each loop.
\end{lemma}

\begin{proof}
	After adding candidates to $x$ in one of the first three while loops or the for loop, the total budget is decreased by $1$. So it remains to show that no individual budget becomes negative in the process. In the first while loop, note that $B_1 \ge B_2 \ge B_3$ as $b_1 \ge b_2 \ge b_3$ and as soon as a voter $i$ pays more than voter $i+1$, voter $i+1$'s individual budget must drop to $0$ in that iteration. Moreover, $\sum_{R \in \mathcal{R}_c} x_R <k$ implies $\sum_{i \in N}B_i \ge 1$, so $B_1$ cannot become negative.
	The second while loop only applies while all individual budgets are still at least $1$ and each iteration decreases each of them by $1$. 
	In the third while loop, non-negativity of the individual budgets follows from $B_{i}+B_{j}\ge 1$ and $B_{i}\ge B_{j}$. Finally, $x_i \le \lfloor B_i \rfloor \le B_i$ in line 16.
\end{proof}

\begin{lemma}\label{lem:algPO}
	Throughout the algorithm, $x_R>0$ implies $x_{R'}=c_{R'}$ for all $R' \supsetneq R$.
\end{lemma}

\begin{proof}
	By construction, we leave the first while loop when $x_{123}=\min\{c_{123},k\}$. If $x_{123}=k$, no further candidates are added to $x$ in the following; in particular, neither the second nor the third while loop is entered as $B_i=0$ for all voters $i$ by \Cref{lem:non-negbudgets}. Otherwise, $x_{123}=c_{123}$ and we can focus on $R'$ corresponding to a pair. Furthermore, if $x_i>0$ for some voter $i$, either $B_i \ge 1$ in line $15$ or $x_i$ is increased in line $21$. For both cases, the respective previous while loop implies $x_{ij}=c_{ij}$ for all $j \neq i$. 
\end{proof}

\begin{proof}[Proof of \Cref{thm:algn=3}]
	First, note that $x_R \le c_R$ for all $R \in \mathcal{R}$ throughout the process.
	By \Cref{lem:non-negbudgets}, $\sum_{R \in \mathcal{R}_c}x_R \leq k$ before line $18$. As the subsequent while loops increase the committee size only by $1$ at a time (and always do so while $\sum_{R \in \mathcal{R}_c}x_R < k$, as $\sum_{R \in \mathcal{R}_c}c_R \ge k$ by assumption), the algorithm returns a committee of size $k$ in the end.
	
	To prove Pareto optimality of $x$, assume for contradiction that there exists a committee $y$ of size $k$ with $u_i(y)\ge u_i(x)$ for all $i \in N$ with strict inequality for at least one voter. Thus, $x_{123}=c_{123}<k$. In particular, the sum over utilities can only increase from $x$ to $y$ if $x_j>y_j$ for at least one voter $j$. By \Cref{lem:algPO}, $x$ also has to contain all candidates that are approved by voter $j$ and at least one other voter. As $x_j>y_j$, this contradicts $u_j(y)\ge u_j(x)$. Thus, $x$ has to be Pareto optimal.
	
	Moving to the proof that $x$ is in the core, Pareto optimality of $x$ also implies that the group of all voters cannot beneficially deviate from $x$. 
	
	Denote the values of the variables $B_i$ at the end of the algorithm by $\bar{B}_i$.
	Next, we show that no single voter $i$ wants to deviate from $x$. If $\bar{B}_i \ge 1$, all candidates that are approved by voter $i$ are already in $x$ and her utility cannot be further increased. If $\bar{B}_i<1$, voter $i$ on her own cannot afford more than $u_i(x)$ candidates as she has spent $b_i-\bar{B}_i$ on approved candidates by the definition of \Cref{alg:coren=3}.
	
	Finally, consider a pair of voters $\{i,j\}$ and assume for contradiction that there exists a committee $y$ with $\sum_{R \in \mathcal{R}_c}y_R \le b_i+b_j$, $u_i(y)>u_i(x)$, and $u_j(y)>u_j(x)$. By the same argument as before, $\bar{B}_i,\bar{B}_j<1$ must hold. If $x_{ij}=c_{ij}$, we have $x_{123}=c_{123}$ by \Cref{lem:algPO} and can assume $y_{123}=c_{123}$. Thus, $u_i(y)>u_i(x)$ implies $y_{i\ell}+y_{i}>x_{i\ell}+x_{i}$, and $u_j(y)>u_j(x)$ implies $y_{j\ell}+y_{j}>x_{j\ell}+x_{j}$, with voter $\ell$ not being part of the pair $(i,j)$. All in all, we get a contradiction:
	\[
	b_i+b_j \ge \sum_{R \in \mathcal{R}_c}y_R \ge 2+\sum_{R \in \mathcal{R}_c}x_R -x_{\ell} > 2+k-(b_{\ell}+2)=k-b_\ell=b_{i}+b_{j}
	\]
	
	where the second inequality follows from $y_{i\ell}+y_{i}>x_{i\ell}+x_{i}$ and $y_{j\ell}+y_{j}>x_{j\ell}+x_{j}$ and the third inequality uses $x_\ell < b_{\ell}+2$ which follows from $\bar{B}_i,\bar{B}_j<1$, i.e., only less than $2$ of $b_i+b_j$ can be allocated to candidates none of them approves; note that $b_i+b_j+b_\ell=k=\sum_{R \in \mathcal{R}_c}x_R$. 
	
	Therefore, $x_{ij}<c_{ij}$ must hold and with that $x_i=x_j=0$ by \Cref{lem:algPO} and $\bar{B}_i+\bar{B}_j<1$, implying that $B_i+B_j<1$ already holds after line $14$. We can further assume $c_{ij}=k$ w.l.o.g. as potentially increasing $c_{ij}$ only yields stronger deviation incentives for $\{i,j\}$ but has no influence on the outcome of the algorithm as $x_{ij}<c_{ij}$. Thus, w.l.o.g., the optimal deviation $y$ is given by $y_{ij}=\lfloor b_i+b_j\rfloor$, and $u_i(x)<\lfloor b_i+b_j\rfloor$ as well as $u_j(x)<\lfloor b_i+b_j\rfloor$ have to hold for $y$ to be a beneficial deviation from $x$.
	
	We distinguish three cases depending on after which while loop $B_i+B_j$ becomes smaller than $1$ for the first time. Denote the committee $x$ after the respective while loop by $\tilde{x}$ and the remaining budget vector by $\tilde{B}$.
	If that happens after the first while loop, $\tilde{x}_{123}>b_i+b_j-1\ge \lfloor b_i+b_j \rfloor-1$ and $\tilde{x}_{123}\ge \lfloor b_i+b_j \rfloor$ as $\tilde{x}_{123}$ is an integer. This contradicts $u_i(\tilde{x}) \le u_i(x)<\lfloor b_i+b_j\rfloor$.
	Otherwise, if that happens after the second while loop, we must have entered it. Consequently, the costs of all purchased candidates in $\tilde{x}$ are distributed uniformly among their approvers. In particular, voters $i$ and $j$ have to pay at most $\frac{1}{2}$ each for every candidate they approve. This yields
	\begin{align*}
	u_i(x)+u_j(x)\ge u_i(\tilde{x})+u_j(\tilde{x}) \ge 2(b_i-\tilde{B}_i+b_j-\tilde{B}_j)>2\lfloor b_i+b_j \rfloor-2
	\end{align*}
	where the last inequality follows from $\tilde{B}_i+\tilde{B}_j<1$. This implies that either $i$ or $j$ have to receive utility at least $\lfloor b_i+b_j\rfloor$ under $x$, a contradiction.
	
	So we are left with the case in which this must happen after the third while loop, meaning that $B_i+B_j \ge 1$ before entering it. Thus, we either have $\min\{C_{i\ell}-x_{i\ell},C_{j\ell}-x_{j\ell}\}=0$, or one of the voters has a budget of less than $1$ before entering the third while loop, as otherwise we would have spent another iteration in the second while loop.
	Denote by $B'_i$ and $B'_j$ the shares of voter $i$'s and $j$'s individual budgets that are used to buy candidates of type $\{i,\ell\}$ and $\{j,\ell\}$ during the execution of that loop. In their deviation $y$, they can afford to buy $\lfloor B'_i+B'_j+\tilde{B}_i+\tilde{B}_j\rfloor$ candidates that they both approve instead. For both voters, that number has to exceed the utility received from candidates purchased in the third while loop that only one of them approves; denote their numbers by $x'_{i\ell}$ and $x'_{j\ell}$, respectively. Otherwise, $y$ would not be a beneficial deviation.
	We claim that $B'_i,B'_j>0$ has to hold.
	Otherwise, assume w.l.o.g that $B'_j=0$. Then,
	\[
	x'_{i\ell}\ge \lceil B'_i \rceil \ge \lfloor B'_i+B'_j+\tilde{B}_i+\tilde{B}_j\rfloor
	\]
	as $B'_j=0$ and $\tilde{B}_i+\tilde{B}_j<1$, contradicting that $y$ is a beneficial deviation for voter $i$.
	
	Therefore, w.l.o.g. $B'_i\ge B'_j>0$ and one of the voters needs to have budget less than $1$ before entering the third while loop. 
	
	Next, we claim that voter $i$ has to pay more than $\frac{1}{2}$ for at least one candidate not approved by voter $j$, as otherwise
	\[
	\lfloor B'_i+B'_j+\tilde{B}_i+\tilde{B}_j\rfloor 
	\le \lfloor 2B'_i+\tilde{B}_i+\tilde{B}_j\rfloor
	\le \lfloor x'_{i\ell}+\tilde{B}_i+\tilde{B}_j\rfloor
	=x'_{i\ell}
	\]
	where the last inequality follows from $x'_{i\ell}$ being an integer and $\tilde{B}_i+\tilde{B}_j<1$.
	
	This implies $\tilde{B}_{\ell}=0$ and also $\tilde{B}_i=\tilde{B}_j=0$ as $\tilde{B}_i+\tilde{B}_j+\tilde{B}_\ell$ is an integer.
	 
	For voter $i$, this means that 
	\[
	\lfloor B'_i+B'_j \rfloor >x'_{i\ell} \ge \lceil B'_i \rceil
	\]
	has to hold which implies $B'_j \ge 1$ and also $B'_i \ge 1$ by assumption.
	
	Hence, $\min\{C_{i\ell}-x_{i\ell},C_{j\ell}-x_{j\ell}\}=0$ before the third while loop cannot hold and one of the voters needs to have budget less than $1$ at that time. As $B'_i,B'_j\ge 1$, this must be voter $\ell$. However, at the point when the first candidate not approved by both $i$ and $j$ is purchased, for example one of type $\{i,\ell\}$, either voter $i$ has budget less than $1$ left, or voter $\ell$ has at least as much budget left as voter $j$ by construction of the algorithm since more candidates of type $\{i,j\}$ are always available. In the first case, $B'_i \ge 1$ cannot hold and in the second case, $B'_j \ge 1$ cannot hold, resulting in a contradiction. 
	
	This concludes the proof that no pair of voters has a beneficial deviation from $x$ and thus, $x$ is in the core.
	
	It is straightforward to see that \Cref{alg:coren=3} runs in time polynomial in $m$ and $k$.
\end{proof}

Unfortunately, it is impossible to define a tie-breaking method for MES that works for arbitrary numbers of voters, as shown by the following example, which is based on the counterexample by \citet[Proposition 4]{PeSk20a} but uses substantially fewer voters. Note that MES exhausts the entire budget on that instance, showing that the counterexample works independently of the chosen completion rule that is sometimes required to increase a committee returned by MES to size $k$.\\

\begin{figure}[!ht]
	\centering
	\begin{tikzpicture}[yscale=0.6,xscale=0.9,voter/.style={anchor=south}]	

	\foreach \i in {1,...,9}
		\node[voter] at (\i-0.5, -1) {$\i$};

    \draw[fill=orange!40] (0,0) rectangle (6,1);
    \node at (3,0.5) {$18 \times a$};

    \draw[fill=violet!40] (6,0) rectangle (8,1);
    \node at (10.1,0.4) {$3 \times b$};

    \draw[fill=violet!40] (6,1) rectangle (7,2);
    \draw[fill=violet!40] (8,1) rectangle (9,2);
    \node at (10.1,1.4) {$3 \times c$};

    \draw[fill=violet!40] (7,2) rectangle (9,3);
    \node at (10.1,2.4) {$3 \times d$};

    \draw[fill=teal!40] (0,3) rectangle (4,4);
    \draw[fill=teal!40] (6,3) rectangle (7,4);
    \node at (3,3.5) {$7 \times e$};

    \draw[fill=teal!40] (0,4) rectangle (4,5);
    \draw[fill=teal!40] (7,4) rectangle (8,5);
    \node at (3,4.5) {$7 \times f$};

    \draw[fill=teal!40] (0,5) rectangle (4,6);
    \draw[fill=teal!40] (8,5) rectangle (9,6);
    \node at (3,5.5) {$7 \times g$};

\end{tikzpicture}
\caption{Example for MES violating the core independently of tie-breaking.}
\label{fig:mes_violates_core_unique}
\end{figure}

\begin{example}
	\label{ex:mes_violates_core_unique}
	Let $n=9$ and $k=27$, where each voter has budget $b_i=3$. 
	The approval profile is shown in \Cref{fig:mes_violates_core_unique}. In words, $\mathcal{R}_c=\{123456,78,79,89,12347,12348,12349\}$ and supplies are $c_a=18$, $c_b=c_c=c_d=3$, and $c_e=c_f=c_g=7$.
	MES first selects the $18$ candidates of type $a$, which exhausts the budgets of these voters. 
	The voter set $\{7,8,9\}$ is then left without any approved candidate. 
	For these voters, candidates approved by two agents are more affordable than candidates approved by only one agent, and MES therefore selects all nine candidates of types $b$, $c$, and $d$. 
	This exhausts the remaining budgets and thus, MES yields the committee
	\[
		W=18\times \{a\}\;\cup\;3\times\{b,c,d\}.
	\]
	Now consider the coalition $S=N\setminus\{5,6\}$, which has size $\lvert S \rvert=7$. Its total budget is therefore $\frac{\lvert S\rvert}{n}\cdot k=21$, so $S$ can afford exactly $21$ candidates. 
	We claim that the committee $T=7\times\{e,f,g\}$ that can be afforded by $S$ forms a valid objection for $S$.
	Indeed, the utilities of Voters $1$-$4$ are increased from $18$ to $21$ and the utilities of Voters $7$-$9$ are increased from $6$ to $7$ when moving from $W$ to $T$.
	Therefore, the committee $W$ returned by MES is not in the core. \qed
\end{example}
 
\end{document}